\documentclass[reqno,12pt] {article}
\usepackage{amsmath, amsfonts, amssymb,amsthm,hyperref,floatrow}
\usepackage{graphicx,color,dsfont,epsfig,caption,subcaption,wrapfig}
\usepackage{xcolor}
\usepackage{tikz}

\hypersetup{
    linktoc=page,
    linkcolor=red,          % color of internal links
    citecolor=blue,        % color of links to bibliography
    filecolor=blue,      % color of file links
    urlcolor=cyan,
    colorlinks=true           % color of external links
}

\newcommand{\R}{\mathbb{R}}

\newcommand{\N}{\mathbb{N}}
\newcommand{\E}{{\mathbb E}}

\renewcommand{\P}{\mathbb{P}}

\newcommand{\TT}{\mathsf{T}}

\newcommand{\wt}{\widetilde}
\RequirePackage{yhmath} 
 \renewcommand\widering[1]{\ring{#1}}
\newtheorem{thm}{Theorem}[section]
\newtheorem{prop}{Proposition}[section]

\newtheorem{cor}{Corollary}[section]
\newtheorem{lem}{Lemma}[section]

\newtheorem{rem}{Remark}[section]
\theoremstyle{definition}

\numberwithin{equation}{section}
 \newcommand{\Dim}{\mathrm{Dim\,}}  
  \newcommand{\dom}{\mathrm{dom\,}}  
 \definecolor{remi}{rgb}{1,0,0}

\begin{document}
\baselineskip 15pt

\title
%[Hausdorff and packing spectra, large deviations, and free energy ]
{Hausdorff and packing spectra, large deviations, and free energy for branching random walks  in $\R^d$}

\author{Najmeddine Attia\footnote{INRIA Paris-Rocquencourt, Domaine de Voluceau, 78153 Le Chesnay cedex, France; najmeddine.attia@inria.fr}, Julien Barral \footnote{LAGA (UMR 7539), D\'epartement de Math\'ematiques, Institut Galil\'ee, Universit\'e
 Paris 13, Paris Sorbonne Cit\'e, 99 avenue Jean-Baptiste Cl\'ement , 93430  Villetaneuse, France; barral@math.univ-paris13.fr}\footnote{The authors were  supported by the French National Research Agency  Project ``DMASC''.
 They thank the referees for their constructive suggestions, in particular the addition of Theorem~1.4 as a consequence of the main results of the preliminary version. They also thank Dr Xiong Jin for his help in the figures elaboration.}}

%\author{Najmeddine Attia}
%\address[N. Attia]{INRIA Paris-Rocquencourt, Domaine de Voluceau, 78153 Le Chesnay cedex, France, and LAGA (UMR 7539), D\'epartement de Math\'ematiques, Institut Galil\'ee, Universit\'e
% Paris 13, Paris Sorbonne Cit\'e, 99 avenue Jean-Baptiste Cl\'ement , 93430  Villetaneuse, France}
%\email{najmeddine.attia@inria.fr}
%\author{Julien Barral}
%\address[J. Barral]{LAGA (UMR 7539), D\'epartement de Math\'ematiques, Institut Galil\'ee, Universit\'e
% Paris 13,  Paris Sorbonne Cit\'e, 99 avenue Jean-Baptiste Cl\'ement , 93430  Villetaneuse, France}
%\email{barral@math.univ-paris13.fr}

\maketitle
\begin{abstract}
 Consider an $\R^d$-valued branching random walk (BRW) on a supercritical Galton Watson tree. Without any assumption on the distribution of this BRW we compute, almost surely and simultaneously, the Hausdorff and packing dimensions of the level sets $E(K)$ of infinite branches in the boundary of the tree  (endowed with its standard metric) along which  the averages of the BRW have a given closed connected set of limit points $K$. This goes beyond multifractal analysis, which only considers  those level sets when $K$ ranges in the set of singletons $\{\alpha\}$, $\alpha\in\R^d$. We also give a $0$-$\infty$ law for the Hausdorff and packing measures of the level sets $E(\{\alpha\})$, and compute the free energy of the associated logarithmically correlated random energy model in full generality.  Moreover, our results  complete the previous works on multifractal analysis by including the levels $\alpha$ which do not belong to the range of the gradient of the free energy. This covers in particular a situation until now badly understood, namely the case where a first order phase transition occurs. As a consequence of our study, we can also describe the whole singularity spectrum of Mandelbrot measures, as well as the associated free energy function (or $L^q$-spectrum),  when a first order phase transition occurs. 
  \end{abstract}
\newpage
\tableofcontents

\normalsize
%%%%%%%%%%%%%%%%%%%%%%%%%%%%%%%%%%%%5 Introduction %%%%%%%%%%%%%%%%%%%
%%%%%%%%%%%%%%%%%%%%%%%%%%%%%%%%%%%%%               %%%%%%%%%%%%%%%%%%
\section{Introduction }

\subsection{Foreword about large deviations, multifractals, and phase transitions in directed polymers on the dyadic tree}\label{foreword}
Let $X_u$, $u\in \TT=\bigcup_{n\ge 1}\{0,1\}^{n}$, be a collection of independent copies of a non constant real valued random variable $X$, indexed by the nodes of the dyadic tree. The  branching random walk $S_nX(u_1\cdots u_n)=\sum_{k=1}^n X_{u_1\cdots u_k}$  is used by Derrida and Spohn \cite{DerSpo} to define the logarithmically correlated random energy model of the directed polymer on $\TT$  which assigns energy $E(u_1\cdots u_n)= S_nX(u_1\cdots u_n)$ to bond $u_1\cdots u_n$.  The associated partition function and free energy  at level $n$ are then defined as 
$$
Z_n(q)= \sum_{u\in\{0,1\}^n} e^{q S_nX(u)} \quad\text{and}\quad n^{-1}\log Z_n(q)\quad (q\in\R)
$$
respectively, and the associated family of Gibbs measures on $\partial \TT=\{0,1\}^{\N_+}$ is given by 
$$
\nu_{q,n}([u])=\frac{e^{q S_nX(u)}}{Z_n(q)}\quad (u\in\{0,1\}^n),
$$
the mass being uniformly distributed on the cylinder $[u]$ with respect to the uniform measure on $\partial\TT$, denoted by $\lambda$. 

The  asymptotic behavior of $n^{-1}\log Z_n(q)$ as $n\to\infty$ is closely related to the large deviations properties of the  random walk on each individual infinite branch, as well as  the multifractal behavior of $S_nX$ on $\partial \TT=\{0,1\}^{\N_+}$, which can be viewed as a geometric realization of  these properties.  This section describes what is well known, or as a motivation for our work what is not completely understood, about these relations, and their connections with possible types of phase transitions in this context. We will also describe the situation regarding the behavior of the partition functions and  Gibbs measures. 

Consider an individual infinite branch $u_1\cdots u_n\cdots$ of $\partial\TT$. When $\E(X)$ is defined, the strong law of large numbers yields $\lim_{n\to\infty}n^{-1}S_nX(u_1\cdots u_n)=\E(X)$. Then,  an application of  Fubini's theorem yields almost surely,  for $\lambda$-almost every infinite branch $u_1\cdots u_n\cdots$ the behavior   $\lim_{n\to\infty} n^{-1}S_nX(u_1\cdots u_n)=\E(X)$. On the  other hand, when the logarithmic generating function $L(q)=\log \E(e^{qX})$ of $X$ is finite  over some non trivial interval around 0, Cramer's theorem (and its extension by Bahadur and Zabell, see \cite{De-Zei}) ensures that for all $\alpha\in\R$, over any fixed infinite branch $u_1\cdots u_n\cdots$ of $\partial \TT$, one has
\begin{eqnarray*}
r(\alpha):=\inf_{q\in\R} L(q)-\alpha q&=&\lim_{\epsilon\to 0}\liminf_{n\to\infty} n^{-1}\log \P( |\alpha-n^{-1}\sum_{k=1}^n X_{u_1\cdots u_k}|\le \epsilon)\\
&=&\lim_{\epsilon\to 0}\limsup_{n\to\infty} n^{-1}\log \P( |\alpha-n^{-1}\sum_{k=1}^n X_{u_1\cdots u_k}|\le \epsilon).
\end{eqnarray*}
 This led B. Mandelbrot to anticipate \cite{BM3}, due to the rate of growth of the trees $\{0,1\}^{n}$, $n\ge 1$, that, almost surely, for all  $\alpha\in\R\setminus \{\E(X)\}$ such that $\log(2)+r(\alpha)\ge 0$, one must have the following ``logarithmic frequency'' for the branches $u=u_1\cdots u_n$ of $\{0,1\}^n$ over which $S_n(u)\approx n \alpha$:  setting 
 $$
 f(n,\alpha,\epsilon)=n^{-1}\log_2 \#\big \{u\in \{0,1\}^n: n^{-1} S_nX(u) \in [\alpha-\epsilon,\alpha+\epsilon])\big \},
 $$ 
\begin{equation}\label{LDB}
\lim_{\epsilon\to 0^+}\liminf_{n\to\infty} f(n,\alpha,\epsilon)=\lim_{\epsilon\to 0^+}\limsup_{n\to\infty} f(n,\alpha,\epsilon)=1+\frac{r(\alpha)}{\log(2)}.
\end{equation} 
Moreover, this large deviation property must have a geometric counterpart stemming from the existence of infinite branches in $\partial \TT$ over which one observes the different law of large numbers $\lim_{n\to\infty} n^{-1}\sum_{k=1}^n X_{u_1\cdots u_n}=\alpha$. Quantifying the size of the sets of such infinite branches, i.e.
$$
E(\alpha)=\Big \{u_1\cdots u_n\cdots\in\partial\TT: \lim_{n\to\infty} n^{-1}S_nX(u_1\cdots u_n)=\alpha\Big \}
$$
leads to the multifractal analysis of the branching random walk $S_nX$: if $\partial \TT $ is endowed with its standard ultrametric distance one can  prove, by using suitable Mandelbrot measures (see the definition below) carried by the sets $E(\alpha)$, that almost surely, for all the $\alpha$ taking the form $L'(q_\alpha)$ for some $q_\alpha\in\R$ and such that $\log(2)+r(\alpha)\ge 0$, 
one has
\begin{equation}\label{formula}
\dim E(\alpha)=1+\frac{r(\alpha)}{\log(2)},
\end{equation}
where $\dim$ stands for the Hausdorff dimension (see Section~\ref{classic} below for more details and references), and as a consequence \eqref{LDB} holds for $\alpha$.  One also has $\dim E(\alpha)=\inf\{P(q)-q\alpha:q\in \R\}/\log(2)$, where $P(q)$ is the free energy 
$$
P(q)=\limsup_{n\to\infty} n^{-1}\log Z_n(q),
$$
hence the so-called multifractal formalism holds at $\alpha$ (see \cite{FrPa,Halsey,Collet,BMP,Ben,Ols,Pey4} for introductions to multifractal formalisms). 

\medskip

What is known about the relations between $P(q)$, $L(q)$, and the convergence of $n^{-1}\log Z_n$ to $P$ can be precized as follows, see \cite{HoWa,CoKo1992,Franchi,Mol,OW} which more generally deal with the free energy associated with Mandelbrot measures (by symmetry, we only discuss the case $q\ge 0$). Either of the following three situations hold (see Figure~\ref{Fig2}):

\medskip

(1) There is no phase transition: $L$ is finite and twice differentiable over $\R_+$ and  $\log (2)+L(q)-qL'(q)> 0$ for all $q\in\R_+$, in which case $P(q)=\lim_{n\to\infty} n^{-1}\log Z_n(q)=\log (2) + L(q)$ over $\R_+$ almost surely.  

\medskip

(2)  There exists $q_c>0$ such that $\log (2)+L(q_c)-q_cL'(q_c)=0$. Then  $P(q)=\lim_{n\to\infty} n^{-1}\log Z_n(q)$ with $P(q)=\log (2) + L(q)$ over $[0,q_c]$ and $P(q)=q L'(q_c)=\frac{\log (2) + L(q_c)}{q_c}q$ over $[q_c,\infty)$ almost surely; consequently $P$ is differentiable, and twice differentiable everywhere except at $q_c$: There is a second order phase transition.  

\medskip

(3)  $q_c=\max\{q>0: L(q)<\infty\}<\infty$ and $\log (2)+L(q_c)-q_cL'(q_c)>0$; then $P(q)=\lim_{n\to\infty} n^{-1}\log Z_n(q)=\log (2) + L(q)$ over $[0,q_c]$ almost surely, and only a partial result is  obtained for $q>q_c$ in \cite{Mol}, where the almost sure convergence of $n^{-1}\log Z_n(q)$ to $\frac{\log (2) + L(q_c)}{q_c}q$ over $(q_c,\infty)$ is shown only along a deterministic subsequence. Moreover, the limit of $n^{-1}\log Z_n$ along some subsequence is twice differentiable everywhere on $\R_+\setminus \{q_c\}$ and not differentiable at $q_c$; this corresponds to a first order phase transition. We will prove that the almost sure simple convergence of $n^{-1}\log Z_n$ to the previous limit holds. 

\medskip

In cases (1) and (2)  the set of those $\alpha$ such that $E(\alpha)\neq\emptyset$ equals the closure of the range of $P'$ and \eqref{LDB} and \eqref{formula} hold for each such $\alpha$. In case of a first order phase transition nothing is known until now about the sets $E(\alpha)$ and the validity of \eqref{LDB} for the points $\alpha$ in $( L'(q_c^-),\frac{\log (2) + L(q_c)}{q_c}]$,  i.e. those $\alpha>  L'(q_c^-)$ for which $1+\frac{r(\alpha)}{\log(2)}\ge 0$ and one can nevertheless conjecture that both  \eqref{LDB} and \eqref{formula} still hold, though $\alpha$ is not of the form $L'(q)$. This will be also one of the achievements of the paper. 

\medskip

Now let us briefly recall that the asymptotic behavior of the partition function $Z_n(q)$ and Gibbs measure $\nu_{q,n}$ are now very well understood in cases (1) and (2) (see  \cite{DerSpo,DL,B2,webb11,aishi11,JW,ma11,BRV,BKNSW}): If no phase transition occurs, then Mandelbrot martingales theory shows that almost surely, all $q\in [0,\infty)$, $Z_n(q)/\E(Z_n(q))$ converges to a positive limit, namely the total mass of the Mandelbrot measure \cite{BM2,BM3} associated with $qS_nX - nL(q)$, i.e.  the weak limit $\mu_q$ of the measure-valued  martingale $\mu_{q,n}$ which assigns uniformly with respect to $\lambda$ the mass 
$$
\mu_{q,n}([u])= 2^{-n}e^{qS_nX(u)-nL(q)}
$$ 
to each cylinder $[u]$ of generation $n$. Also, the Gibbs measures $\nu_{q,n}$ converge weakly to $\mu_q/\|\mu_q\|$.  In presence of a the second order phase transition the so called freezing phenomenon at temperature $q_c^{-1}$ occurs:  after replacing $X$ by $q_c X- \log \E(e^{q_cX})-\log (2)$ to normalize the situation to $q_c=1$ and $\log(2)+L(1)=L'(1)=0$, one has: for $q\in [0,1)$, the same properties  as without phase transition hold;  at $q=1$, under weak assumptions, $\sqrt{n} Z_n(1)$ converges in probability to a multiple of $\|\mu_1\|$, where $\mu_1$ is the continuous critical Mandelbrot measure built from $S_nX$, that is the weak limit of the derivative martingale $  -\left (\frac{\mathrm {d} \mu_{q,n}}{\mathrm{d}q}\right )_{|q=1}$; moreover,  the associated Gibbs measures $\nu_{1,n}$ weakly converge in probability to $\mu_1/\|\mu_1\|$.  For $q>1$, under weak assumptions, $n^{3q/2} Z_n(q)$ converges in law (this is related to the fact that $\sup\{S_nX(u)+\frac{3}{2}\log(n): u\in\{0,1\}^n\}$ converges in law) to a positive random variable. This random variable can be written $L_{1/q}(\|\mu_1\|)$, where $L_{1/q}$ is a stable L\'evy subordinator of index $1/q$ independent of $\mu_1$, and the Gibbs measure $\nu_{q,n}$ weakly converge in law to an atomic measure, namely $\mu_q/\|\mu_q\|$, where $\mu_q$ is an independently scattered  $1/q$-stable random measure with control measure given by $\mu_1$ (the assumption $L(q+\epsilon)<\infty$ is used in \cite{BRV} in which the study of the fixed points of the smoothing transformation achieved in \cite{DL} is used, but it is not necessary anymore due to \cite[Theorem 6.1]{AlMe12}).
\medskip

One can wonder if in the case of a first order phase transition one has similarly a freezing phenomenon. After using the same change of variable $X:=q_c X- \log \E(e^{q_cX})-\log (2)$, for $q\in [0,1]$ it is easily seen that  the situation is similar to that without phase transition: Gibbs measures almost surely weakly converge to normalized continuous Mandelbrot measures. For $q>1$, it is natural to ask whether, under suitable assumptions, one has  properties such as: (1) there exists a (universal?) non negative  sequence $\epsilon_n$ converging to 0 such that $(a)$ $e^{q n \epsilon_n}Z_n(q)$ converges in law to a positive variable  and  $(b)$ $\sup\{S_nX(u)+n \epsilon_n: u\in\{0,1\}^n\}$ converges in law to some random variable (2) the Gibbs measures $\nu_{q,n}$ converge weakly in law (to a limit formally the same as in the second order phase transition case, except that $\mu_1$ would be the non degenerate Mandelbrot measure obtained as the limit of $\mu_{1,n}$?).  We will not further discuss these questions here.

\medskip

More generally, all the previous considerations and properties hold for branching random walks on supercritical Galton-Watson trees and their boundary. It is natural to also consider $\R^d$-valued branching random walks on such a tree (i.e. $d$ real valued branching random walks built simultaneously on the tree). This paper deals with the computation of the free energy in this case, as well as with a more general geometric question than multifractal analysis: the boundary of the tree  is partitioned into the sets of infinite branches along which  the averages of the walk have a given set of limit points in $\R^d$. We will measure, through their Hausdorff and packing dimensions or measures, the sizes of these sets when the boundary of the tree is endowed with the standard ultrametric distance.  Our approach is new and makes it possible to work without any restriction on the branching random walk distribution. This solves questions left open for the multifractal analysis of the branching random walk averages or Mandelbrot measures, as well as for the free energy functions asymptotic behavior (now  phase transitions can occur along submanifolds of dimension $d-1$ when $d\ge 2$).  In particular we can fully describe the cases where first order phase transitions occur. This yields new multifractal behaviors associated with statistically self-similar measures.

\subsection{Setting and statement of the main results}

Let $(N,X_1,X_2,\ldots)$ be a random vector taking values in $\N_+\times (\R^d)^{\N_+}$. Then consider $\{(N_{u0},X_{u1},X_{u2}),\ldots)\}_u$  a family of independent copies of $(N,X_1,X_2,\ldots)$ indexed by the finite sequences $u=u_1\cdots u_n$, $n\ge 0$, $u_i\in\N_+$ ($n=0$ corresponds to the empty sequence denoted $\emptyset$), and let $\TT$ be the Galton-Watson tree with defining elements $\{N_u\}$: we have $\emptyset \in \TT$ and, if $u\in \TT$ and $i\in\N_+$ then $ui$, the concatenation of $u$ and $i$, belongs to $\TT$ if and only if $1\le i\le N_u$. Similarly, for each $u\in \bigcup_{n\ge 0} \N_+^n$, denote by $\TT(u)$ the Galton-Watson tree rooted at $u$ and defined by the $\{N_{uv}\}$, $v\in \bigcup_{n\ge 0} \N_+^n$.

If $\wt \TT$ is another tree contained in $\TT$, for all $n\ge 1$ $\wt \TT_n$ stands for $\wt\TT\cap \N_+^n$. 

\medskip

For each $u\in  \bigcup_{n\ge 0} \N_+^n$, we denote by $|u|$ its length, i.e. the number of letters of $u$, and $[u]$ the cylinder $u\cdot{\N_+}^{\N_+}$, i.e. the set of those $t\in{\N_+}^{\N_+}$ such that $t_1t_2\cdots t_{|u|}=u$. If $t\in  {\N_+}^{\N_+}$, we set $|t|=\infty$, and the set of prefixes of $t$ consists of $\{\emptyset\}\cup\{t_1t_2\cdots t_n: \ n\ge 1\}\cup\{t\}$. Also we set $t_{|n}=t_1t_2\cdots t_n$ if $n\ge 1$ and $t_{|0}=\emptyset$. 

\medskip

The probability space over which the previous random variables are built  is denoted $(\Omega,\mathcal A, \mathbb P)$, and the expectation with respect to $\mathbb P$ is denoted $\mathbb E$.

\medskip

We assume that $\mathbb E(N)>1$ so that the Galton-Watson tree is supercritical. Without loss of generality, we also assume that the probability of extinction equals 0, so that $\mathbb P(N\ge 1)= 1$.

\medskip

The boundary of $\TT$ is the subset of $\N_+^{\N_+}$ defined as 
$$
\partial\TT=\bigcap_{n\ge 1}\bigcup_{u\in \TT_n}[u],
$$
where $\TT_n=\TT\cap \N_+^n$. 

$\N_+^{\N_+}$ is endowed with the standard ultrametric distance 
$$
d_1:(s,t)\mapsto \exp (-|s\land t|),
$$
where $s\land t$ stands for the longest common prefix of $s$ and $t$, and with the convention that $\exp(-\infty)=0$. Endowed with the induced distance the set $\partial \TT$ is almost surely compact.

Recalls about Hausdorff and packing dimensions for sets and measures are given in Section~\ref{HandP}.  The random set $\partial \TT$ has finite Hausdorff (and packing) dimension if and only if $\E(N)<\infty$.

\medskip

We are both interested in the asymptotic behavior of the branching random walk 
$$
S_nX(t)=\sum_{k=1}^nX_{t_1\cdots t_k} \quad (t\in \partial \TT)
$$
and the associated free energy function
\begin{equation}\label{Pn}
P_n(q)=n^{-1} \log \sum_{u\in\TT_n} \exp (\langle q|S_nX(u)\rangle)\quad (q\in\R^d).
\end{equation}
Since $S_nX(t)$ depends on $t_1\cdots t_n$ only, we also denote by $S_nX(u)$ the constant value of $S_nX(\cdot)$ over $[u]$ whenever $u\in \TT_n$. 

\medskip

The vector space $\R^d$ is endowed with the canonical scalar product and the associated Euclidean norm respectively denoted $\langle\cdot|\cdot\rangle$ and $\|\cdot \|$. 

\medskip

The multifractal analysis of $S_nX$ is a first natural consideration.  It consists in computing the Hausdorff and packing dimensions of the sets 
$$
E_X(\alpha)=\Big\{t\in \partial\TT: \lim_{n\to\infty} \frac{S_nX(t)}{n}=\alpha\Big \},\quad (\alpha\in\R^d).
$$ 

\medskip

Without loss of generality we assume the following property about $X$: 

\begin{equation}\label{pP}
 \text{$\not\exists\ (h,c)\in(\R^d\setminus\{0\})\times \R$,\ $\langle h|X_i\rangle =c \quad \forall\ 1\le i\le N$ a.\,s.}
\end{equation}
that is there is no strict affine subspace of $\R^d$ which contains all the vectors $X_i$, $1\le i\le N$, almost surely. Indeed, if \eqref{pP} fails, the $X_i$, $1\le i\le N$, belong to the same affine hyperplane so that using a translation we can reduce the study to the case of  $\R^{d-1}$  valued random variables (with the convention $\R^{0}=\{0\}$). 

When $E(N\log N)<\infty$ and $\mathbb{E}\big (\sum_{i=1}^N\|X_i\|\big )<\infty$,  considering the branching measure on $\partial \TT$ (see the measure $\mu_0$ below in \eqref{defmuq})  makes it possible to show, by an application of the strong law of large numbers, a generalization of the result mentioned above on the dyadic tree, namely  that $E_X(\alpha_0)$ is of full Hausdorff dimension in $\partial \TT$, where $\alpha_0=\frac{\mathbb{E}\big (\sum_{i=1}^NX_i\big )}{\E(N)}$ (see \cite{KP,LiuR}).  It is worth quantifying the existence of other infinite branches over which such a law of large numbers holds, with different values of $\alpha$, or the existence of infinite branches over which the averages of $S_nX(t)$ has a given set of limit points. We need new notations to discuss these questions. 

\medskip

The domain of a convex (resp. concave) function $f$ from $\R^d$ to $\R_+\cup\{-\infty,\infty\}$, i.e. the convex set of those $q\in\R^d$ such that $f(q)<\infty$ (resp. $>-\infty$), is denoted $\dom f$. Its relative interior and relative boundary are denoted $\mathrm{ri} (\dom f)$ and $\mathrm{rel}\, \partial (\dom f)$ respectively.   When $\dom f\neq \emptyset$ and $f$ takes only values different from $-\infty$ (resp. $\infty$), $f$ is said to be proper. Then, the closure of $f$, denoted $\mathrm{cl}(f)$, is the greatest (resp. smallest) lower (resp. upper) semi-continuous function dominated by (resp. dominating) $f$, and equals $\liminf_{q'\to q}f(q')$ (resp. $\limsup_{q'\to q}f(q')$) at each $q\in\R^d$. 

\medskip

For any function $f:\R^d\to\R\cup\{\infty\}$ and any $\alpha\in\R^d$ define the Legendre-Fenchel transform of $f$ at $\alpha$ by 
$$
f^*(\alpha)=\inf\{f(q) -\langle q|\alpha\rangle: q\in\R^d\} \in \R\cup \{-\infty,\infty\}
$$
(notice that $f^*$ is everywhere equal to $\infty$ if $f$ is, while it is everywhere $<\infty$ if $f$ is finite at some point). By construction $f^*$ is concave and upper semi-continuous, so that our definition, which is convenient to express the values of some Hausdorff and packing dimensions, does not follow the standard ``convex'' convention.  According to the duality results \cite[Theorem 12.2, Corollary 12.2.2]{Roc}, if $f$ is a closed convex function, then $f^*$ is a closed concave function and we have 
\begin{equation}\label{duality}
\begin{split}
f^*(\alpha)&=\inf\{f(q)-\langle q|\alpha\rangle: q\in  \mathrm{ri}(\dom f^*)\} \quad \forall \ \alpha\in\R^d,\\  f(q)&=\sup\{\langle q|\alpha\rangle +f^*(\alpha): \alpha\in \mathrm{ri}(\dom f^*)\} \quad \forall \ q\in\R^d
\end{split}.
\end{equation}

For $q\in\R^d$ let 
$$
S(q)=\sum_{i=1}^N \exp (\langle q|X_i\rangle).
$$
Then define the lower semi-continuous function
\begin{equation}\label{wtP}
\wt P:q\in\R^d\mapsto \displaystyle \log \mathbb{E}(S(q))\in \R\cup\{\infty\}.
\end{equation}
Due to \eqref{pP}, the function $\wt P$  is strictly convex on $\dom \wt P$ whenever this set is neither empty nor  a singleton (see Remark~\ref{sc}). We define the closed convex set
\begin{equation}\label{IX}
I=\{\alpha\in\R^d: \widetilde P^*(\alpha)\ge 0\}. 
\end{equation}

After the general properties mentioned above, the function $\wt P^*$ is concave and upper semi-continuous and we always have 
$$
\wt P(0)= \sup\{\wt P^*(\alpha): \alpha\in\R^d\}=\log (\E(N))\in (0,\infty],
$$
(this is trivial if $\dom \wt P=\emptyset$, and follows from \eqref{duality} otherwise). We will see that the set $I$ has a non-empty interior over which $\wt P^*$ is positive. This allows us to assume, without loss of generality, that $0\in \widering I$, by replacing if necessary the $X_i$ by $X_i-\alpha$ with  $\alpha\in \widering I$.

\medskip

We now describe the main ideas used in previous approach to  the multifractal analysis of $S_nX$ and the kind of results they lead to. Then we will state and comment our main results. 

\subsubsection{Description of previous approach to the multifractal analysis of $S_nX$}\label{classic}

In order to describe the nature of the standard approach to the multifractal analysis problem, let us assume momentarily that
\begin{equation}\label{finiteness0}
\forall\ q\in\R^d,\ \E(S(q))<\infty \text{ and } \exists\ \gamma>1,\ \E(S(q')^\gamma)<\infty \text{ in a neighborhood of $q$} .
\end{equation}
In particular, the logarithmic moment generating function  $\widetilde P$  is finite and strictly convex over $\R^d$. We will see that  the domain of those $\alpha$ for which $E_X(\alpha)\neq \emptyset $ in this case is  the convex compact set $I$.  

Let 
\begin{equation}\label{J}
J=\{q\in\R^d: \wt P(q)-\langle q|\nabla \wt P(q)\rangle >0\};
\end{equation}
notice that $\nabla \wt P(J)\subset I$.

\medskip

A natural approach to the Hausdorff and packing dimensions of $E_X(\alpha)$ consists in taking $\alpha\in I$ such that $\alpha=\nabla \widetilde P(q)$ for some $q\in\R^d$ and  $\widetilde P^*(\alpha)>0$ (we will see in the proof of Proposition~\ref{detI} that this is equivalent to taking $\alpha\in \widering I$). Then, one considers the associated Mandelbrot measure (\cite{BM2,BM3}) on $\partial\TT$ defined as 
\begin{equation}\label{defmuq}
\mu_{q}([u])= \exp( \langle q|S_nX(u) \rangle-n \widetilde P(q)) Z(q, u), \quad (u\in\TT_n),
\end{equation}
where 
$$
Z(q, u)=\lim_{p\to\infty}\sum_{v\in \TT_p(u)}  \exp( \langle q|(S_{n+p}X(u\cdot v)-S_{n}X(u))\rangle-p \widetilde P(q))
$$ 
and here and afterwards we simply denote $[u]\cap \partial \TT$ by $[u]$. Under our assumptions, the fact that $\widetilde P^*(\alpha)>0$ implies that this measure is almost surely positive \cite{KP,Biggins1,Lyons} (while $\mu_{q}=0$ almost surely if $\widetilde P^*(\alpha)\le 0$).  Moreover,  $\mu_{q}$ is carried by $E_X(\alpha)$ and is exact dimensional (see Section~\ref{HandP} for this notion) with  $\dim \mu_q=\Dim\mu_q=\wt P^*(\alpha)$ (see \cite{KP,K2} for the dimension of Mandelbrot measures), so that $\dim E_X(\alpha)\ge \widetilde P^*(\alpha)$. Then, a simple covering argument yields $\dim E_X(\alpha)=\widetilde P^*(\alpha)$. This approach holds for each $\alpha\in\widering I$ almost surely, and it has been followed in the case $d=1$ for the multifractal analysis of Mandelbrot measures, especially in \cite{Kahane1991',HoWa,Ol1,Falc,Mol,jul2,Biggins3}, where more general metric than $d_1$ are considered to get results for geometric realizations of these measures in $\R^n$.  Nevertheless, it is possible to strengthen the result to get, with probability 1,  $\dim E_X(\alpha)=\widetilde P^*(\alpha)$ for all $\alpha\in \widering I$. This is done in  \cite{B2} for the case $d=1$ on Galton Watson trees with bounded branching number, and in \cite{A} for the present context:
\medskip

\noindent  
{\bf Theorem }(\cite{A}){\bf .} \textit{Assume~\eqref{finiteness0} holds. With probability $1$, for all $  \alpha\in \widering I$, $\dim E_X (\alpha)=\wt P^*(\alpha)>0$; in particular, $E_X(\alpha) \neq \emptyset$.}

\medskip

We mention that since   $\Dim E_X (\alpha)\le\wt P^*(\alpha)$, this result also yields the packing dimensions simultaneously. The method used to prove this theorem requires to simultaneously building the measures  $\mu_{q}$ and computing their Hausdorff dimension; it extensively uses techniques combining analytic functions in several variables theory and large deviations estimates.  However, this approach is unable to cover the levels $\alpha\in \partial I$. When $d=1$, this boundary consists of two points, and the question has been solved partially in \cite{B2} and completely in \cite{BJ} in the case of homogeneous trees: when $\alpha\in\partial I$ takes the form $\widetilde P'(q)$ for $q\in\R$ such that  $\widetilde P^*(\widetilde P'(q))=0$, one substitutes to $\mu_q$ a measure that is naturally deduced from the fixed points of the associated smoothing transformation in the ``boundary case'' (see \cite{BiKy,B2}); when $\alpha\in\partial I$ and $\alpha=\lim_{q\to\infty}\widetilde P'(q)$ (resp. $\lim_{q\to-\infty}\widetilde P'(q)$), a ``concatenation'' method  is used in \cite{BJ} to build an ``inhomogeneous'' Mandelbrot measure carried by $E_X(\alpha)$ and with the right dimension. It turns out that these methods are not sufficient to deal with the points of $\partial I$ when $d\ge 2$. Moreover, even if they are compatible with some relaxation of the assumption \eqref{finiteness0}, there is no expectation to  treat for instance the important case where $J$ is not empty  but there are $\alpha$ such that $\wt P^*(\alpha)\ge 0$ and  $\alpha$ is not in the closure of $\nabla \wt P (J)$  ($\alpha\in \partial\wt P\setminus \overline{\mathrm{range}(\nabla \wt P)}$, where $\partial \wt P$ stands for the subdifferential of $\wt P$), a situation typical in presence of a first order phase transition  as described in the foreword of this paper.

\subsubsection{ Main results for branching random walks}\label{mainr} In this paper, we adopt an approach based on a more sophisticated construction of inhomogeneous Mandelbrot measures than that considered in \cite{BJ}, and we do not assume restriction on the distribution of the branching random walk anymore. This approach is partially inspired by the study of vector Birkhoff averages on mixing subshifts of finite type and their geometric realizations on conformal repellers \cite{FW,FLW,BSS,Ol,Iommi,Ol2}.  It  makes it possible to conduct the calculation of the Hausdorff and packing dimension of far more general sets and obtain in the context of branching random walks the counterpart of the main results obtained in the papers mentioned above about the asymptotic behavior of Birkhoff averages.  Also, we obtain a general result about the asymptotic behavior of the free energy functions $P_n$ defined in \eqref{Pn}. One advantage of the method is that it makes it possible to treat all the levels sets in a unified way.  For instance, in the special situation described above under \eqref{finiteness0},  no difference is now made between the points of  $\widering I$ and those of~$ \partial I$.

\medskip

If $K$ is a  closed subset of $\R^d$, let 
$$
E_X(K)=\Big\{t\in\partial\TT: \bigcap_{N\ge 1}\overline{ \Big \{\frac{S_nX(t)}{n}:n\ge N\Big\}}=K\Big\},
$$
the set of those $t\in\partial\TT$ such that the set of limit points of $(S_nX(t)/n)_{n\ge 1}$ is equal to~$K$. We have the natural decomposition
\begin{equation}\label{decomp}
\partial \TT= \bigcup_{K\text{ closed subset of } \R^d} E_X(K).
\end{equation}

When $t\in\partial \TT$ is such that $\frac{S_{n+1}X(t)}{n+1}-\frac{S_nX(t)}{n}= \frac{nX_{t_1\cdots t_nt_{n+1}}-S_nX(t)}{n(n+1)}$ converge to 0, then if the set of limit points $ \bigcap_{N\ge 1}\overline{ \Big \{\frac{S_nX(t)}{n}:n\ge N\Big\}}$ is not empty (i.e. $\|S_nX(t)\|/n$ does not tend to $\infty$),  it belongs to the set  of closed subsets $K$ of $\R^d$ such that each  connected component of $K$ is at distance 0 of the union of the other ones. Moreover,  it is not hard, using a standard large deviations argument, to see that if $0$ is an interior point of $\dom\wt P$, then almost surely  $\frac{nX_{t_1\cdots t_nt_{n+1}}-S_nX(t)}{n(n+1)}$ converges to 0 uniformly in $t$ and $S_nX(t)/n$ is uniformly bounded in $n$ and $t$, so that $ \bigcap_{N\ge 1}\overline{ \Big \{\frac{S_nX(t)}{n}:n\ge N\Big\}}$ is compact and connected for all $t\in\partial \TT$.   In this case, our result provides a full description of $\partial \TT$ in terms of the decomposition \eqref{decomp}, where the sets $K$ have to be taken connected and compact. More generally,  our approach makes it possible to compute $\dim E_X(K)$ and $\Dim E_X(K)$ when $K\in \mathcal K$, where 
$$
\mathcal K=\{K\subset \R^d:\, K \text{ is closed, $K\neq\emptyset$, and $K_\delta$ is pathwise connected for all $\delta>0$}\},
$$
where $K_\delta=\{\alpha\in\R^d: d(\alpha, K)\le \delta\}$. Notice that $\mathcal K$ clearly contains all the closed connected subsets of $\R^d$, and if $K\in \mathcal K$ is compact, it is necessarily connected.

\begin{thm}\label{thm-1.1}
With probability 1, we have $\dim \partial \TT=\Dim\partial \TT= \wt P(0)=\sup_{\alpha\in I} \wt P^*(\alpha)$ and for all   $K\in \mathcal K$  we have $E_X(K)\neq \emptyset$ if and only if $K\subset  I $, and in this case 
\begin{enumerate}
\item  $\dim E_X(K)=\inf_{\alpha\in K}\widetilde P^*(\alpha)$;
\item 
if $K$ is compact then $\Dim E_X(K)=\sup_{\alpha\in K}\widetilde P^*(\alpha)$, otherwise  $\Dim E_X(K)=\Dim\partial \TT$. 
\end{enumerate}
\end{thm} 
While it is expected that $\dim E_X(K)$ and $\Dim E_X(K)$ differ in general, the alternative which occurs for the packing dimension according to whether $K$ is compact or not is quite remarkable, and in some sense reminiscent from the same alternative met when studying large deviations principles in large deviations theory. 

\medskip

The first part of the next Theorem~\ref{thm-1.2} about multifractal analysis is a corollary of Theorem~\ref{thm-1.1}.  When the set $J$ defined in \eqref{J} is not empty and $\alpha$ takes the form $\nabla\wt P(q)$, $q\in J$, the second part of Theorem~\ref{thm-1.2} completes  the information provided by the approach consisting in putting on $E_X(\alpha)$ the Mandelbrot measure $\mu_q$ to get the dimension of  $E_X(\nabla\wt P(q))$. The third part, which uses the result of the second one,  provides a $0$-$\infty$ law for the Hausdorff and packing measures of the sets $E_X(\alpha)$ which are of positive but not maximal dimension. The level sets  of maximal Hausdorff dimension have a particular status. In particular when $\E(N\log (N))<\infty$, such a level set is unique and  carries a Mandelbrot measure of maximal Hausdorff  dimension on $\partial \TT$, namely the branching measure, and  the behaviors of its Hausdorff and packing measures turns out to differ from that of the other sets $E_X(\alpha)$ because they are closely related to those of the Hausdorff and packing measures of $\partial\TT$. We refer the reader to \cite{GMW,Liu0,Wa,Wa2} for the study of the Hausdorff measures and packing measures  of $\partial\TT$. We notice that   in the deterministic case, such a $0$-$\infty$ law has been obtained  in \cite{MWW} when $d=1$ for the sets $E_X(\alpha)$ seen as  Besicovich subsets of the attractor of an IFS of contractive similtudes of $\R$  satisfying the open set condition. 

The notions of generalized Hausdorff and packing measures with respect to a gauge function are recalled in Section~\ref{HandP}. 
\begin{thm}(Multifractal analysis)\label{thm-1.2}
With probability 1, 
\begin{enumerate}
\item for all $\alpha\in \R^d$, $E_X(\alpha)\neq\emptyset$ if and only if $\alpha\in I$, and in this case $\dim E_X(\alpha)=\Dim E_X(\alpha)=\widetilde P^*(\alpha)$; 

\item  if $\alpha\in I$, then $E_X(\alpha)$ carries uncountably many mutually singular inhomogeneous Mandelbrot measures of exact  dimension $\widetilde P^*(\alpha)$ (in particular $E_X(\alpha)$ is not countable when $\widetilde P^*(\alpha)=0$). 

\item for all $\alpha\in  I$ such that $0<\wt P^*(\alpha)<\dim\, \partial \TT$, for all gauge functions $g$, we have  $\mathcal H^g(E_X(\alpha))=\infty$ if $\limsup_{t\to 0^+}\log(g(t))/\log(t)\le \wt P^*(\alpha)$ and  $\mathcal H^g(E_X(\alpha))=0$ otherwise, while $\mathcal P^g(E_X(\alpha))=\infty$ if $\liminf_{t\to 0^+}\log(g(t))/\log(t)\le \wt P^*(\alpha)$ and  $\mathcal P^g(E_X(\alpha))=0$ otherwise. 
\end{enumerate}
\end{thm}

We also obtain the following large deviations results. Recall the definition \eqref{Pn}.

\begin{thm}(Large deviations and free energy)\label{thm-1.3} Let 
$$
f(\alpha):=\begin{cases}\widetilde P^*(\alpha)&\text{if }\alpha\in I\\
-\infty&\text{otherwise}
\end{cases}\quad \text{and}\quad \widehat P(q):= \sup\{ \langle q|\alpha\rangle +f(\alpha):\alpha\in\R^d\}\quad (q\in\R^d). 
$$
\begin{enumerate}
\item With probability 1, for all $\alpha\in \R^d$, 
\begin{align*}
&\lim_{\epsilon\to 0^+}\liminf_{n\to\infty} n^{-1}\log \#\big \{u\in \TT_n: n^{-1} S_nX(u) \in B(\alpha,\epsilon)\big \}\\
=&\lim_{\epsilon\to 0^+}\limsup_{n\to\infty} n^{-1}\log \#\big \{u\in \TT_n: n^{-1} S_nX(u) \in B(\alpha,\epsilon)\big \}=f(\alpha). 
\end{align*}
\item 
\begin{enumerate}
\item For all $q\in\R^d$ we have 
$$
\widehat P(q)=
\displaystyle
\inf\Big \{\frac{\wt P(\theta q)}{\theta}: 0<\theta\le 1\Big \}.
$$ 

Set $\underline P=\liminf_{n\to\infty} P_n$ and $\overline P=\limsup_{n\to\infty} P_n$.
\item
For all $q\in\R^d$, with probability 1, we have $\lim_{n\to\infty} P_n(q)=\widehat P(q)$. 

\item With probability 1, we have $\underline P\ge \widehat P$ over $\R^d$, and $\lim_{n\to\infty} P_n=\widehat P$ over $\R^d\setminus \mathrm{rel}\,\partial (\dom \widehat P)$,  as well as over the points of $\mathrm{rel}\,\partial (\dom \widehat P)$ at which $\overline P$ is lower semi-continuous.

\item If $I$ is compact or $d=1$, with probability 1, we have $\lim_{n\to\infty} P_n=\widehat P$ over $\R^d$.  
\end{enumerate}
\end{enumerate} 
\end{thm}

According to the definition given in Section~\ref{foreword}, we say that there is a phase transition at $q\in\R^d$ if there is one at $\theta_c=1$ for the random energy model associated with the branching random walk $\langle q|S_nX\rangle$, whose moment generating function is $\theta\mapsto \wt P(\theta q)$. Equivalently, this corresponds to a linearization of $\widehat P$ at $q$ in the direction of $q$, with first order phase transition if $\widehat P$ is not differentiable at $q$ in the direction of $q$, and second order phase transition if it is.

Theorem~\ref{thm-1.3}(1) is known for $\alpha\in\nabla\wt P(J')$ when $J'$, the set of those points in $J$ such that there exists $\gamma>1$ for which $\E(S(q')^\gamma)<\infty$ in a neighborhood of $q$, is not empty (see \cite{Biggins2,A}); in the case $d=1$ and when $N$ is  bounded, the result follows from \cite{HoWa} (under the assumption that the $X_i$ are independent and bounded) and \cite{BJ} in absence of a first order phase transition. It also  follows from \cite{Biggins2,A} that almost surely  $\lim_{n\to\infty} P_n(q)=\wt P(q)$ for all $q\in J'$. When $d=1$, and $N$ is bounded, Theorem~\ref{thm-1.3}(2)(c) is known in absence of first order phase transition \cite{HoWa,CoKo1992,Franchi,Mol,OW}, and a weak version is obtained in \cite{Mol} in case of first order phase transition is the sense that the convergence is proved to hold along a subsequence (see also  Section~\ref{foreword}).

\medskip Theorems~\ref{thm-1.2} and~\ref{thm-1.3} have the following consequence on the growth of the minimal supporting subtree for the free energy of polymers on disordered trees, which completes the results  obtained in \cite{Moerters-Ortgiese} for regular trees in the case of no phase transition or second order phase transition; in order to simplify the statement, we will assume that $\E (N)<\infty$, i.e. $\widetilde P(0)<\infty$ (see Figure~\ref{Fig2} for an illustration of the free energy behavior in the different cases): 
\begin{thm}\label{thm-1.4}Suppose that $d=1$ and $\E (N)<\infty$. Let $q_c=\sup\{q> 0: \widetilde P(\widetilde P'(q^-))\ge 0\}$, with the convention $\sup\emptyset =0$. With probability~1:
\begin{enumerate}
\item If $0< q<q_c$, there are uncountably many trees $\widetilde {\TT}\subset \TT$ such that 
$$
\lim_{n\to\infty}  \frac{1}{n}\log  \#\widetilde \TT_n= \widetilde P^*(\widetilde P'(q))\text{ and }
\lim_{n\to\infty} \frac{1}{n}\log\sum_{u\in \widetilde \TT_n}e^{qS_n(u)}=\widetilde P(q);
$$
Moreover, for any sequence $(A_n)_{n\ge 1}$ such that $A_n\subset \TT_n$ for all $n\ge 1$ and $\limsup_{n\to\infty}  \dfrac{1}{n}\log  \#A_n< \widetilde P^*(\widetilde P'(q))$, one has 
\begin{equation}\label{subest}
\limsup_{n\to\infty} \frac{1}{n}\log\sum_{u\in A_n }e^{qS_n(u)}<\widetilde P(q).
\end{equation}
\item If $0<q_c<\infty $ and $\widetilde P'(q_c^-)=\widetilde P(q_c)/q_c$ (second order phase transition), then  there are uncountably many $t\in \partial\TT$ such that for any $q\ge q_c$ and sequence $(A_n)_{n\ge 1}$ such that $A_n\subset \TT_n$ and $t_{|n}\in A_n$ for all $n\ge 1$, one has  
$$
\lim_{n\to\infty} \frac{1}{n}\log\sum_{u\in A_n }e^{qS_n(u)}=  q\widetilde P'(q_c)=\frac{\wt P(q_c)}{q_c}q.
$$
\item   If $0<q_c<\infty $ and $\widetilde P(\widetilde P'(q_c^-))>0$ (first order phase transition), then for all $\alpha \in [\widetilde P'(q_c^-), \widetilde P(q_c)/q_c]$, there are uncountably many trees $\TT_\alpha\subset\TT$ such that $\lim_{n\to\infty} \dfrac{1}{n}\log  \# \TT_{\alpha,n}= \widetilde P^*(\alpha)$ and for any $q\ge q_c$, one has 
$$
\lim_{n\to\infty} \frac{1}{n}\log\sum_{u\in  \TT_{\alpha,n}}e^{qS_n(u)}= \frac{\wt P(q_c)}{q_c}q.
$$
If, moreover, $\widetilde P(q_c)=0$, the same conclusion as in 2. holds. 

\item If $q_c=0$ (degenerate case), then for all $\alpha \in \R_+$, there are uncountably many trees $\TT_\alpha\subset\TT$ such that $\lim_{n\to\infty} \dfrac{1}{n}\log  \# \TT_{\alpha,n}= \widetilde P(0)$ and for any $q\ge q_c$, one has 
$$
\lim_{n\to\infty} \frac{1}{n}\log\sum_{u\in  \TT_{\alpha,n}}e^{qS_n(u)}=\wt P(0)+\alpha q.
$$
\end{enumerate}
\end{thm}

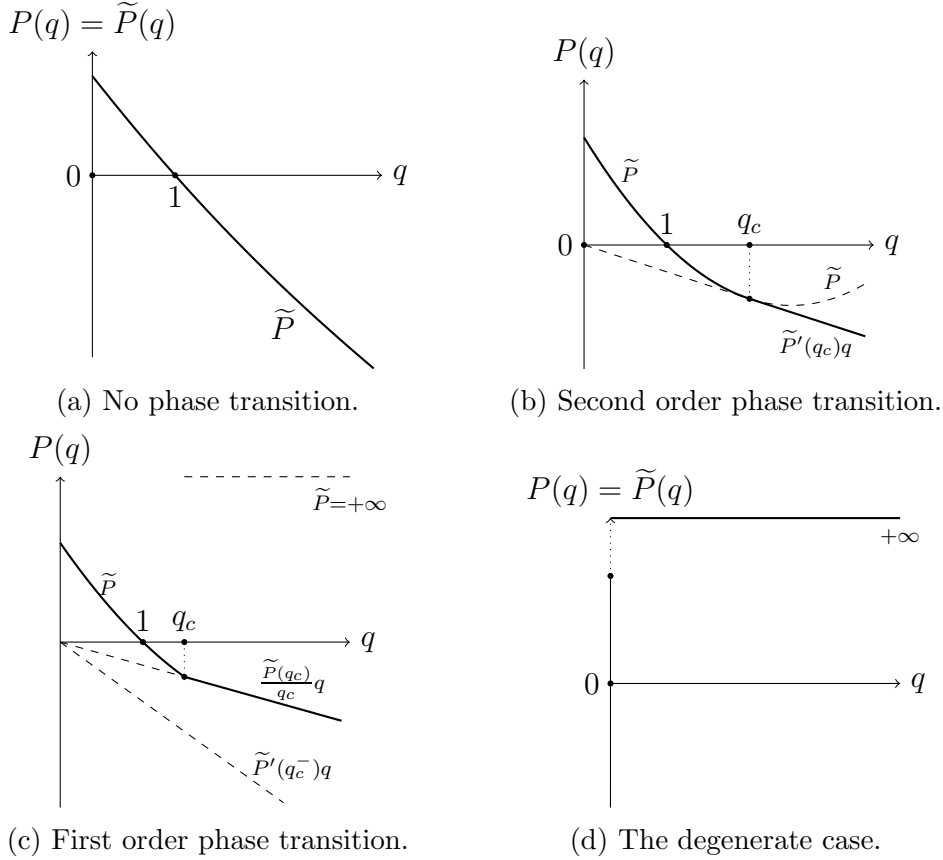
\begin{figure}[ht]
\begin{subfigure}[b]{0.45\textwidth}
\centering
\begin{tikzpicture}[xscale=1.1,yscale=1.1]
\draw [->] (0,-2.2) -- (0,1.5) node [above] {$P(q)=\wt P(q)$};
\draw [->] (0,0) -- (3.5,0) node [right] {$q$};
\draw [thick, domain=0:3.4] plot (\x, {-1.2*(\x-1-ln((1/3)*exp(\x*ln(3/2))+(2/3)*(exp(\x*ln(3/4))))});
\draw [fill] (0,0) circle [radius=0.03] node [left] {$0$};
\draw [fill] (1,0) circle [radius=0.03] node [below] {$1$};
%\draw [fill] (0,1) circle [radius=0.03] node [left] {$1$};
\node at (2.3,-1.8) {$\wt P$};
\end{tikzpicture}
\caption{No phase transition.}\label{fig2a}
\end{subfigure}
\begin{subfigure}[b]{0.45\textwidth}
\centering
\begin{tikzpicture}[xscale=1.1,yscale=1.1]
\draw [->] (0,-1.5) -- (0,2) node [above] {$P(q)$};
\draw [->] (0,0) -- (3.5,0) node [right] {$q$};
\draw [thick, domain=0:2] plot (\x, {1.3*0.25*(\x-4)*(\x-1)});
\draw [thick, domain=2:3.4] plot (\x, {-1.3*0.25*\x});
\draw [dashed,domain=2:3.4] plot (\x, {1.3*0.25*(\x-4)*(\x-1)});
\draw [dashed,domain=0:2] plot (\x, {-1.3*0.25*\x});
\draw [fill] (0,0) circle [radius=0.03] node [left] {$0$};
\draw [fill] (2,0) circle [radius=0.03] node [above] {$q_c$};
\draw [fill] (1,0) circle [radius=0.03] node [above] {$1$};
%\draw [fill] (0,1) circle [radius=0.03] node [left] {$1$};
\draw [fill] (2,-1.3*0.5) circle [radius=0.03];
\draw [dotted] (2,-1.3*0.5) -- (2,0);
\node at (2.8,-1.2) {${\scriptstyle\wt P'(q_c) q}$};
\node at (0.55,0.9) {${\scriptstyle\wt P}$};
\node at (3,-.4) {${\scriptstyle\wt P}$};
\end{tikzpicture}
\caption{Second order phase transition.}\label{fig2b}
\end{subfigure}
\begin{subfigure}[b]{0.45\textwidth}
\centering
\begin{tikzpicture}[xscale=1.1,yscale=1.1]
\draw [->] (0,-2) -- (0,2) node [above] {$P(q)$};
\draw [->] (0,0) -- (3.5,0) node [right] {$q$};
\draw [thick, domain=0:1.5] plot (\x, {1.2*0.2*(\x-5)*(\x-1)});
\draw [thick, domain=1.5:3.4] plot (\x, {-1.2*(0.7/3)*\x});
\draw [dashed, domain=0:1.5] plot (\x, {-1.2*(0.7/3)*\x});
\draw [dashed, domain=0:2.7] plot (\x, {-1.2*0.6*\x});
%\draw [fill] (0,0) circle [radius=0.03] node [left] {$0$};
\draw [fill] (1.5,0) circle [radius=0.03] node [above] {$q_c$};
\draw [fill] (1,0) circle [radius=0.03] node [above] {$1$};
%\draw [fill] (0,1) circle [radius=0.03] node [left] {$1$};
\draw [fill] (1.5,-1.2*0.35) circle [radius=0.03];
\draw [dotted] (1.5,0) -- (1.5,-1.2*0.35);
\draw [dashed] (1.5,2) -- (3.5,2) node [below] {${\scriptstyle \wt P=+\infty}$};
\node at (0.57,0.7) {${\scriptstyle\wt P}$};
\node at (2.8,-1.5) {${\scriptstyle\wt P'(q_c^-)q}$};
\node at (2.8,-0.45) {${\scriptstyle\frac{\wt P(q_c)}{q_c}q}$};
\end{tikzpicture}
\caption{First order phase transition.}\label{fig2c}
\end{subfigure}
\begin{subfigure}[b]{0.45\textwidth}
\centering
\begin{tikzpicture}[xscale=1.1,yscale=1.1]
\draw [dotted,->] (0,1) -- (0,2) node [above] {$P(q)=\wt P(q)$};
\draw (0,-1.5) -- (0,1.3);
\draw [->] (0,0) -- (3.5,0) node [right] {$q$};
\draw [thick] (0,2) -- (3.5,2) node [below] {${\scriptstyle +\infty}$};
\draw [fill] (0,1.3) circle [radius=0.03] node [left] {$$};
\draw [fill] (0,0) circle [radius=0.03] node [left] {$0$};
\draw [dotted] (0,-1) -- (0,-1.5);
\end{tikzpicture}
\caption{The degenerate case.}\label{fig2d}
\end{subfigure}
\caption{The free energy function $P=\lim_{n\to\infty}P_n$ over $\R_+$  in the four possible situations when $\wt P(0)=\log \E(N)<\infty$, and in the three first cases under the normalization $(\wt P(1)=0,\wt P'(1^-)<0$), which is necessary and sufficient for the associated Mandelbrot measure $\mu_1$ to be non degenerate.}\label{Fig2}
\end{figure}

We notice that in \cite{Moerters-Ortgiese}, parts 1. and 2. of the previous statement are obtained only for each fixed $0<q <q_c$ almost surely, and that the approach used there provides at most countably rather than uncountably many minimal supporting trees.

Before stating our result on Mandelbrot measures, we give some examples and other comments.

\subsection*{Examples and additional comments}  Let us  now describe  some explicit  situations. The case of a first order phase transition will be illustrated in the section dedicated to Mandelbrot measures. 
\medskip

\noindent
(1) ($I=\R^d$) This is the case when $\dom  \wt P=\{0\}$ or $\dom \wt P=\emptyset$. Then $I=\R^d$, and almost surely, for any  closed connected subset $K$ of ${\R^d}$, one has $\dim K=\Dim K=\dim\partial T=\widetilde P(0)$. 

\medskip

\noindent
(2) ($I$ is a half space). If $\dom  \wt P=\{q_0\}$ with $q_0\neq0$, then $I$ is the half-space $\{\alpha\in\R^d: \wt P(q_0)\ge \langle q_0|\alpha\rangle\}$,  and $\wt P^*(\alpha):  \wt P(q_0)-\langle q_0|\alpha\rangle$ for all $\alpha\in I$. In that case $\dim \partial \TT=\infty$ almost surely, and one meets all the possible cases  considered in Theorem~\ref{thm-1.1} regarding the dimensions of the sets $E_X(K)$. Here is an example.

For $k\ge 2$, let $p_k$ and $f(k)$ be defined by 
$$
\begin{cases}
\displaystyle p_{k}= \frac{c}{k^2} \text{ and } e^{f(k)}=\frac{1}{\log (k)  (\log\log (k))^2}&\text{if $k$ is even}\\
\displaystyle p_k= \frac{c}{k^2 (\log k)^2 (\log\log (k))^2}\text{ and } e^{f(k)}=\log (k) &\text{otherwise},
\end{cases}
$$
where $c$ is chosen so that $\sum_{k\ge 2} p_k=1$. 

Let $N$ be a random integer of distribution $\sum_{k \ge 2} p_k\delta_k$. Fix $Y_2,\ldots, Y_{d}$, $d-1$  independent copies of a real valued random variable $Y$ such that $\E(e^{qY})<\infty$ at $q=0$ only. Then, for all $i\ge 1$, let $X_i=(f(N),Y_2,\cdots,Y_{d})$. We have $\wt P(q)= \log (\E(N\exp(q_1f(N))) +\sum_{s=2}^{d} \log (\E(e^{q_sY}))$, which by construction is finite at $(1,0\ldots,0)$ only.

\medskip

\noindent
(3) ($I$ is compact). Suppose the $X_i$  are identically distributed Gaussian vectors independent of $N$, and $\E(N)<\infty$.  In this case, denoting by $\lambda_1,\ldots, \lambda_d$ the eigenvalues of  covariance matrix common to the $X_i$, in a suitable orthonormal basis we have $\displaystyle \wt P(q)=\log (\E (N))+ \sum_{i=1}^d\frac{\lambda_i q_i^2}{2}$,  and $\displaystyle\wt P^*(\alpha)=\log (\E (N))-\sum_{i=1}^d\frac{\alpha_i^2}{2\lambda_i}$; hence $I$ is the convex hull of  the ellipsoid $(\wt P^*)^{-1}(0)$.  

\medskip

\noindent
(4) ($I$ is not bounded, distinct from $\R^d$ and $\Dim\partial \TT<\infty$) Consider for the $X_i$ random vectors with $d$ i.i.d.  negative $\beta$-stable components  ($\beta\in(0,1)$), and independent of $N$.  There exists $c>0$ such that $\wt P(q)= \log (\E(N))-c \sum_{i=1}^d q_i^\beta$ for $q\in [0,\infty)^d$ and $\wt P(q)=\infty$ for $q\in\R^d\setminus [0,\infty)^d$. We have the explicit expression $\displaystyle \wt P^*(\alpha)= \log (\E(N)) -\frac{c(1-\beta)}{(c\beta)^{\beta/(\beta-1)}}\sum_{i=1}^d \alpha_i^{\beta/(\beta-1)}$ for $\alpha\in (0,\infty)^d$, and $\displaystyle \wt P^*(\alpha)=\infty$ for $\alpha\in \R^d\setminus (0,\infty)^d$, and $I$ is  unbounded. Moreover, it is easily checked that $I=\overline{\nabla \wt P(J)}$  (see \eqref{J} for the definition of $J$) and $\wt P^*$ is strictly concave in this case. 

However, $I$ can be unbounded, distinct from $\R^d$, and $\wt P^*$ constant over a non trivial convex subset of $I$. To be simple let $d=1$. Suppose that all the $X_i$ are identically distributed with $X$ and independent of $N$, and that $e^X$ has the density  $c \mathbf{1}_{\{u\ge 1\}} \frac{\mathrm{d}u}{u\log (u)^3}$. Then $P(q)$ is finite over $(-\infty,0]$, infinite over $(0,\infty)$, and since $\wt P'(0^-)<\infty$, we have $\wt P^*(\alpha)=\wt P(0)$ for all $\alpha\ge \wt P'(0^-)$. In this case, $I\supsetneq\overline{\nabla \wt P(J)}$. 

\medskip

\noindent
(5) If $\mathrm{int}(\dom \wt P)\neq\emptyset$,  the set $I$ is unbounded only if $0$ is not an interior point of ${\dom \wt P}$. Indeed, if $(\alpha_n)_{n\in\N}\in I^\N$ tends to $\infty$, at least one of its components tends to $\infty$ in absolute value. Without loss of generality assume it is the first one, and let $e_1$ stand for the first vector of the canonical basis. If $0$ is interior to $\dom \wt P$, then we can find $\epsilon_1>0$ such that both $\wt P(-\epsilon_1e_1)$ and  $\wt P(\epsilon_1 e_1)$ are finite, and one of the sequences $\wt P(-\epsilon_1e_1)-\langle -\epsilon_1e_1|\alpha_{n}\rangle$ or  $\wt P(\epsilon_1e_1)-\langle \epsilon_1e_1|\alpha_{n}\rangle$ must have a subsequence converging to $-\infty$. However, $\wt P^*(\alpha_{n})\le \min (\wt P(-\epsilon_1e_1)-\langle -\epsilon_1e_1|\alpha_{n}\rangle,\wt P(\epsilon_1e_1)-\langle \epsilon_1e_1|\alpha_{n}\rangle)$, so $\wt P^*(\alpha_n)<0$ for $n$ large enough,  which contradicts the fact that $(\alpha_n)_{n\in\N}\in I^\N$.

\medskip

\noindent
(6) Let us give other  consequences of our study. The first one concerns  the branching process itself: The previous results apply to  the natural branching random walk associated with the branching numbers, namely $S_nN(t)=N_{t_1}+N_{t_1t_2}+\cdots+N_{t_1\cdots t_n}$ and provide, if $N$ is not constant, geometric and large deviations information on the heterogeneity of the birth process along  different infinite branches. 

Information can also be derived for the branching random walk obtained from an homogeneous percolation process: the $X_{u}$, $u\in \bigcup_{n\ge 1}\N_+^n$ are independent copies of the same Bernoulli variable, and are independent of $\partial \TT$. The branching random walk $S_nX(t)$ must be interpreted as the  covering number of $t$ by the family of balls $[u]$ of generation not greater than $n$  such that $X_u=1$. Here our results cover and improve those of \cite{FK} about the random covering of dyadic tree, which only achieves the  multifractal analysis in the weak sense ``for each $\alpha\in\widering I$ almost surely''; however it is worth mentioning that in \cite{FK}  inhomogeneous Mandelbrot martingales are built individually to describe some fine subsets of the sets $E_X(\alpha)$ and show that the set of those points $t$ for  which $S_n(t)/n$ does not converge is almost surely of full Hausdorff dimension.  

Finally, our results also provide for instance a joint multifractal analysis of $S_nN(t)$ and $S_nX(t)$. 

\medskip

\noindent
(7) It follows from Theorems~\ref{thm-1.2} and \ref{thm-1.3} that $\dim E_X(\alpha)=\Dim E_X(\alpha)=\overline P^*(\alpha)$ almost surely for all $\alpha\in I$. In this sense the vector multifractal formalism considered in \cite{Pey4} is fulfilled by $S_nX(t)$. 

\medskip

\noindent
(8) Theorem~\ref{thm-1.1} should be compared to the results of \cite{FW} and \cite{Ol,Ol2}, and  Theorems~\ref{thm-1.2}(1) and~\ref{thm-1.3} to the results of \cite{FF}, obtained in the context of Birkhoff averages of continuous potentials  over a symbolic space endowed with the standard metric. However, in this context the set of possible levels for the averages is always compact. One meets unbounded such sets when one works on topological dynamical systems encoded by symbolic spaces over infinite alphabets \cite{Iommi}. However, in this last context, to our best knowledge only the multifractal analysis has been considered, and for $1$-dimensional potentials. 

\subsubsection{Application to Mandelbrot measures} Suppose that $d=1$,  $\dom \wt P$ contains a neighborhood of $0$,  and  $(N,X_1,X_2,\ldots)$ is normalized so that 
\begin{equation}\label{defMM}
\wt P(1)=0,\ \wt P'(1^-)<0, \ \text{and } \E\big(S(1) \log^+(S(1))\big)<\infty
\end{equation}
In particular $0<\wt P(0)<\infty$ . Also the fact that $0$ is in the interior of $\dom \wt P$ implies that $I$ is compact (see comment (5) above). Conditions \eqref{defMM} are the necessary and sufficient conditions for the Mandelbrot measure defined by \eqref{defmuq} in the case $q=1$ to be almost surely positive. We simply denote this measure $\mu_1$ by $\mu$.  

Recall that the interval $J$ is defined in \eqref{J}. Suppose in addition that $\P(N\ge 2)=1$ (this gives $\E(\|\mu\|^q)<\infty$ for all $q\in (\inf (\dom \wt P),0)$; see \cite[Theorem 2.4]{Liu3}), and 
$$
   \E(S(1)^q)<\infty\quad (\forall\, q\in J\cap (1,\infty))
$$
(this gives the finiteness of $\E(\|\mu\|^q)$ for $q\in J\cap (1,\infty)$, see \cite[Theorem 2.1]{Liu2}; notice that $J\cap (1,\infty)$ is empty if and only if $1=\max (\dom \wt P)$).
\begin{thm}\label{MM}
The results of Theorems~\ref{thm-1.1}, \ref{thm-1.2} and \ref{thm-1.3} hold if we replace $S_nX(t)$ by $\log (\mu([t_1\cdots t_n]))$ in the definition of the sets $E_X(K)$ and in the statements. 
\end{thm}
\noindent{\bf Comments.} (1) Under our assumptions, situations presenting a first order phase transition are now covered by our result. A particularly interesting case corresponds to $\max (\dom \wt P)=1$. In this case, the assumption about $q\in J\cap (1,\infty)$ is empty, and  the singularity spectrum of the continuous exact dimensional measure $\mu$, i.e. the mapping $\displaystyle \beta\mapsto \dim E_\mu(\beta)=f(-\beta)$, where $E_\mu(\beta)=\Big \{t\in\partial \TT: \lim_{n\to\infty}\frac{\log (\mu([t_1\cdots t_n]))}{-n}=\beta\Big \}$, has the property to be equal to the identity map over $[0,\dim \mu]$. This is a remarkable phenomenon, since for statistically self-similar measures, spectra whose graph contains a linear part starting from $(0,0)$ are usually associated with discrete measures obtained essentially  after subordinating the indefinite integral of a Mandelbrot measure to a stable L\'evy subordinator \cite{BaSe07}, in which case the slope is equal to the index of the L\'evy process, so is smaller than 1. We notice that the spectrum contains another (decreasing) linear part if a second second order phase transition occurs at $\inf(\dom \wt P)$).

\begin{figure}[ht]
\begin{subfigure}[b]{0.45\textwidth}
\centering
\begin{tikzpicture}[xscale=0.6,yscale=0.6]
\draw [->] (0,-1) -- (0,6) node [above] {$P(q)$};
\draw [->] (-3,0) -- (3.5,0) node [right] {$q$};
\draw [thick,domain=-2:1] plot (\x, {0.25*(\x-4)*(\x-1)});
\draw [thick,domain=-2:-2.6] plot (\x, {-(9/4)*\x)});
\draw [thick] (1,0) -- (3.5,0);
\draw [dashed, domain=-2:0] plot (\x,{(-9/4)*\x});
\draw [fill] (0,0) circle [radius=0.04] node [below left] {$0$};
\draw [fill] (1,0) circle [radius=0.04] node [above] {$1$};
\draw [fill] (-2,0) circle [radius=0.04] node [above] {$q_{-}$};
\draw [fill] (0,4.5) circle [radius=0.04] node [right] {$ P(q_{-})$};
\draw [fill] (0,-1) circle [radius=0.04] node [right] {$-1$};
\draw [fill] (-2,4.5) circle [radius=0.04];
\draw [dotted] (-2,0) -- (-2,4.5);
\draw [dotted] (-2,4.5) -- (0,4.5) ;
\end{tikzpicture}
\caption{The free energy function $P$.}\label{fig3a}
\end{subfigure}
\begin{subfigure}[b]{0.45\textwidth}
\centering
\begin{tikzpicture}[xscale=1.8,yscale=1.8]
\draw [->] (0,0) -- (0,1.5) node [above] {$P^*(-\beta)=\dim E_\mu(\beta)$};
\draw [->] (0,0) -- (3,0) node [right] {$\beta$};
\draw [thick, domain=3/4:9/4] plot (\x, {1-(\x-5/4)^2});
\draw [thick, domain=0:3/4] plot (\x,{\x});
\draw [fill] (0,0) circle [radius=0.02] node [left] {$0$};
\draw [fill] (9/4,0) circle [radius=0.02] node [below] {${\scriptstyle -\wt P(q_{-})}$};
\draw [fill] (5/4,0) circle [radius=0.02] node [below] {${\scriptstyle -\wt P'(0)}$};
\draw [fill] (3/4-.1,0) circle [radius=0.02] node [below] {${\scriptstyle -\wt P'(1^-)}$};
\draw [fill] (0,3/4-.1) circle [radius=0.02] node [left] {${\scriptstyle -\wt P'(1^-)}$};
\draw [fill] (0,1) circle [radius=0.02] node [left] {$\scriptstyle \wt P(0)$};
\draw [fill] (5/4,1) circle [radius=0.02];
\draw [fill] (3/4-.1,3/4-.1) circle [radius=0.02];
\draw [dotted] (5/4,0) -- (5/4,1) -- (0,1);
\draw [dotted] (0,3/4-.1) -- (3/4-.1,3/4-.1) -- (3/4-.1,0);
\end{tikzpicture}
\caption{The singularity spectrum of $\mu$.}\label{fig3b}
\end{subfigure}

\caption{Multifractal nature of the Mandelbrot measure $\mu$ when  a second order phase transition occurs at $q_->0$ and a first order phase transition occurs at $q_+=1$.}\label{Fig3}
\end{figure}
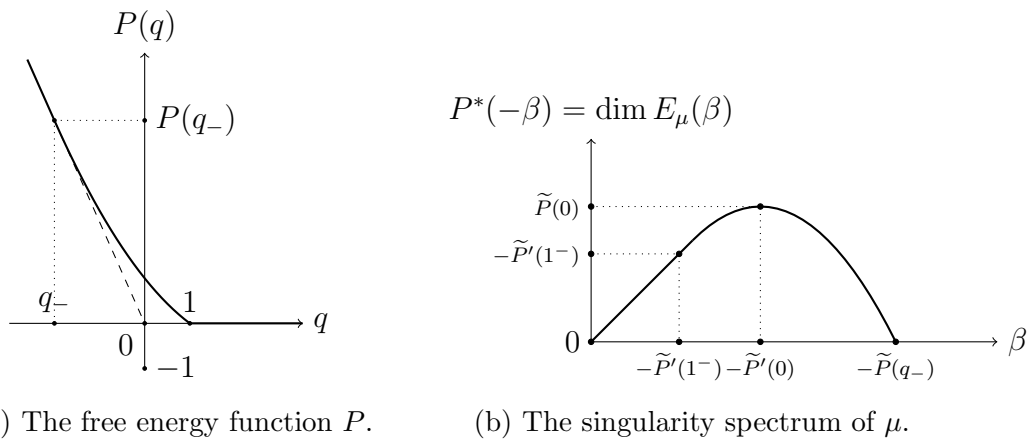

\medskip

\noindent (2) Let us explain how to build examples with a first order phase transition. To simplify the purpose, we suppose that $N=2$ almost surely. Let $Z$ be a real valued random variable such that  for $q>0$, $\mathbb E(e^{qZ})<\infty$ if and only if  $q\le 1$, and there exists $\epsilon>0$ such that $\mathbb E(e^{-\epsilon Z})<\infty$. Consider $(Z_i)_{i\ge 1}$ a collection of independent copies of $Z$, $\theta$ a Poisson random variable of positive expectation, and set $Y=Z_1+\cdots +Z_\theta$.   By construction, the law of $Y$ is infinitely divisible and $Y$ inherits the property that  for $q>0$, $\mathbb E(e^{qY})<\infty$ if and only if  $q\le 1$, and  $\mathbb E(e^{-\epsilon Y})<\infty$. Then let $Y_1$ and $Y_2$ be independent copies of $Y$. By adding the same  constant to $Y_1$ and $Y_2$ if necessary, we may assume that  $\E(e^{Y_i})=1$. For each $k\ge 1$, write $Y_i=X^{(k)}_{i,1}+\cdots +X^{(k)}_{i,1}$ where the $X^{(k)}_{i,j}$ are i.i.d. By construction, we have $\mathbb E(e^{X^{(k)}_{i,j}})=1$ and $\mathbb E(e^{X^{(k)}_{i,j}}X^{(k)}_{i,j})=k^{-1}\mathbb E(e^{Y}Y)$. Consequently, taking $X_i=X^{(k)}_{i,1}-\log(2)$ and $X=(X_1,X_2)$, we have $\widetilde P(-\epsilon)<\infty$, $\widetilde P(1)=0$, $\wt P'(1)=\mathbb E(e^{X^{(k)}_{1,1}}X^{(k)}_{1,1})-\log(2)$, which is negative for $k$ large enough, and $\wt P(q)=\infty$ for $q>1$.

\medskip

\noindent (3) When $N$ is constant and equal to a fixed integer of the form $b^m$ with $b\ge 2$, it is convenient to replace $d(s,t)$ by $b^{-|s\land t|}$ in accordance with the natural geometric realization of the measure $\mu$ on $[0,1]^m$. Then all the dimensions must be divided by $\log (b)$.  

\medskip

\noindent (4) In this section we discarded the case where $N=1$ with positive probability. This situation creates additional difficulties due to a more delicate control of the moments of negative orders of $\|\mu\|$. In \cite{Moerters} the special situation of the branching measure is considered and partial results are obtained for the multifractal analysis.  It would be interesting to try combining our  method with some ideas in  \cite{Moerters}  to deal with the general case.

\bigskip

Section~\ref{sec2} is dedicated to the proof of Theorem~\ref{thm-1.1}, and Section~\ref{other proofs} to those of Theorems \ref{thm-1.2} and \ref{thm-1.3} in the case $\dom\wt P\neq\emptyset$, as well as to the proof of Theorem~\ref{thm-1.4}. Then in Section~\ref{domvide} we explain how to extend Theorems \ref{thm-1.2} and \ref{thm-1.3} in the case $\dom\wt P=\emptyset$. Finally, Section~\ref{MMM} establishes Theorem~\ref{MM}. 

In a companion paper we will extend our results on sets $E_X(K)$ to the case where $\partial \TT$ is endowed with a  metric naturally built from a positive branching random walk. This requires additional technicalities, and the results take a slightly different form, but the main ideas  follow those developed in the present work.

\section{Dimensions of general level sets. Proof of Theorem~\ref{thm-1.1} when $\dom \wt P\neq\emptyset$}\label{sec2}

In this section, we assume that $\dom \wt P^*\neq\emptyset$. Hence $\wt P$ is a proper lower semi-continuous convex function, and $\wt P^*$ a proper upper semi-continuous function. However, the results of Sections~\ref{upb} and \ref{uppack} are trivially  valid without this assumption. 

\subsection{Approximations of $\wt P$}\label{sec2.1}

Recall that the extinction probability of a Galton-Watson process is the smallest fixed point of its generating function. 

For each $A\in\N_+$ let  $N_A=\sum_{i=1}^{N\land A} \mathbf{1}_{\{\|X_i\|\le A\}}$. 

\begin{prop}\label{extinction}
The extinction probability of a Galton-Watson process with offspring distribution given by the law of $N_A$ converges to $0$ as $A\to\infty$.  
\end{prop}
\begin{proof} Since $N_A\nearrow N$, we have $ \P(N_A\ge k)\nearrow\P(N\ge k)$ for all $k\in\N$. This implies that   the generating function of $N_A$, namely $g_{A}=\sum_{k=0}^\infty \P(N_A=k)t^k=1-\sum_{k=1}^\infty \P(N_A\ge k)(t^{k-1}-t^{k})$ converges uniformly to  the generating function $g$ of $N$ as $A\to\infty$ over $[0,1]$. Moreover, our assumption $\mathbb E(N)>1$ implies that $g$ is strictly convex, hence has only $0$ and $1$ as fixed points. Consequently, $\P(N_A=0)$, the smallest fixed point of $g_A$, converges to that of $g$, i.e. 0. 
\end{proof}
Let $\wt A\in\N_+$ such that for all integers $A\ge \wt A$ we have $\E(N_A)>1$. Then for all  integers $A\ge \wt A$ set
\begin{equation}\label{wtPA}
\wt P_A:q\in\R^d\mapsto \displaystyle \log \mathbb{E}(S_A(q))\in \R,\text{ where } S_A(q)=\sum_{i=1}^{N\land A} \mathbf{1}_{\{\|X_i\|\le A\}}\exp (\langle q|X_i\rangle).
\end{equation}
Also, let $I_A=\{\alpha\in\R^d: \widetilde P_A^*(\alpha)\ge 0\}$ and  $J_A=\{q\in\R^d: \widetilde P_A^*(\nabla \widetilde P_A(q))>0\}$.

\begin{prop}\label{detI} 
\begin{enumerate}
\item  For all integers $A\ge \wt A$ the set $I_A$ is  convex, compact, non-empty, and included in $I$. 

\item There exists $A_0\ge \wt A$ such that for all $A\ge A_0$, $\wt P_A$ is strictly convex, $\widering I_A= \nabla \widetilde P_A (J_A)$, hence $I_A=\overline{ \{\nabla \widetilde P_A(q): q\in J_A\}}$. 
\end{enumerate}
\end{prop}

\begin{rem}{\rm The following proof shows that the same conclusions as in the proposition hold for $I$ and $\wt P$ if $\dom \wt P=\R^d$. 

}
\end{rem}
\begin{proof} Recall that we assumed  \eqref{pP} without loss of generality. Let $A\ge \wt A$. 

\medskip

\noindent
(1) Since $\wt P\ge \wt P_A$ we have $\wt P_A^*\le \wt P^*$, hence $I_A\subset I$. Now let us simply denote $I_A$ by $I$, $N_A$ by $N$ and $\widetilde P_A$ by $\wt P$. At first notice that $I$ contains $\nabla\wt P(0)$ ($\wt P^*(\nabla \wt P(0))=\wt P(0)= \log (\mathbb{E}(N))>0$). The convexity of $I$ comes from the concavity of the function $\widetilde P^*$. The fact that $I$ is closed results from the upper semi-continuity of $\widetilde P^*$.  It remains to show that $I$ is bounded. Suppose that this is not the case. Let  $(\alpha_{n})_{n\geq 1}\in I^{\N_+}$ which tends to $\infty$ as $n$ tends to $\infty$. Since $\mathbb S^{d-1}$ is compact, without loss of generality we can assume that $(\frac{\alpha_{n}}{\|\alpha_{n}\|})_{n\geq 1}$  converges to a limit $u\in \mathbb S^{d-1}$ as $n\to\infty$. 
Let $ \lambda > 0$. From the definition of $\widetilde P^*$, since $\widetilde P^*(\alpha)\ge 0$, we have 
\begin{eqnarray*}
0 \leq \widetilde{P}(\lambda u) - \lambda \langle u|  \alpha_{n}\rangle  &=& \widetilde{P}(\lambda u) - \lambda \langle u|  \|\alpha_{n}\| u\rangle +  \lambda \langle u|  \|\alpha_{n}\| u- \alpha_{n}\rangle\\\
&=& \widetilde{P}(\lambda u) - \lambda  \| \alpha_{n}\| \Big (1 + \Big \langle u \Big | u -\frac{ \alpha_{n}}{\| \alpha_{n}\|}\Big \rangle \Big ).\\ 
\end{eqnarray*} 
Since $\frac{ \alpha_{n}}{\| \alpha_{n}\|}$ converge to $u$ as $n\to\infty$, this yields  
 $$  \widetilde{P}(\lambda u) \geq \lambda  \| \alpha_{n_k}\| (1 + o(1)) \quad \text{as $n\to\infty$},
 $$
 hence $\widetilde P(\lambda u)=\infty$ for all $\lambda> 0$. This contradics the finiteness of $\widetilde P$ over $\R^d$.  

\medskip

\noindent (2) We first show that \eqref{pP}  implies the strict convexity of $\widetilde P_A$ for $A$ large enough, this coming from the strongest property that for $A$ large enough the second derivative of $\wt P_A$ is positive definite at every $q\in\R^d$. This implies $\nabla \widetilde P_A$ is everywhere a local diffeomorphism, so $\nabla \wt P_A(J_A)$ is open.  

Let $A\ge \wt A$. Suppose that the second derivative of $\wt P_A$ is not positive definite at $q\in\R^d$. A calculation shows that there exists $h\in\R^d\setminus \{0\}$ such that 
$$
e^{\wt P_A(q)}\E\Big (\sum_{i=1}^{N\land A}\mathbf{1}_{\{\|X_i\|\le A\}} e^{\langle q|X_i\rangle } \Big (\sum_{j=1}^dh_jX_{i,j}\Big )^2\Big )=\Big (\E\Big( \sum_{i=1}^{N\land A}\mathbf{1}_{\{\|X_i\|\le A\}} e^{\langle q|X_i\rangle } \sum_{j=1}^dh_jX_{i,j}\Big )\Big )^2.
$$
Notice that the expression in the second expectation is the scalar product of the vectors $(\mathbf{1}_{\{\|X_i\|\le A\}}e^{\langle q|X_i\rangle/2 } )_{1\le i\le N\land A}$ and $( \mathbf{1}_{\{\|X_i\|\le A\}}e^{\langle q|X_i\rangle/2 }  \sum_{j=1}^dh_jX_{i,j})_{1\le i\le N\land A}$ in $\R^{N\land A}$. Then applying successively the Cauchy-Schwarz inequality to the previous scalar product  in $\R^{N\land A}$ and to $\E(Z^{1/2} {Z'}^{1/2})$, where
$$
Z=\sum_{i=1}^{N\land A} \mathbf{1}_{\{\|X_i\|\le A\}}e^{\langle q|X_i\rangle} \text{ and }Z'= \sum_{i=1}^{N\land A} \mathbf{1}_{\{\|X_i\|\le A\}}e^{\langle q|X_i\rangle}(\sum_{i=j}^dh_jX_{i,j}\Big )^2,$$
we see that this forces the existence of a deterministic $c_A\in \R_+^*$ such that  $\langle  h|X_i\rangle=c_A$ almost surely for all $1\le i\le N\land A$ such that $\|X_i\|\le A$, hence there exists an affine hyperplane $H_A$ such that almost surely, if $1\le i\le N\land A$ and $\|X_i\|\le A$, we have $X_i\in H_A$. Suppose now that there exists an increasing sequence $(A_\ell)_{\ell\ge 1}$ such that for each $\ell$ the second derivative of $\wt P_{A_\ell}$ is not positive definite at some point. Let $H=\bigcup_{L\ge 1}\bigcap_{\ell\ge L} H_{A_\ell}$. By construction, we have $X_i\in H$ for all $1\le i\le N$ almost surely, and the set $H$ is a non decreasing union  of  strict affine subsets of $\R^d$, so it is a strict affine subset of $\R^d$, which contradicts \eqref{pP}. 

\medskip

The previous lines imply that $\widering I_A\neq\emptyset$ since it must contain $\nabla \wt P(J_A)$, which is an open set.  
Now, we use the general facts about the concave function $\widetilde P_A^*$: since $\wt P_A$ is differentiable, the domain of $\wt P_A^*$, i.e. $\{\alpha\in\R^d: \wt P_A^*(\alpha)>-\infty\}$, is included in the closure of the range of $\nabla \wt P_A$, and its interior is included in the image of $\nabla \wt P_A$ (see \cite[Sec. 24, p. 227]{Roc}). 

Suppose that $\alpha\in \widering I_A$ and there exists a sequence $(q_n)_{n\ge 1}$ of vectors in $\R^d$ such that $\alpha=\lim_{n\to\infty}\nabla \wt P_A(q_n)$ and $\wt P_A^*(\nabla \wt P_A(q_n))\le 0$ for all $n\ge 1$.  Then the concave function $\wt P_A^*$ being continuous at $\alpha$, we have $\wt P_A^*(\alpha)=0$. Since $\nabla \wt P_A(0)\in \widering I$ and $\wt P^*(\nabla \wt P_A(0))= \wt P_A(0)>0$ (or more generally since $J_A\neq\emptyset$), the concavity of $\wt P_A^*$  implies that $\wt P_A^*$ takes negative values over $\widering I_A$, which is excluded by definition of $I_A$.  Thus $\widering I_A\subset \nabla \wt P_A(J)$. Consequently, $ \widering I_A={\nabla \wt P_A(J)} $ and $I_A=\overline {\nabla \wt P_A(J)}$. 
\end{proof}

\begin{rem}\label{sc}
Two successive applications of the Cauchy-Schwarz inequality as above would get the strict convexity of $\wt P$ when $\dom \wt P$ is neither empty nor a singleton.
\end{rem}

Given a family of functions $f_\lambda:\R^d\to \R\cup\{\infty\}$, $\lambda>\lambda_0$, and $f:\R^d\to \R\cup\{\infty\}$, following \cite{Wi}, we say that $f_\lambda$ converges  infimally to $f$ as $\lambda\to\infty $ if 
$$
\lim_{\delta\to 0}\liminf_{\lambda\to\infty} \inf\{f_\lambda(q'):|q-q'|<\delta\}=\lim_{\delta\to 0}\limsup_{\lambda\to\infty} \inf\{f_\lambda(q'):|q-q'|<\delta\}=f(q).
$$
\begin{prop}\label{infim}
The family $(\wt P_A)_{A\ge A_0}$ converges infimally to $\wt P$ as $A\to\infty$. 
\end{prop}
\begin{proof} This is due to the fact that the function $\wt P$ is lower semi-continuous as logarithmic generating function, and it is the non-decreasing limit of the continuous functions $\wt P_A$ as $A\to\infty$: Let $q\in\R^d$. Suppose that $\wt P(q)<\infty$. Fix $\epsilon>0$.  Let $A$ be big enough so that $\wt P_A(q)\ge \wt P (q)-\epsilon/2$, and let $\delta>0$ such that $|q'-q|\le \delta$ implies $\wt P_A(q')\ge \wt P_A(q)-\epsilon/2$. Then, for  $A'\ge A$ and $|q'-q|\le \delta$ we have $\wt P_{A'}(q')\ge \wt P_A(q')\ge  \wt P (q)-\epsilon$. This yields $\lim_{\delta\to 0}\liminf_{A\to\infty} \inf\{\wt P_A(q'):|q-q'|<\delta\}\ge \wt P(q)$. The inequality $\lim_{\delta\to 0}\limsup_{A\to\infty} \inf\{\wt P_A(q'):|q-q'|<\delta\}\le \wt P(q)
$ just comes from the simple convergence of $\wt P_A$ to $\wt P$. If $\wt P(q)=\infty$, since $\wt P_A(q)$ converges to $\infty$ as $A\to\infty$ and the $\wt P_A$ are continuous, for any $p>0$ there exists $A\ge A_0$ and a neighborhood $U_A$ of $q$  such that $\wt P(q')\ge p$ for all $q'\in U_A$. Consequently, since $(\wt P_{A'})_{A'\ge A}$ is non decreasing, we have $\lim_{\delta\to 0}\liminf_{\lambda\to\infty} \inf\{f_\lambda(q'):|q-q'|<\delta\}\ge p$, where $p$ is arbitrary.
\end{proof}

\begin{cor}\label{convPA*}
 The non decreasing family $(\wt P_A^*)_{A\ge A_0}$ converges simply to $\wt P^*$ over $\mathrm{int} (\dom(-\wt P^*))$.
\end{cor}
\begin{proof}
Due to Proposition~\ref{infim} and \cite[Theorem 6.2]{Wi}, since the functions $\wt P_A$, $A\ge A_0$ and $\wt P$ are proper and closed convex functions, $(-\wt P^*_A)_{A\ge A_0}$ converges infimally to $-\wt P^*$. Moreover,  $(-\wt P^*_A)_{A\ge A_0}$  is non increasing and lower bounded by $-\wt P^*$, so it converges to a convex function, say $g$. The infimal convergence of $-\wt P^*_A$ to $-\wt P^*$ imposes that the convex set $\dom (g)$ contains the convex set $\dom (-\wt P^*)$ in its closure. Consequently it contains $\mathrm{int} (\dom(-\wt P^*))$ in its interior. Now, over each compact ball $B$ of positive radius contained in $\mathrm{int} (\dom(-\wt P^*))$, the convergence of the family of convex functions $(-\wt P^*_A)_{A\ge A_0}$ is uniform, so $f$ must be equal to $-\wt P^*$ over $B$, for otherwise the infimal convergence of  $(-\wt P^*_A)_{A\ge A_0}$ to $-\wt P^*$ would be clearly violated in the interior of $B$.  Covering $\mathrm{int} (\dom(-\wt P^*))$ with such balls  yields the result.

\end{proof}

\begin{cor}\label{approxi}
For all integers $A\ge A_0$ we have $I=\overline{\bigcup\uparrow_{A'\ge A}I_{A'}}$. We can choose an increasing sequence of integers $A_j\ge A_0$, $j\ge 1$, and a sequence of finite sets $(D_j)_{j\ge 1}$ such that   for all $j\ge 1$ we have $D_j\subset  J_{A_j}$, $\#D_j\ge 2$, and for all $\alpha\in I\cap B(0,j)$, there exists $q_j\in D_j$ such that  $\|\wt P_{A_j} (q_j)-\alpha\|\le j^{-1}$ and $| \wt P_{A_j}^*(\nabla\wt P_{A_j}(q_j))-\wt P^*(\alpha)|\le j^{-1}$, where $B(0,j)$ stands for the closed Euclidean ball of radius $j$ centered at $0$ in $\R^d$. 
\end{cor}

\begin{proof} Let  $\alpha\in I$.  Due to Proposition~\ref{wtPA}, since $\wt P^*\ge \wt P_A^*$ for all $A$, there exists $\beta\in \widering I$ such that $\wt P^*(\beta)>0$. Since $\wt P^*$ is upper semi-continuous and concave, it is continuous over the segment $[\alpha,\beta]$, hence $\alpha$ is limit of points $\beta_n\in I$, $n\ge 1$,  such that $\wt P^*(\beta_n)>0$ and $\lim_{n\to\infty} \wt P^*(\beta_n) =\wt P^*(\alpha)$.  Due to Corollary~\ref{convPA*}, for each such $\beta_n$ we have $\lim_{A\to\infty}\wt P^*_A(\beta_n)=\wt P^*(\beta_n)$, hence  for $A$ large enough $\wt P^*_{A}(\beta_n)>0$ and  $\beta_n\in I_{A}$. This implies $\alpha\in \overline{\bigcup_{A'\ge A}I_{A'}}$. Moreover, by Proposition~\ref{wtPA}(2), for each $\beta_n$, for $A$ large enough, we have $\beta_n\in \overline {\nabla\wt P_A (J_A)}$. In addition, we can require that $\beta_n$ be the limit of $\nabla \wt P_A(q_k)$, with $q_k\in J_A$ such that $\lim_{k\to\infty} \wt P_A^*(\nabla \wt P_A(q_k))=\wt P^*_{A}(\beta_n)$, by the same argument as that used to build $(\beta_n)_{n\ge 1}$ associated with $\alpha$, since $\widering I_A=\nabla\wt P_A(J_A)$. By construction of $\beta_n$, it follows that we can find an increasing sequence of integers $(A_j)_{j\ge 1}$ and for each $j\ge 1$ a vector $q_j\in J_{A_j}$ such that $\|\nabla\wt P_{A_j}(q_j)-\alpha\|\le 1/2j$, and $|\wt P_{A_j}^*(\nabla \wt P_{A_j}(q_j))-\wt P^*(\alpha)|\le 1/2j$. 

Now fix $j\ge 1$, and let $G_j=\{(\alpha,\wt P^*(\alpha)):\alpha\in B(0,j)\cap I\}$. Cover the compact set $\overline {G_j}$ by a finite collection of balls $\{B_i\}_{i\in\mathcal I_j}$ of radius $1/2j$ centered at $(\alpha_i, p^*_i) \in \overline {G_j}$ in $\R^d\times \R$. Thanks to the above lines and the finiteness of $\mathcal I_j$,  we can choose the sequence $(A_j)_{j\ge 1}$ such that for each $i\in\mathcal I_j$ there is $q_{i,j}\in J_{A_j}$ such that $\|\nabla\wt P_{A_j}(q_{i,j})-\alpha_{i}\|\le 1/2j$, and $|\wt P_{A_j}^*(\nabla \wt P_{A_j}(q_{i,j}))-p_i^*|\le 1/2j$. Setting $D_j=\{q_{i,j}: i\in\mathcal I_j\}$ we are done (adding one point in $D_j$ if necessary to make sure that $\#D_j\ge 2$). 

\end{proof}

\subsection{Upper bound for the Hausdorff  dimensions of the sets $E_X(K)$}\label{upb}

Let us consider the pressure like function  
\begin{equation}\label{pressure}
P(q)=\displaystyle\limsup_{n\rightarrow\infty}\frac{1}{n}\log\Big (\displaystyle\sum_{u\in \TT_n}  \exp ({\langle q|S_{n}X(u)\rangle})\Big) \quad (q\in\R^d).
\end{equation}
It is the function $\overline P$ defined in Theorem~\ref{thm-1.3}(2). The first part of the following proposition is inspired by the approach used in the literature to study the free energy function of Mandelbrot measures. 
\begin{prop}\label{upb1}
\begin{enumerate}
\item For all $q\in\R^d$, with probability 1, $P(q)\le \inf\{\theta^{-1} \wt P(\theta q):0< \theta\le 1\}\le \wt P(q)$. 

\item With probability 1, $\mathrm {cl}(P)(q)\le \widetilde P(q)$ for all $q\in \dom \wt P$, where $\mathrm {cl}(P)$ stands for the closure of $P$.
\end{enumerate}
\end{prop}
\begin{proof}
The  functions $\mathrm{cl}(P)$ and $\widetilde{P}$ being convex and lower semi-continuous, we only need to prove (1). 

Fix  $q\in\R^d$. If $\infty=\inf\{\theta^{-1} \wt P(\theta q):0< \theta\le 1\}\le \wt P(q)$, there is nothing to prove.  Otherwise, let $\theta\in (0,1]$ such that $\wt P(\theta q)<\infty$.  For $s>\widetilde{P}(\theta q)$ we have 
\begin{equation}\label{Pntheta}
\E(\displaystyle\sum_{n\geq1}e^{-ns}\displaystyle\sum_{u\in \TT_n}  \exp ({\langle \theta q|S_{n}X(u)\rangle})= \displaystyle\sum_{n\geq1}e^{-ns} \E (\sum_{i=1}^{N}   \exp ({\langle \theta q|X_i\rangle}))^n =\displaystyle\sum_{n\geq1}e^{n(\widetilde{P}(\theta q)-s)}.
\end{equation}
Consequently, $\displaystyle\sum_{n\geq1}e^{-ns}\displaystyle\sum_{u\in \TT_n}  \exp ({\langle \theta q|S_{n}X(u)\rangle})<\infty $, and  $\displaystyle\sum_{u\in \TT_n} \exp ({\langle \theta q|S_{n}X(u)\rangle})=O(e^{ns})$  almost surely. On the other hand,
\begin{eqnarray*}
\displaystyle\sum_{u\in \TT_n}  \exp ({\langle q|S_{n}X(u)\rangle})\Big )\le \Big (\displaystyle\sum_{u\in \TT_n}  \exp ({\langle  \theta q|S_{n}X(u)\rangle})\Big )^{1/\theta}=O(e^{n\theta ^{-1}s}),
\end{eqnarray*}
hence $P(q)\le \theta ^{-1} s$. Since $s>\wt P(\theta q)$ is arbitrary, we get $P(q)\le \theta^{-1}\wt P(\theta q)$ almost surely, and taking the infimum over $\theta$ yields the conclusion. 
\end{proof}

For $\alpha\in\R^d$ let 
$$
\widehat E_X(\alpha)=\Big\{t\in\partial\TT: \alpha\in \bigcap_{N\ge 1}\overline{ \Big \{\frac{S_nX(t)}{n}:n\ge N\Big\}}\Big\} .
$$ 

\begin{prop}\label{uHD}
With probability 1, for all  $\alpha\in \R^d$, $\dim \widehat E_X(\alpha)\le P^*(\alpha)$, a negative dimension meaning that $\widehat E_X(\alpha)$ is empty. 
\end{prop}
\begin{proof}
We have 
\begin{eqnarray*}
\widehat E_X(\alpha)&=&\bigcap_{\epsilon>0}\bigcap_{N\ge 1}\bigcup_{n\ge N}\{t\in\partial\TT: \|S_nX(t)-n\alpha\|\le n\epsilon\}\\
&\subset& \bigcap_{q\in\R^d}\bigcap_{\epsilon>0}\bigcap_{N\ge 1}\bigcup_{n\ge N}\{t\in\partial\TT: |\langle q|S_nX(t)-n\alpha\rangle|\le n\|q\|\epsilon\}.
\end{eqnarray*}
Fix $q\in\mathbb{R}^d$ and $\epsilon>0$. For $N\ge 1$, the set $E(q,N,\epsilon,\alpha)=\bigcup_{n\ge N}\{t\in\partial\TT: |\langle q|S_nX(t)-n\alpha\rangle|\le n\|q\|\epsilon\}$ is covered by the union of those $[u]$ such that $u\in \bigcup_{n\ge N}\TT_n$ and $\langle q|S_nX(u)-n\alpha\rangle +n\|q\|\epsilon\ge 0$. Consequently, for $s\ge 0$,
$$
\mathcal H^s_{e^{-N}}(E(q,N,\epsilon,\alpha))\le \sum_{n\ge N}\sum_{u\in \TT_n} e^{-ns} \exp(\langle q|S_nX(u)-n\alpha\rangle +n\|q\|\epsilon).
$$
Thus, if $\eta>0$ and $s>P(q)+\eta-\langle q|\alpha\rangle+\|q\|\epsilon$, by definition of $P(q)$, for $N$ large enough  we have 
$$
\mathcal H^s_{e^{-N}}(E(q,N,\epsilon,\alpha))\le \sum_{n\ge N}e^{-n\eta/2}.
$$
This yields $\mathcal H^s(E(q,N,\epsilon,\alpha))=0$, hence $\dim E(q,N,\epsilon,\alpha)\le s$. Since this holds for all $\eta>0$, we get $\dim E(q,N,\epsilon,\alpha)\le P(q)-\langle q|\alpha\rangle+\|q\|\epsilon$. It follows that $\dim \widehat E_X(\alpha)\le \inf_{q\in\R^d} \inf_{\epsilon>0}P(q)-\langle q|\alpha\rangle+\|q\|\epsilon=P^*(\alpha)$. If $P^*(\alpha)<0$, we necessarily have $ \widehat E_X(\alpha)=\emptyset$. 
\end{proof}

\begin{cor}
With probability 1, for all subsets $K$ of $\R^d$, we have $E_X(K)= \emptyset$  if $K\not \subset I$, and $
\dim E_X(K)\le \inf_{\alpha\in K}\widetilde P^*(\alpha)$ otherwise. 
\end{cor}
\begin{proof}
We have $\displaystyle E_X(K)\subset \bigcap_{\alpha\in K}\widehat E_X(\alpha)$. Consequently, due to the previous propositions,  if $K\not\subset I$,   $E_X(K)=\emptyset$, otherwise, since $P^*=(\mathrm{cl}(P))^*$ (\cite[Theorem 12.2]{Roc}), we have $\dim E_X(K)\le \inf_{\alpha\in K}\dim \widehat E_X(\alpha)\le  \inf_{\alpha\in K} \mathrm{cl}(P)^*(\alpha)\le  \inf_{\alpha\in K}\widetilde P^*(\alpha)$.
\end{proof}

\subsection{Upper bounds for the packing dimension of the sets $E_X(K)$}\label{uppack}
For reader's convenience we include the following proposition, which is standard.
\begin{prop}
With probability 1, $\Dim \partial\TT\le \wt P(0)$.
\end{prop}
\begin{proof}
If $\wt P(0)=\infty$ we have nothing to prove. Suppose that $\wt P(0)<\infty$. Let $u\in\TT$ and $N\ge |u|$. Due to the ultra-metric nature of the distance $d_1$, an $e^{-N}$ packing of $[u]$ is a countable union of disjoint  cylinders $[v_j]$, $j\in \mathcal J$,  included in $[u]$. Let $\epsilon>0$.   We thus have 
$$
\sup_{\{v_j\}_{j\in\mathcal J}}\sum_{j\in\mathcal J}|[v_j]|^{\wt P(0) +\epsilon}\le \sum_{n\ge N} \sum_{v\in \TT_n} e^{-n( \wt P(0) +\epsilon)},
$$
where $\{v_j\}_{j\in\mathcal J}$ runs over the $e^{-N}$ packing of $[u]$,  and the expectation of the right hand side equals  $\sum_{n\ge N}e^{-n\epsilon}$, so that the right hand side converges to 0  exponentially fast almost surely. This yields $\overline P^{\wt P(0)+\epsilon}([u])=0$ for all $u$ almost surely. The set of cylinders being countable, we get $\overline {\mathcal P}^{\wt P(0)+\epsilon}([u])=0$ almost surely for all $u\in \TT$, and consequently $\mathcal P^{\wt P(0)+\epsilon}(\partial \TT)=0$ almost surely. Finally, $\Dim\partial \TT\le \wt P(0)+\epsilon$ for all $\epsilon>0$. 
\end{proof}

\begin{prop}
With probability 1, if $K\subset I$ is a non-empty compact subset of $I$, then  $\Dim E_X( K)\le \sup_{\alpha\in K} \wt P^*(\alpha)$.
\end{prop}
\begin{proof}
At first we note that setting for $\alpha\in I$, $n\ge 1$ and $\epsilon>0$
$$
f(n,\alpha,\epsilon)=\frac{\log \#\{v\in \TT_n:\, \|S_n(v)-n\alpha\|\le n\epsilon\}}{n},
$$
we have:
\begin{equation}\label{LD'} 
\text{Almost surely, } \forall\, \alpha\in I,\ \lim_{\epsilon\to 0}\limsup_{n\to\infty} f(n,\alpha,\epsilon)\le \wt P^*(\alpha).
\end{equation}
Indeed,  we clearly have almost surely for all $\alpha\in I$, $q\in\R^d$, $n\ge 1$ and $\epsilon>0$
$$
\#\{v\in \TT_n:\, \|S_n(v)-n\alpha\|\le n\epsilon\}\le \sum_{u\in \TT_n} \exp(\langle q|S_nX(u)-n\alpha\rangle +n\|q\|\epsilon), 
$$
hence $
\limsup_{n\to\infty} f(n,\alpha,\epsilon)\le P(q)-\langle q|\alpha\rangle + \|q\|\epsilon
$, by definition of $P$. Consequently, $\lim_{\epsilon\to 0}\limsup_{n\to\infty} f(n,\alpha,\epsilon)\le P(q)-\langle q|\alpha\rangle$ for all $q\in\R^d$, hence the desired upper bound since $P^*\le \widetilde P^*$. 

\medskip

Now let $K$ be a non-empty compact subset of $I$. Let $\eta>0$. For each $\alpha\in K$, let $\epsilon_\alpha>0$ and $n_\alpha\ge 1$ such that $f(n,\alpha,\epsilon_\alpha)\le \wt P^*(\alpha)+\eta$ for all $n\ge n_\alpha$. Since $K$ is compact, we can find $\alpha_1,\ldots,\alpha_\ell$ in $K$ such that $K\subset \bigcup_{i=1}^\ell B(\alpha_i,\epsilon_i)$, where $\epsilon_i$ stands for $\epsilon_{\alpha_i}$. Now we notice that 
$
E_X(K)\subset\bigcup_{N\ge 1}F_N 
$, 
where 
$$
F_{N}=\bigcap_{n\ge N} \bigcup_{i=1}^{\ell}\{t: \|S_n(t)-n\alpha_i\|\le n\epsilon_i\}.
$$
Fix $N\ge 1$. Let $u\in\TT$ with $|u|\ge N$. For each $N'\ge \max(|u|, n_{\alpha_1},\ldots,n_{\alpha_\ell})$, any $e^{-N'}$-packing of $[u]\cap F_{N}$ is included in $\bigcup_{n\ge N'} \bigcup_{i=1}^{\ell}\bigcup_{\substack{v\in \TT_n,\\
 \|S_n(v)-n\alpha_i\|\le n\epsilon}}[v]$. Consequently, setting $s=\sup_{\alpha\in K}\wt P^*(\alpha)+2\eta$, we have 
\begin{eqnarray*}
&&\sup_{\substack{\{[v_j]\}_{j\in\mathcal J}, \, \text{$e^{-N'}$-packing  of $[u]\cap F_{N}$} }}\sum_{j\in\mathcal{J}}|[v_j]|^s\le \sum_{n\ge N'}\sum_{i=1}^\ell \sum_{\substack{v\in \TT_n,\\
 \|S_n(v)-n\alpha_i\|\le n\epsilon_i}} |[v]|^s\\
 &\le&\sum_{n\ge N'}\sum_{i=1}^\ell |[v]|^sf(n,\alpha,\epsilon_i) \le \sum_{n\ge N'}\sum_{i=1}^\ell e^{-n (s-\wt P^*(\alpha_i)-\eta)}=\ell \sum_{n\ge N'}e^{-n\eta}.
 \end{eqnarray*}
It follows that $\overline{\mathcal P}^s([u]\cap F_{N})=0$. This holds for all $u$ of generation bigger than $N$, so $\mathcal P^s(F_{N})=0$. Since this is true for all $N\ge 1$, we get $\mathcal P^s(E_X(K))=0$ and $\Dim E_X(K)\le  \sup_{\alpha\in K}\wt P^*(\alpha)+2\eta$, for all $\eta>0$. 

\end{proof}
                                                    %%%%%%%%%%%%%%%%%%%%%%%%%%%%%%%%%%%%%%%%%%%%%%%%%%%%%%%%%%%%%%%%%%%%%%%%%
%%%%%%%%%%%%%%%%%%%%%%%%%%%%%%%%%%%%%%%%%%%%%%%%%%%%%%%%%%%%%%%%%%%%%%%%%%%%%%%%%%%%%%%%%%%%%%%%%%%%%%%%%%%%%%%%%%%%%%%%%%%%%%%%%%%%%%%
%%%%%%%%%%%%%%%%%%%%%%%%%%%%%%%%%%%%%%%%%%%%%%%%%%%%%%%%%%%%%%%%%%%%%%%%%%%%%%%%%%%%%%%%%%%%%%%%%%%%%%%%%%%%%%%%%%%%%%%%%%%%%%%%%%%%%%%
 %%%%%%%%%%%%%%%%%%%%%%%%%%%%%%%%%%%%%%%%%%%%%%%%%%%%%%%%%%%%%%%%%%%%%%%%%
\subsection{Lower bounds for the dimensions of the sets $E_X(K)$ via inhomogeneous Mandelbrot measures}\label{LB}

\subsubsection{A family of inhomogeneous Mandelbrot martingales} \label{sec-2.2.1}

\subsubsection*{The set of parameters} Choose $(D_j)_{j\ge 1}$ as in Corollary~\ref{approxi}. Let $(L_j)_{j\ge 0}$ be a sequence of integers such that $L_0=0$, and that we will specify below. Then let $(M_{j})_{j\geq 0}$ be the  increasing sequence defined as 
\begin{equation}\label{Mj}
M_{j}=\displaystyle\sum_{k=1}^{j}L_{k} \text{ for all $j\ge 0$}.
\end{equation}
For  $n\in \N$, let   $j_{n}$ denote  the unique integer  satisfying 
  $$M_{j_{n}-1}+1\le  n \le  M_{j_{n}}.$$

Let
\begin{equation}\label{setJ}
{\mathcal {J}}=\{ \varrho=(q_k)_{k\geq 1} : \forall j \geq 1,  q_{M_{j-1} +1} =  q_{M_{j-1} +2} = \cdots =  q_{M_{j }} \in D_{j} \}.
\end{equation}
Since each $D_j$ is finite for all $j\ge 1$, the set $\mathcal J $ is compact for the metric 
$$
d(\varrho=(q_k)_{k \geq 1} ,\varrho'= (q'_k)_{k \geq 1}) = \displaystyle \sum_{k \geq 1} 2^{-k} \frac{ |q_k- q'_k |}{1+|q_k- q'_k |}.
$$

For $\varrho=(q_k)_{k\ge 1} \in \mathcal{J}$ and $n\ge 1$ we will denote by $\varrho_{|n}$ the sequence $(q_k)_{1\le k\le n}$. 

We denote 
$$
d_j= \prod_{j'=1}^{j}\# D_{j'}.
$$
By construction, for all $n\ge 1$ we have 
$$
\#\{\varrho_{|n}: \varrho\in\mathcal J\}=d_{j_n}.
$$

We are going to build on $\partial \TT$ a family of measures $\{\mu_\varrho\}_{\varrho\in\mathcal J}$. All these measures will be positive conditionally on the event $\Omega_{A_0}$ of non extinction of the subset $\partial\widetilde \TT$ of $\partial\TT$ made of the infinite branches $u$ such that $u_{n}\le N_{u_{|n-1}}\land A_{j_n}$ and $\|X_{u_1\ldots u_n}\|\le A_{j_n}$ for all $n\ge 1$.  Moreover, on this event these measures will lead to the expected lower bounds for the Hausdorff and packing dimensions of the sets $E_X(K)$. Moreover, by construction $\partial\widetilde \TT$ contains the set of  the infinite branches $u$ such that $u_{n}\le N_{u_{|n-1}}\land A_{0}$ and $\|X_{u_1\ldots u_n}\|\le A_{0}$ for all $n\ge 1$, that is the boundary of a Galton-Watson tree with offspring distribution that of $N_{A_0}$, whose extinction probability  tends to 0 as $A_0$ tends to $\infty$ by Proposition~\ref{extinction}. Consequently  the probability of the event $\Omega_{A_0}$ tends to 1 as $A_0$ tends to $\infty$, and this will give the conclusion.

\subsubsection*{ Inhomogeneous Mandelbrot martingales indexed by $\mathcal J$}

For $n\ge 1$ we set ${\mathcal{F}}_{n}=\sigma \Big ((N_{u},X_{u1},X_{u2},\ldots): u\in \bigcup_{k=0}^{n-1} \N_+^{k}\Big )$ and $\mathcal G_n=\sigma \Big ((N_{u},X_{u1},X_{u2},\ldots): u\in \bigcup_{k\ge n} \N_+^{k}\Big )$.   We also denote by ${\mathcal{F}}_{0}$ the trivial $\sigma$-field.

\medskip

For $u\in \bigcup_{n\ge 0} \N_+^n$, $1\le i\le N_u$  and $\varrho=(q_k)_{k\ge 1} \in {\mathcal J}$ we define 
\begin{eqnarray*}
W_{\varrho,ui} =  \mathbf{1}_{\{1\le i\le N_u\land A_{ j_{|u|+1}}\}}  \mathbf{1}_{\{\|X_{ui}\|\le A_{j_{|u|+1}}\}}  \exp \big (\langle q_{|u|+1}|X_{ui}\rangle-\wt P_{A_{j_{|u|+1}}}(q_{|u|+1})\big ).
\end{eqnarray*}

\medskip

For $\varrho= (q_k)_{k\geq 1} \in {\mathcal J}$, $u\in \bigcup_{n\ge 0}\N_+^n$ and $n\ge 0$ we define
$$
Y_n(\varrho, u)= \displaystyle\sum_{v_1\cdots v_n \in \TT_n(u)} \displaystyle\prod_{k=1}^n W_{\varrho,u\cdot v_1\cdots v_k}.
$$
When $u=\emptyset $ this quantity will be denoted by $Y_n( \varrho)$, and when $n=0$, its value equals 1.

\subsubsection*{Specification of the sequence $(L_j)_{j\ge 1}$} The following requirements about the sequence $(L_j)_{j\ge 0}$ can be skipped until the proofs of the next section. 
\medskip

The functions $\widetilde P_{A_j}$, $j\ge 1$ are analytic. We denote by $H_j$ the Hessian matrix of $\wt  P_{A_j}$. For each $j\ge 1$, 
\begin{equation}\label{mj}
m_{1,j}=\displaystyle\sup_{t\in [0,1]}\sup_{v\in \mathbb S^{d-1}} \sup_{1\le j'\le j}\sup_{q\in D_{j'}} {}^tv H_{j'}(q+tv)v
\end{equation}
and 
\begin{equation}\label{tmj}
m_{2,j}=\displaystyle\sup_{t\in [0,1]}\sup_{p\in [1, 2]}\sup_{1\le j'\le j} \sup_{q\in D_{j'}} {}^tq H_{j'}(q+t(p-1)q)q
\end{equation}
are finite. Then let 
\begin{equation}\label{mjj}
m_j=\max (m_{1,j},m_{2,j}).
\end{equation}
\medskip

Next we notice that  due to \eqref{pP}, two applications of the Cauchy-Schwarz inequality as in the proof of Proposition~\ref{detI} show that, by choosing $A_0$ larger if necessary, for $j\ge 1$ we have 
\begin{equation}\label{sing}
c_j(q,q')=\mathbb{E}\Big (\sum_{i=1}^{N\land A_j}\mathbf{1}_{\{\|X_i\|\le A_j\}}\exp\Big [\frac{1}{2} (\langle q|X_i\rangle -\wt P_{A_j}(q))\Big ]\exp\Big [\frac{1}{2} (\langle q'|X_i\rangle-\wt P_{A_j}(q') )\Big ]\Big)<1
\end{equation}
if $q\neq q'\in\R^d$. For $j\ge 1$ let 
\begin{equation}\label{sing'}
c_{j}=\sup_{1\le j'\le j}\sup_{q\neq q'\in D_{j'}} c_{j'}(q,q')<1.
\end{equation}

Let $(\gamma_j)_{j\ge 1}\in (0,1]^{\N_+}$ be a positive  sequence converging to 0 such that 
\begin{equation}\label{gammajcj}
\lim_{j\to\infty}\gamma_jm_j=0\quad\text{and}\quad \gamma_{j}^2m_{j}=o(\log c_{j}) \text{ as } j\to\infty. 
\end{equation}

\medskip

For $j\ge 1$ we set 
$$
\varphi_j:(q,p)\in \R^d\times\R \mapsto \wt P_{A_j}(pq)-p\wt P_{A_j}(q).
$$ 
For each $q\in J_{A_j}$  there exists  a real number  $1 < p_q< 2 $ such that 
$\displaystyle\varphi_j (p,q) <  0$ for all $p\in (1,p_q)$. Indeed, since   $\widetilde{P}_{A_j}^{*}(\nabla\widetilde{P}_{A_j}(q)) > 0$ one has   $\frac{\partial \phi_j}{\partial p}  (q, 1^+) < 0$. 

\medskip

For all $j\ge 1$ we set 
$$
p_j= \inf_{1\le j'\le j}\inf_{q\in D_{j'}} p_q, \quad e_j=\exp\Big (\frac{p_j}{p_j-1}\Big )\quad\text{and}\quad a_j=\sup_{1\le j'\le j}\sup_{q\in D_{j'}} \displaystyle\varphi (p_{j},q).
$$
 By construction, we have $a_j<0$. Then let 
\begin{equation}\label{rj}
s_j=\max\left \{\frac{ \|S_{A_j}(q)\|_{p_j}}{ \|S_{A_j}(q)\|_{1}}: q\in D_{j}\right\}\quad \text{and}\quad r_j=\frac{a_j}{p_j}.
\end{equation}
\medskip

Recall that $L_0=0$. For $j\ge 2$ choose inductively a positive integer $L_{j-1}>L_{j-2}$ big enough so that 

\begin{equation}\label{control}
\frac{d_j e_{j}  s_{j}}{ 1-\exp( r_{j})}\exp (L_{j-1} r_{j})\le j^{-2}
\end{equation}
(which is possible since $r_{j}<0$), 
\begin{equation}\label{control2}
\frac{d_j e_{j} s_{j}}{ 1-\exp( r_{j})}+\frac{d_{j+1} e_{j+1} s_{j+1}}{ 1-\exp( r_{j+1})}\le C_0\exp (\sqrt{L_{j-1}} \gamma_{j}^2 m_{j})
\end{equation}
with $C_0=\displaystyle \frac{ e_1s_{1}}{ 1-\exp( r_{1})}+\frac{2 e_2 s_{2}}{1-\exp( r_{2})}$,
\begin{equation}\label{control0}
L_{j-1}\ge \max (\gamma_{j}^{-8}m_j^{-4},5\log (d_j)),
\end{equation}
\begin{equation}\label{sing2}
(d_{j})^2 c_{j}^{L_{j-1}/2}\le j^{-2}
\end{equation}
and 
\begin{equation}\label{control2bis}
\Big (\sum_{k=1}^{j-1}L_k\Big )\max_{1\le j'\le j-1}\wt m_{j'} \le j^{-2} L_j \widehat m_j,
\end{equation}
where 
\begin{align*}
\wt m_j&=  \max(1,\max\{ \|\nabla \widetilde P_{A_{j}}(q)\|, \, \widetilde P_{A_{j}}^*(\nabla \widetilde P_{A_{j}}(q)): q\in D_{j}\})\\
\widehat m_j&=\min (1, \{ \|\nabla \widetilde P_{A_j}(q)\|, \,  \widetilde P_{A_{j}}^*(\nabla \widetilde P_{A_{j}}(q)): q\in D_{j}\}).
\end{align*}

\subsubsection{ A family of inhomogeneous Mandelbrot measures}

If $u\in \bigcup_{n\ge 0}\N_+^n$, let $\partial \widetilde \TT(u)=\{v\in \TT(u):\, \forall\, k\ge 1, \, 1\le v_k\le  N_{uv_{|k-1}}\land A(j_{|u|+k}) \text{ and } \|X_{uv_{|k}}\|\le A(j_{|u|+k})\}$; we have  $\partial \widetilde \TT(\emptyset)=\partial \widetilde \TT$, the subset of $\partial \TT$ defined at the beginning of  the previous section. 

\begin{prop}\label{pro-2.4} $\ $
\begin{enumerate}
\item For all $u\in \bigcup_{n\ge 0}\N_+^n$, the sequence of continuous functions $Y_n(\cdot,u)$ converges uniformly on $\mathcal J$, almost surely and in $L^1$ norm. Moreover the events $\{Y(\cdot,u)>0\}$ and  $\partial \widetilde \TT(u)\neq\emptyset$ differ by a set of probability 0. 

\item With probability 1, for all $\varrho\in \mathcal J$, the mapping 
\begin{equation}\label{defmurho}
\mu_\varrho ([u])= \Big (\displaystyle\prod_{k=1}^n W_{\varrho,u_1\cdots u_k} \Big )   Y(\varrho,u), \quad u\in \TT_n.
\end{equation}
defines  a  measure supported on  $\partial\TT$. Moreover, on $\{\partial \wt \TT\neq \emptyset\}$  this measure is supported on $ \partial \widetilde\TT$.
\medskip

\item With probability 1, for all $(\varrho,\varrho')\in \mathcal J^2$, if the measures $\mu_\varrho$ and $\mu_{\varrho'}$ do not vanish, they are absolutely continuous with respect to each other or mutually singular according to whether  $\varrho$ and $\varrho'$ coincide ultimately or not. 

\end{enumerate}
\end{prop}

The measures $\mu_\varrho$ will be used to approximate from below the Hausdorff dimensions of the sets $E_X(K)$ in the next section. 
\begin{lem}\cite{vB-E}\label {lem-2.1}
 Let $(X_j)_{j\geq 1}$ be a sequence of centered independent  real valued  random variables. For every finite  $I\subset\N_+$ and $p\in (1, 2)$ 
$$\E\Big (\Big |\displaystyle\sum_{i\in I} X_{i}\Big|^p\Big ) \leq 2^{p-1}  \sum_{i\in I}  \E( \left|X_{i}\right|^p).$$  
\end{lem}
\begin{lem} \label{lem-2.2}
Let $\varrho\in\mathcal J$. Define  $Z_n(\varrho)=  Y_n(\varrho)- Y_{n-1}(\varrho)$ for $n\ge 0$. For every $p\in (1,2)$ we have 
\begin{equation} \E( |Z_n(\varrho)|^{p}) \leq 2^p \frac{\mathbb{E}(S_{A_{j_n}}(q_n)^p)}{\mathbb{E}(S_{A_{j_n}}(q_n))^p}\displaystyle\prod_{k=1}^{n-1}\exp \big( \widetilde P_{A_{j_k}} (pq_k)-p\widetilde P_{A_{j_k}} (q_k) \big ).
\end{equation}
\end{lem} 
\begin{proof}
Fix $p\in (1,2]$. By using the branching property we can write 
$$
Z_n(\varrho) = \displaystyle\sum_{u\in \TT_{n-1}} \prod_{k=1}^{n-1} W_{\varrho,u_1\cdots u_k} \Big (\displaystyle\sum_{i=1}^{N_u} W_{\varrho,ui} - 1\Big ).
$$
Let  $B_u= \displaystyle\sum_{i=1}^{N_u} W_{\varrho,ui} $. By construction, the random variables $(B_u-1)$,  ${u\in \TT_{n-1}}$, are centered and i.i.d., and  independent of $\mathcal{F}_{n-1}$ (in particular $\E(Z_n(\varrho)=0$ so that $\E(Y_n(\varrho))=\E(Y_{n-1}(\varrho))=\ldots=\E(Y_{0}(\varrho)=1$). Consequently, conditionally on $\mathcal{F}_{n-1}$, we can apply Lemma~\ref{lem-2.1} to the family $\{ B_u \prod_{k=1}^{n-1}  W_{\varrho,u_1\cdots u_k} \}_{u\in \TT_{n-1}}$. Noticing that the $B_u$, $u\in \TT_{n-1}$, have the same distribution, this yields
$$
\E(\left|Z_n(\varrho)\right|^{p}| {\mathcal{F}}_{n-1} \leq 2^{p-1}\mathbb{E}(|B-1|^p) \sum_{u\in T_{n-1}}\prod_{k=1}^{n-1} W_{\varrho,u_1\cdots u_k} ^p,
$$
where $B$ stands for any of the identically distributed variables $B_u$. Since $\mathbb{E}(B)=1$ and $B\ge 0$, convexity inequalities yield $\mathbb{E}(|B-1|^p) \le 2\mathbb{E}(B^p)$. Moreover,  $B\sim\frac{S_{A_{j_n}}(q_n)}{\mathbb{E}(S_{A_{j_n}}(q_n))}$, so $2^{p-1}\mathbb{E}(|B-1|^p) \le 2^p \frac{\mathbb{E}(S_{A_{j_n}}(q_n)^p)}{\mathbb{E}(S_{A_{j_n}}(q_n))^p}$. Also, a recursion using of the branching property and the independence of the random vectors $(N_u,X_{u1},\ldots)$ used in the constructions yields, setting 
$$
W_{q_k,i}=  \mathbf{1}_{\{1\le i\le N\land A_{j_k}\}}\mathbf{1}_{\{\|X_i\|\le A_{j_k}\}} \exp (\langle q_{k}|X_{i}\rangle -\wt P_{A_{j_k}}(q_k)):$$
$$
\mathbb{E}\Big (\sum_{u\in \TT_{n-1}}\prod_{k=1}^{n-1} W_{\varrho,u_1\cdots u_k} ^p\Big )=\prod_{k=1}^{n-1}\mathbb{E}\Big (\sum_{i=1}^N W_{q_k,i}^p\Big )=\prod_{k=1}^{n-1}\exp \big (\widetilde P_{A_{j_k}}(pq_k)-p\widetilde P_{A_{j_k}}(q_k)\big ).
$$
Collecting the previous estimates yields the conclusion.
 \end{proof}
 
 \medskip
 
 \noindent
 {\it Proof of Proposition~\ref{pro-2.4}}. (1) Recall the definitions of the paragraph of section~\ref{sec-2.2.1} in which  the parameter set $\mathcal J$ is defined.
 
 Let us first assume that $u=\emptyset$ and observe that if $n\ge 1$,  it is easily seen from its construction that $Y_n(\cdot)=Y_n(\cdot,\emptyset)$ is a continuous function, constant over the set of those $\varrho$ sharing the same $n$ first components. 
 
 For $n\ge 1$ and $\varrho\in\mathcal J$,  we have $M_{j_n-1}+1\le n\le M_{j_n}$, and Lemma~\ref{lem-2.2} applied with $p=p_{j_n}$ provides us with the inequality
\begin{multline*}
\|Y_n(\varrho)-Y_{n-1}(\varrho)\|^{p_{j_n}}_{p_{j_n}}\le 2^{p_{j_n}}\frac{\mathbb{E}(S_{A_{j_n}}(q_n)^{p_{j_n}})}{\mathbb{E}(S_{A_{j_n}}(q_n))^{p_{j_n}}}\prod_{k=1}^{n-1}\exp \big (\widetilde P_{A_{j_k}}(p_{j_n}q_k)-p_{j_n}\widetilde P_{A_{j_k}}(q_k)\big )\\
= 2^{p_{j_n}}\frac{\mathbb{E}(S_{A_{j_n}}(q_n)^{p_{j_n}})}{\mathbb{E}(S_{A_{j_n}}(q_n))^{p_{j_n}}}\prod_{k=1}^{n-1}\exp \big( \varphi_{j_k}(p_{j_n}, q_k)\big ) \le  2^{p_{j_n}}s_{j_n}^{p_{j_n}}\prod_{k=1}^{n-1} \exp \big(\sup_{q\in D_{j_k}} \varphi_{j_k}(p_{j_n}, q)\big )\\
\le  2^{p_{j_n}}s_{j_n}^{p_{j_n}} \exp ((n-1)p_{j_{n}}r_{j_n})\quad(\text{due to  \eqref{rj}}).
\end{multline*}
 Recall that by the definition of $\mathcal J$, the cardinality of $\{\varrho_{|n}: \varrho\in \mathcal{J}\}$ is equal to that of $\prod_{j=1}^{j_n}D_j$, i.e. $d_{j_n}$, and  $Y_n(\varrho)-Y_{n-1}(\varrho) $ only depends on $(q_1,\cdots,q_n)$. Consequently,
\begin{eqnarray*}
 \|\sup_{\varrho\in \mathcal J}|Y_n(\varrho)-Y_{n-1}(\varrho)|\|_1\le \sum_{\{\varrho_{|n}: \varrho\in \mathcal{J}\}}\|Y_n(\varrho)-Y_{n-1}(\varrho)\|_{p_{j_n}}\le2 d_{j_n}s_{j_n}  \exp ((n-1)r_{j_n}).
 \end{eqnarray*}
 We deduce from this that 
\begin{multline*}
 \sum_{n\ge 1}\|\sup_{\varrho\in \mathcal J}|Y_n(\varrho)-Y_{n-1}(\varrho)|\|_1\le \sum_{j\ge 1}\sum_{M_{j-1}+1\le n\le M_{j}} 2 d_js_{j}  \exp ((n-1)r_{j})\\
 \le \sum_{j\ge 1} 2 d_js_{j}  \frac{\exp (M_{j-1}r_{j})}{1-\exp (r_{j})}
\le\frac{2s_1}{1-\exp (r_{1})}+  \sum_{j\ge 2} 2 d_js_{j}  \frac{\exp (L_{j-1}r_{j})}{1-\exp (r_{j})}\\\le  \frac{2s_1}{1-\exp (r_{1})}+2\sum_{j\ge 2}j^{-2}<\infty,
 \end{multline*}
where we have used \eqref{control}.  The convergence of the above series gives the  desired uniform convergence, almost surely and in $L^1$ norm, of $Y_n$ to a function $Y$. 
 
\medskip
Now let $u\in \bigcup_{n\ge 1}\N_+^n$. By using the same calculations as above, for all $n\ge 1$ we get  
$$
\sup_{\varrho\in \mathcal J} \|Y_n(\varrho,u)-Y_{n-1}(\varrho,u)\|_{p_{j_{|u|+n}}}\le 2 s_{j_{|u|+n}}  \exp ((n-1)r_{j_{|u|+n}})
$$
and then 
$$
 \|\sup_{\varrho\in \mathcal J}|Y_n(\varrho,u)-Y_{n-1}(\varrho,u)|\|_1\le2 \,d_{|u|+n}\,s_{j_{|u|+n}}  \exp ((n-1)r_{j_{|u|+n}}).
$$
Thus
\begin{eqnarray*}
 &&\sum_{n\ge 1}\|\sup_{\varrho\in \mathcal J}|Y_n(\varrho,u)-Y_{n-1}(\varrho,u)|\|_1\\
 &\le& \sum_{ |u|+n\le M_{j_{|u|}}} 2 d_{j_{|u|}}s_{j_{|u|}}  \exp ((n-1)r_{j_{|u|}})+  \sum_{j\ge j_{|u|}+1}\sum_{M_{j-1}+1\le |u|+n\le M_{j}} 2 d_js_{j}  \exp ((n-1)r_{j})\\
 &\le & \frac{2 d_{j_{|u|}}s_{j_{|u|}}}{1-\exp(r_{j_{|u|}})}+ \sum_{j\ge j_{|u|}+1} 2 d_js_{j}  \frac{\exp ((M_{j-1}-|u|)r_{j})}{1-\exp (r_{j})}\\
 &\le&  \sum_{j_{|u|}\le j\le j_{|u|}+1} \frac{2 d_js_{j}}{1-\exp (r_{j})}+ \sum_{j\ge j_{|u|}+2} 2 d_js_{j}  \frac{\exp ((M_{j-1}-|u|)r_{j})}{1-\exp (r_{j})}.
 \end{eqnarray*}
 Consequently, since $M_{j_{|u|}-1}+1\le |u|\le M_{j_{|u|}}$, for  $j\ge j_{|u|}+2$ we have  $M_{j-1}-|u|\ge L_{j-1}$ and  
 \begin{eqnarray*}
&&\sum_{n\ge 1}\|\sup_{\varrho\in \mathcal J}|Y_n(\varrho,u)-Y_{n-1}(\varrho,u)|\|_1
\label{sharp}\\&\le& \sum_{j_{|u|}\le j\le j_{|u|}+1} \frac{2 d_js_{j}}{1-\exp (r_{j})}+ \sum_{j\ge j_{|u|}+2} 2 d_js_{j}  \frac{\exp (L_{j-1}r_{j})}{1-\exp (r_{j})}\\
\nonumber & \le& 2 C_0\exp (\sqrt{L_{j_{|u|}-1}}\gamma_{j_{|u|}}^2m_{j_{|u|}})+2 \sum_{j\ge j_{|u|}+2}j^{-2},
 \end{eqnarray*}
where we have used \eqref{control} and \eqref{control2}. This yields the desired convergence to a limit $Y(\cdot, u)$. Moreover, since $\bigcup_{k\ge 0}\N_+^k$ is countable, the convergence holds almost surely, simultaneously for all $u$. Then the proof finishes as for $u=\emptyset$.

Let us put the previous upper bound in a form that will be useful. Let $\epsilon_k=\gamma_{j_{k}}^2m_{j_{k}}$ for all $k\ge 0$. It follows from the above calculations, the fact that $Y_0(\cdot, u)=1$ for all $u\in \bigcup_{k\ge 0}\N_+^k$, and the fact that $|u|\ge L_{j_{|u|}-1}$ that there exists a constant $C_{\mathcal J}$ such that:
\begin{equation}\label{control3}
\|\sup_{\varrho\in \mathcal J}Y(\varrho, u)\|_1\le C_{\mathcal J} \exp (\epsilon_{|u|} \sqrt{L_{j_{|u|}-1}} ) \le C_{\mathcal J} \exp (\epsilon_{|u|} \sqrt {|u|}) \quad (\forall\ u\in \bigcup_{k\ge 0}\N_+^k).
\end{equation}

In fact by our choice of $(L_j)_{j\ge 0}$ we even have 
\begin{equation}\label{control33}
\E\Big (\sup_{\varrho\in \mathcal J}Y(\varrho, u)\exp (\sqrt{\log (Y(\varrho,u)+3)})\Big ) \le C_{\mathcal J} \exp (\epsilon_{|u|} \sqrt{|u|}) \quad (\forall\ u\in \bigcup_{k\ge 0}\N_+^k),
\end{equation}
with a possibly different constant $C_{\mathcal J}$. 

Indeed, setting $f(x)=xg(x)$ with $g(x)=\exp (\sqrt{\log (x+3)})$ for $x\ge 0$, we have 
$$
\E( \sup_{\varrho\in \mathcal J}f(Y(\varrho,u)))\le \liminf_{n\to\infty}\E( \sup_{\varrho\in \mathcal J}f(Y_n(\varrho,u))),
$$
and 
\begin{eqnarray*}
\E( \sup_{\varrho\in \mathcal J}f(Y_n(\varrho,u)))&\le&  \E( \sup_{\varrho\in \mathcal J}f(Y_{n-1}(\varrho,u)))+\E(\sup_{\varrho\in \mathcal J}|f(Y_n(\varrho,u))-f(Y_{n-1}(\varrho,u))|)\\
&\le&f(1)+ \sum_{k=1}^n\E(\sup_{\varrho\in \mathcal J}|f(Y_k(\varrho,u))-f(Y_{k-1}(\varrho,u))|)\\
&\le & f(1)+ \sum_{k=1}^n d_{j_{|u|+k}} \, \sup_{\varrho\in \mathcal J}\E(|f(Y_k(\varrho,u))-f(Y_{k-1}(\varrho,u))|). 
\end{eqnarray*}
Noting that $f'(x)\le 2 g(x)$ for all $x\ge 0$ and that $g$ is monotonic, we can apply the mean value theorem and  get 
$$
 \E(|f(Y_k(\varrho,u))-f(Y_{k-1}(\varrho,u))|\le 2 \E\big (|Y_k(\varrho,u)-Y_{k-1}(\varrho,u)|  \max (g(Y_k(\varrho,u)),g(Y_{k-1}(\varrho,u)))\big ).
 $$
Now, H\"older's inequality gives 
$$
 \E(|f(Y_k(\varrho,u))-f(Y_{k-1}(\varrho,u))|\le M(\varrho, u) \|Y_k(\varrho,u)-Y_{k-1}(\varrho,u)\|_{p_{j_{|u|+k}}},
 $$
where 
$$
M(\varrho, u)=2\|\max (g(Y_k(\varrho,u)),g(Y_{k-1}(\varrho,u)))\|_{\frac{p_{j_{|u|+k}}}{p_{j_{|u|+k}}-1}}.
$$
We notice that if $x\ge \exp \Big (\frac{p_{j_{|u|+k}}^2}{(p_{j_{|u|+k}}-1)^2}\Big )$, then $g(x)^{\frac{p_{j_{|u|+k}}}{p_{j_{|u|+k}}-1}}\le x$. Since $\E(Y_k(\varrho,u))=\E(Y_{k-1}(\varrho,u))=1$, this implies that $
M(\varrho, u)\le 2 \Big (\exp \Big (\frac{p_{j_{|u|+k}}}{p_{j_{|u|+k}}-1}\Big )+1\Big ).
$
Finally,  recalling that $e_{j_{|u|+k}}=\exp \Big (\frac{p_{j_{|u|+k}}}{p_{j_{|u|+k}}-1}\Big )$, 
$$
E( \sup_{\varrho\in \mathcal J}f(Y_n(\varrho,u)))\le f(1)+ \sum_{k=1}^\infty 4d_{j_{|u|+k}}  e_{j_{|u|+k}}\sup_{\varrho\in \mathcal J} \|Y_k(\varrho,u)-Y_{k-1}(\varrho,u)\|_{p_{j_{|u|+k}}},
$$
and \eqref{control33} comes from similar estimations as those leading to \eqref{control3} by using \eqref{control} and \eqref{control2}.

\medskip 

Let us show that $Y$ does not vanish on $\mathcal J$, conditionally on $\Omega_{A_0}=\{\widetilde\TT\neq\emptyset\}$.  
 For each $n\ge 1$ let $\mathcal J_n=\{\varrho_{|n}:\varrho\in\mathcal J\}$, and for $\gamma=(\gamma_1,\ldots,\gamma_n)\in \mathcal J_n$ define the event $E_\gamma = \{ \omega\in\{\widetilde\TT\neq\emptyset\}: \exists \,  \varrho\in \mathcal J,\  Y(\varrho)=0,\ \varrho_{|n}=\gamma \}$.  Let $E=\{ \omega\in\{\widetilde\TT\neq\emptyset\}: \exists  \varrho\in \mathcal J,\  Y(\varrho)=0\}$. By construction, since for all $u\in\bigcup_{n\ge 0}\N_+^n$ we have the branching property 
 \begin{equation}\label{BrPr}
 Y(\varrho,u)=\sum_{i=1}^{N_u}W_{\varrho,ui} Y(\varrho,ui)= \sum_{i=1}^{N_u\land A_{j_{|u|+1}}}\mathbf {1}_{\{\|X_{ui}\|\le A_{j_{|u|+1}}\}} W_{\varrho,ui} Y(\varrho,ui),
 \end{equation}
 conditionally on $\{\widetilde\TT\neq\emptyset\}$, the event  $\{\exists\,   \varrho\in \mathcal J,\  Y(\varrho)=0\}$ equals $\{\exists \,  \varrho\in \mathcal J,\, \forall \, n\ge 1, \forall\, u\in\widetilde\TT_n,\,   Y(\varrho,u)=0\}$, which  is a tail event conditionally on $\{\widetilde\TT\neq\emptyset\}$ for the filtration $\big (\{\{\widetilde\TT \neq\emptyset\}\cap B: B\in\mathcal G_n\})_{n\ge 1}$.   Since the events of $\bigcap_{n\ge 1}\mathcal G_n$ have probability 0 or 1, it follows that $E$  has probability 0 or $\P(\{\widetilde\TT\neq\emptyset\})$. The same property holds for the events $E_\gamma$, $\gamma\in \bigcup_{n\ge 1}\mathcal J_n$. 
 
 Suppose that $E$ has probability $\P(\{\widetilde\TT\neq\emptyset\})$. Since $E=\bigcup_{\varrho_1\in \mathcal J_1} E_{\gamma_1}$, necessarily, there exists $\gamma_1\in\mathcal J_1$ such that $\mathbb P(E_{\gamma_1})>0$, and so $\mathbb P(E_{\gamma_1})=\P(\{\widetilde\TT\neq\emptyset\})$. Iterating this remark we can build an infinite deterministic sequence $\gamma=(\gamma_k)_{k\ge 1}\in \mathcal J$ such that $\mathbb P(E_{(\gamma_1,\ldots,\gamma_n)})=\P(\{\widetilde\TT\neq\emptyset\})$ for all $n\ge 1$. This means that almost surely, for all $n\ge 1$, there exists $\varrho^{(n)}\in  \mathcal J$ such that $\varrho^{(n)}_{|n}=(\gamma_1,\ldots,\gamma_n)$ and $Y(\varrho^{(n)})=0$. But $\varrho^{(n)}_{|n}=(\gamma_1,\ldots,\gamma_n)$ implies that $\varrho^{(n)}$ converges to $\gamma$. Hence, by continuity of $Y(\cdot)$ at $\gamma$, we get $Y(\gamma)=0$ on $\{\widetilde\TT\neq\emptyset\}$ almost surely, hence almost surely, since $\{\widetilde\TT=\emptyset\}\subset \{Y(\cdot)=0\}$. However, a consequence of our convergence result for $Y_n$ is that the martingale $Y_n(\gamma)$ converges in $L^1$ to $Y(\gamma)$, so that $\mathbb{E}(Y(\gamma))=1$. This is a contradiction. Thus $\mathbb{P}(E)=0$ and  $\P(\{\widetilde\TT\neq\emptyset\}\Delta \{Y(\cdot) >0\})=0$. Similarly, $\P(\{\widetilde\TT(u)\neq\emptyset\}\Delta \{Y(\cdot,u) >0)=0$ for all $u\in \bigcup_{n\ge 1}\N_+^n$.

\medskip

\noindent
(2) This is a direct consequence of the branching property \eqref{BrPr}.

\medskip

\noindent
(3) We will use the notion of Hellinger distance between probability measures  (it was already used in the context of Mandelbrot martingales in \cite{LiuR} to prove the mutual singularity of the branching and visibility measures on $\partial\TT$). 

For $j\ge 1$ let 
\begin{equation*}
\begin{split}
\widetilde {\mathcal J}_j&=\{(\varrho,\varrho')\in \mathcal J\times \mathcal J: q_k\neq q'_k, \ \forall\ M_{j-1}+1\le k\le M_j\},\text{ and }\\
\widetilde {\mathcal J}&=\bigcap_{\ell \ge 1}\bigcup_{j\ge \ell}\widetilde {\mathcal J}_j=\{(\varrho,\varrho')\in \mathcal J\times \mathcal J: q_k\neq q'_k\text{ for infinitely many $k$}\}.
\end{split}
\end{equation*}
For $n\ge 1$ and $(\varrho,\varrho')\in\widetilde {\mathcal J}$ let 
$
A_n(\varrho,\varrho')=\sum_{u\in\TT_n} \mu_\varrho([u])^{1/2}\mu_{\varrho'}([u])^{1/2}.
$
Notice that $(A_n(\varrho,\varrho'))_{n\ge 1}$ is non increasing. Let $A(\varrho,\varrho')$ denote  its limit. If we show that $A(\varrho,\varrho')=0$, then if $\mu_\varrho\neq0\neq\varrho'$ by definition the Hellinger distance between $\mu_\varrho/\|\mu_\varrho\|$ and $\mu_{\varrho'}/\|\mu_{\varrho'}\|$ is 1, i.e. $\mu_\varrho$ and $\mu_{\varrho'}$ are mutually singular. Of course, we want $A(\varrho,\varrho')=0$ almost surely, conditionally on $\wt T\neq\emptyset$, simultaneously for all $(\varrho,\varrho')\in \mathcal J$.  We notice that 
for every $j\ge \ell\ge 1$ and $(\varrho,\varrho')\in \widetilde {\mathcal J}_j$, we have 
$
A(\varrho,\varrho')\le \sum_{u\in\TT_{M_j}} \mu_\varrho([u])^{1/2}\mu_{\varrho'}([u])^{1/2}.
$
Consequently, for all $\ell\ge 1$
$$
A=\sup_{(\varrho,\varrho')\in \widetilde {\mathcal J}}A(\varrho,\varrho')\le\sum_{j\ge\ell} \sup_{(\varrho,\varrho')\in \widetilde {\mathcal J}_j}A(\varrho,\varrho')\le  \sum_{j\ge \ell}\sup_{(\varrho,\varrho')\in \widetilde {\mathcal J}_j} \sum_{u\in\TT_{M_j}} \mu_\varrho([u])^{1/2}\mu_{\varrho'}([u])^{1/2},
$$
so 
\begin{eqnarray*}
\E (A)
&\le &\sum_{j\ge \ell} \mathbb{E}\Big (\sup_{(\varrho,\varrho')\in\widetilde {\mathcal J}_j}\sum_{u\in\TT_{M_j}} \mu_\varrho([u])^{1/2}\mu_{\varrho'}([u])^{1/2}\Big )\\
&\le &\sum_{j\ge \ell} \E \Big ( \sup_{(\varrho,\varrho')\in\widetilde {\mathcal J}_j}\sum_{u\in\TT_{M_j}}\Big (\prod_{k=1}^{M_{j}}W_{\varrho,u_1\cdots u_k}^{1/2}W_{\varrho',u_1\cdots u_k}^{1/2}\Big ) Y(\varrho,u)^{1/2}Y(\varrho',u)^{1/2}\Big )\\
&\le& \sum_{j\ge \ell} \E \Big ( \sup_{(\varrho,\varrho')\in\widetilde {\mathcal J}_j}\sum_{u\in\TT_{M_j}}\Big (\prod_{k=1}^{M_{j}}W_{\varrho,u_1\cdots u_k}^{1/2}W_{\varrho',u_1\cdots u_k}^{1/2}\Big ) \sup_{\varrho\in  {\mathcal J}}Y(\varrho,u)\Big )\\
&\le& \sum_{j\ge \ell}  \sum_{(\varrho_{M_j},\varrho'_{M_j}):(\varrho,\varrho') \in\widetilde {\mathcal J}_j}\E \Big (\sum_{u\in\TT_{M_j}}\Big (\prod_{k=1}^{M_{j}}W_{\varrho,u_1\cdots u_k}^{1/2}W_{\varrho',u_1\cdots u_k}^{1/2}\Big ) \sup_{\varrho\in  {\mathcal J}}Y(\varrho,u)\Big )\\
&\le& \sum_{j\ge \ell}  \sum_{(\varrho_{M_j},\varrho'_{M_j}):(\varrho,\varrho') \in\widetilde {\mathcal J}_j}\E \Big (\sum_{u\in\TT_{M_j}}\Big (\prod_{k=1}^{M_{j}}W_{\varrho,u_1\cdots u_k}^{1/2}W_{\varrho',u_1\cdots u_k}^{1/2}\Big )\Big )\E( \sup_{\varrho\in  {\mathcal J}}Y(\varrho,u(j))\\
&\le& \sum_{j\ge \ell} (\# \{\varrho_{|M_j}: \varrho\in {\mathcal J}\})^2\sup_{(\varrho,\varrho')\in\widetilde {\mathcal J}_j}\E \Big ( \sum_{u\in\TT_{M_j}}\prod_{k=1}^{M_{j}}W_{\varrho,u_1\cdots u_k}^{1/2}W_{\varrho',u_1\cdots u_k}^{1/2}\Big) \E(\sup_{\varrho\in \mathcal J} Y(\varrho,u(j)))\\
&\le &\sum_{j\ge \ell} d_j^2 \Big (\sup_{(\varrho,\varrho')\in\widetilde {\mathcal J}_j}\prod_{k=1}^{M_{j}} c(q_k,q_k')\Big ) \E\big (\sup_{\varrho\in \mathcal J} Y(\varrho,u(j))\big ),
\end{eqnarray*}
where $u(j)$ is any word in $\N_+^{M_{j}}$, $c(q_k,q_k')$ is defined in \eqref{sing}, and we used the fact that 
$$
\E \Big ( \sum_{u\in\TT_{M_j}}\prod_{k=1}^{M_{j}}W_{\varrho,u_1\cdots u_k}^{1/2}W_{\varrho',u_1\cdots u_k}^{1/2}\Big)=\prod_{k=1}^{M_{j}} c(q_k,q_k').
$$
Moreover, $
\prod_{k=1}^{M_{j}} c(q_k,q_k')\le \prod_{k=M_{j-1}+1}^{M_j}c(q_k,q_k')
$
since $c(q_k,q_k')$ is always bounded by~1. Consequently, due to the definition of  $\widetilde {\mathcal J}_j$ and $c_{j}$ in \eqref{sing'}, and recalling \eqref{control3} we have
$$
\E (A)\le \sum_{j\ge \ell}\mathcal C_{\mathcal J} d_j^2 \exp (\sqrt{L_{j-1}}\gamma_{j}^2m_{j}) c_{j}^{L_{j}} .$$
This implies, for $\ell$ large enough, that $\E (A)\le \sum_{j\ge \ell}\mathcal C_{\mathcal J}/j^2$ due to \eqref{gammajcj} and \eqref{sing2}, and  yields $A=0$ almost surely.

\hfill $\square$

\subsubsection{ Hausdorff and packing dimensions of the measures $\{\mu_\varrho\}_{\varrho\in\mathcal J}$}

The main result of this section is the following proposition, which provides almost surely, simultaneously for all $\varrho\in \mathcal J$,  the Hausdorff and packing dimension of $\mu_\varrho$ whenever this measure does not vanish. For individual inhomogeneous Mandelbrot measures, such a result was obtained in \cite{B3} under different assumptions. 
\begin{prop}\label{lb}
With probability 1, if $\partial \wt\TT\neq\emptyset$,  for all $\varrho=(q_k)_{k\ge 1}\in \mathcal J$, 
$$
\dim (\mu_\varrho)= \liminf_{n\to\infty}n^{-1}\sum_{k=1}^n\wt P^*(\nabla\wt P_{A_{j_k}} (q_k)) \text{ and } \Dim (\mu_\varrho)= \limsup_{n\to\infty}n^{-1}\sum_{k=1}^n\wt P_{A_{j_k}}^*(\nabla\wt P_{A_{j_k}} (q_k)). 
$$ 
\end{prop}
The proposition easily follows from the next lemma, \eqref{defmurho}, and the fact that $\log |[t_n]|=-n$ for all $t\in \partial \TT$ and $n\ge 1$. 
\begin{lem}\label{dDimmurho}
There exists a positive sequence $(\delta_n)_{n\ge 1}$ converging to 0 as $n\to\infty$ such that with probability 1, conditionally on $\partial \widetilde \TT\neq\emptyset$, for all $\varrho\in \mathcal J$, for $\mu_\varrho$-almost every $t$,  for $n$ large enough, we have 
\begin{equation}\label{dimmurho1}
n^{-1}\Big |\sum_{k=1}^n (\langle q_k|X_{t_1\cdots t_k}\rangle- \wt P_{A_{j_k}}(q_k))+ \wt P_{A_{j_k}}^*(\nabla \wt P_{A_{j_k}}(q_k))\Big |\le \delta_n 
\end{equation}
and 
\begin{equation}\label{dimmurho2}
n^{-1}|\log Y(\varrho, t_1\cdots t_n)| \le \delta_n.
\end{equation}
\end{lem}
\begin{proof} We first prove \eqref{dimmurho1}. Fix a positive sequence $(\delta_n)_{n\ge 1}$ converging to 0 as $n\to\infty$ to be specified. For  $\varrho\in \mathcal J$ and  $n \geq 1$ we set ~:
\begin{equation*}
\begin{split}
E_{\varrho,n}^{1}&= \Big \{t\in \partial \wt \TT :\sum_{k=1}^n (\langle q_k|X_{t_1\cdots t_k}\rangle- \wt P_{A_{j_k}}(q_k))+ \wt P_{A_{j_k}}^*(\nabla \wt P_{A_{j_k}}(q_k)) \ge   n \delta_n \Big \},\\
E_{\varrho,n}^{-1}&=\Big \{t\in \partial\wt \TT :\sum_{k=1}^n (\langle q_k|X_{t_1\cdots t_k}\rangle- \wt P_{A_{j_k}}(q_k))+ \wt P_{A_{j_k}}^*(\nabla \wt P_{A_{j_k}}(q_k)) \le   -n \delta_n\Big \}.
\end{split}
\end{equation*}
Suppose we have shown that for $\lambda\in\{-1,1\}$  we have 
\begin{equation} \label{eq222}
\E\Big (\displaystyle \displaystyle  \sup_{\varrho \in {\mathcal J} } \sum_{n\geq 1} \mu_\varrho (E_{\varrho,n}^{\lambda})\Big )< \infty. 
\end{equation}
Then, with probability 1, for all $\varrho\in\mathcal J$ and $\lambda\in \{-1,1\}$ we have $ \sum_{n\geq 1} \mu_\varrho (E_{\varrho,n}^{\lambda})<\infty$,  hence conditionally on $\partial \wt \TT\neq\emptyset$, since $\mu_\varrho\neq 0$,  $\eqref{dimmurho1}$ follows from the Borel-Cantelli lemma applied to $\mu_\varrho$. 

We prove \eqref{eq222} for $\lambda=1$,  the case $\lambda=-1$ being similar. Let   $\varrho \in {\mathcal J}$. For  every $\gamma>0$ have 
$$
\mu_\varrho (E_{\varrho,n}^{1})\leq  f_{n,\gamma}(\varrho),
$$
where
$$
 f_{n,\gamma}(\varrho)=\displaystyle\sum_{u\in \TT_n} \mu_\varrho ([u]) \prod_{k=1}^{n} \exp \big (\gamma \langle q_k|X_{u_1\cdots u_k}\rangle- \gamma \wt P_{A_{j_k}}(q_k))+\gamma \wt P_{A_{j_k}}^*(\nabla \wt P_{A_{j_k}}(q_k))- \gamma   \delta_n\big ).
 $$
The function $ f_{n,\gamma}$ can be rewritten, using \eqref{defmurho}:
$$
 f_{n,\gamma}(\varrho)=\sum_{u\in\TT_n}Y(\varrho,u)\Pi(\varrho,u)\le \sum_{u\in\TT_n} M(u)\Pi(\varrho,u),
 $$
 with $M(u)= \sup_{\varrho'\in\mathcal J}Y(\varrho',u)$ and 
 $$
 \Pi(\varrho, u)=\prod_{k=1}^{n} \exp \big (\langle (1+\gamma)q_k|X_{u_1\cdots u_k}\rangle- (1+\gamma) \wt P_{A_{j_k}}(q_k))+\gamma \wt P_{A_{j_k}}^*(\nabla \wt P_{A_{j_k}}(q_k))- \gamma   \delta_n\big ).
$$
Thus, $f_{n,\gamma}(\cdot)$ depending only on the the finite sequence $\varrho_{|n}$, we have, by using the independence between the  $\sigma(\sup_{\varrho\in\mathcal J}Y(\varrho,u): |u|\ge n)$, and $\mathcal F_n$, and denoting by $u_n$ any element of $\N_+^n$, 
\begin{multline*}
\mathbb{E}(\sup_{\varrho\in\mathcal J}  f_{n,\gamma}(\varrho))\le  \mathbb{E}(\sup_{\varrho_{|n}:\varrho\in\mathcal J}\sum_{u\in\TT_n} M(u)  \Pi(\varrho, u))\le  \mathbb{E}(\sum_{\varrho_{|n}:\varrho\in\mathcal J}\sum_{u\in\TT_n} M(u)  \Pi(\varrho, u))\\\le \sum_{\varrho_{|n}:\varrho\in\mathcal J}   \E(M(u_n))  \E(\sum_{u\in\TT_n} \Pi(\varrho,u))\le  \E(M(u_n)) \, d_{j_n} \sup_{\varrho_{|n}:\varrho\in\mathcal J} \E(\sum_{u\in\TT_n}\Pi(\varrho,u)),
\end{multline*}
since $\#\{\varrho_{|n}:\varrho\in\mathcal J\}=d_{j_n}$. We have 
$$
\E(\sum_{u\in\TT_n}\Pi(\varrho,u))=\prod_{k=1}^{n}\exp \big (\wt P_{A_{j_k}}((1+\gamma) q_k)-(1+\gamma)\wt P_{A_{j_k}}(q_k)+\gamma \wt P_{A_{j_k}}^*(\nabla \wt P_{A_{j_k}}(q_k))- \gamma   \delta_n\big ).
$$
Then, using for each $1\le k\le n$ the Taylor expansion with integral rest of $\gamma\mapsto  \wt P_{A_{j_k}}((1+\gamma) q_k)-(1+\gamma)\wt P_{A_{j_k}}(q_k)$ at $0$, and taking $\gamma=\gamma_{j_n}$, we get, using \eqref{tmj} and \eqref{mjj}
$$
 \wt P_{A_{j_k}}((1+\gamma) q_k)-(1+\gamma)\wt P_{A_{j_k}}(q_k)+\gamma \wt P_{A_{j_k}}^*(\nabla \wt P_{A_{j_k}}(q_k))- \gamma   \delta_n\le -\gamma_{j_n} (\delta_n- \gamma_{j_n}m_{j_n}).
 $$ 
Combining this with \eqref{control3} (in which we bound $\sqrt{n}$ by $n$) we obtain
$$
\mathbb{E}(\sup_{\varrho\in\mathcal J}  f_{n,\gamma}(\varrho))\le C_{\mathcal J}d_{j_n} \exp (-n(\delta_n-2\gamma_{j_n}m_{j_n})). 
$$
Now, notice that due to  \eqref{control0}:  $d_{j_n}\le \exp (L_{j_n-1}^{1/5})\le \exp (n^{1/5})$ and $\gamma_{j_{n}}^2 m_{j_n}\ge L_{j_n-1}^{-1/4}\ge n^{-1/4}$. 

Let $\delta_n=3\gamma_{j_n}m_{j_n}$. 
We thus have 
$$
\mathbb{E}(\sum_{n\ge 1}\sup_{\varrho\in\mathcal J}  f_{n,\gamma_{j_n}}(\varrho))\le C_{\mathcal J} \sum_{n\ge 1}d_{j_n}\exp (-n\gamma^2_{j_n} m_{j_n})\le C_{\mathcal J} \sum_{n\ge 1}\exp (n^{1/5}-n^{3/4})<\infty.
$$
\medskip

Now we prove \eqref{dimmurho2}. Let $(\delta_n)_{n\ge 1}$ be a positive sequence converging to 0. For  $\varrho\in \mathcal J$ and  $n \geq 1$ we set ~:
\begin{equation*}
\begin{split}
F_{\varrho,n}^{1}= \Big \{t\in \partial \wt \TT : \, Y(\varrho, u)\ge  \exp( n \delta_n) \Big \}\quad\text{and}\quad 
F_{\varrho,n}^{-1}=\Big \{t\in \partial\wt \TT :\, Y(\varrho, u) \le   \exp (-n \delta_n)\Big \}.
\end{split}
\end{equation*}
It is enough to show that for $\lambda\in\{-1,1\}$  we have 
\begin{equation} \label{eq2222}
\E\Big (\displaystyle \displaystyle  \sup_{\varrho \in {\mathcal J} } \sum_{n\geq 1} \mu_\varrho (F_{\varrho,n}^{\lambda})\Big )< \infty. 
\end{equation}
Let us start with the case $\lambda=1$. Let $g(x)=\exp (\sqrt{\log (x+3)})$. We have 
$$
\mu_\varrho (F_{\varrho,n}^1)\le \sum_{u\in\TT_n} \mu_\varrho([u]) g(Y(\varrho,u)) g(\exp( n \delta_n))^{-1}.
$$
Consequently, using \eqref{defmurho} and \eqref{control33}
\begin{eqnarray*}
\E(\sup_{\varrho\in\mathcal J}\mu_\varrho (F_{\varrho,n}^1))&\le& g(\exp( n \delta_n))^{-1}   \E\big (\sup_{\varrho\in\mathcal J} Y(\varrho,u)g(Y(\varrho,u))\big )d_{j_n}\sup_{\varrho_\in\mathcal J}\E (\sum_{u\in\TT_n} W_{\varrho,u_1\cdots u_n})\\
&=&d_{j_n} g(\exp( n \delta_n))^{-1}   \E\big (\sup_{\varrho\in\mathcal J} Y(\varrho,u)g(Y(\varrho,u))\big )\\
&\le & C_{\mathcal J} d_{j_n} g(\exp( n \delta_n))^{-1} \exp (\gamma_{j_n}^2 m_{j_n}\sqrt{n}) .
\end{eqnarray*}
Since $g(\exp( n \delta_n))^{-1}\le \exp (-\sqrt{n\delta_n})$, taking $\delta_n=(2\gamma_{j_n}^2 m_{j_n})^2$ yields (using again that $d_{j_n}\le \exp (n^{1/5})$ and $\gamma_{j_n}^2 m_{j_n}\ge n^{-1/4}$)
$$
\E(\sum_{n\ge 1}\sup_{\varrho\in\mathcal J}\mu_\varrho (F_{\varrho,n}^1))\le C_{\mathcal J} \sum_{n\ge 1} d_{j_n} \exp (-\gamma_{j_n}^2 m_{j_n}\sqrt{n})\le C_{\mathcal J} \sum_{n\ge 1} \exp (n^{1/5}-n^{1/4})<\infty.
$$
The case $\lambda=-1$ is simpler. Indeed, since 
$$
\mu_\varrho (F_{\varrho,n}^{-1})\le \sum_{u\in\widetilde \TT_n} \mu_\varrho([u]) Y(\varrho,u)^{-1/2}\exp ( -n \delta_n/2),
$$
we have, using computations similar to those used above, 
\begin{eqnarray*}
\E(\sup_{\varrho\in\mathcal J}\mu_\varrho (F_{\varrho,n}^1))&\le& d_{j_n}\exp ( -n \delta_n/2)  \E\big (\sup_{\varrho\in\mathcal J} Y(\varrho,u)^{1/2})\\
&\le& C_{\mathcal J}\exp ( -n \delta_n/2) d_{j_n}
\exp (\gamma_{j_n}^2 m_{j_n}\sqrt{n}/2).
\end{eqnarray*}
Then, taking $\delta_n= 3\gamma_{j_n}^2 m_{j_n}/\sqrt{n}$ yields
$$
\E(\sum_{n\ge 1}\sup_{\varrho\in\mathcal J}\mu_\varrho (F_{\varrho,n}^{-1}))\le C_{\mathcal J}  \sum_{n\ge 1}  d_{j_n}
\exp (-\gamma_{j_n}^2 m_{j_n}\sqrt{n}) \le  C_{\mathcal J}\sum_{n\ge 1} \exp (n^{1/5}-n^{1/4})<\infty.
$$ 
\end{proof}

\subsubsection{Lower bounds for the Hausdorff and packing dimensions of the sets $E_X(K)$}

The sharp lower bound estimates  for the Hausdorff and  packing dimensions of the set $E_X(K)$ are direct consequences of Proposition~\ref{lb} and the following last two propositions.
\begin {prop} \label{pp3}
There is a positive sequence $(\delta_n)_{n\ge 1}$ converging to 0 such that, with probability  $1$, for all $\varrho =(q_k)_{k\geq 1} \in {\mathcal J}$,  for $\mu_\varrho$-almost all  $t \in\partial\wt \TT$, for $n$ large enough, we have  
 $$n^{-1}\Big |\displaystyle S_nX(t) - \displaystyle\sum_{k=1}^{n} \nabla \widetilde P_{A_{j_k}}(q_k) \Big |  \le \delta_n.$$
\end{prop}

\begin{proof} Fix  a positive sequence $(\delta_n)_{n\ge 1}$ converging to 0, to be specified later. 
 Let  $v$ a be vector of the canonical basis $\mathcal B$ of  $\R^d$. For  $\varrho\in \mathcal J$, $n \geq 1$  we set ~:
\begin{equation*}
\begin{split}
E_{\varrho,n}^{1}(v)&= \Big \{t\in \partial\wt\TT :\Big  \langle v \Big | \displaystyle S_nX(t) - \displaystyle\sum_{k=1}^{n} \nabla \widetilde P_{A_{j_k}}(q_k)\Big \rangle \ge   n \delta_n\Big \},\\
E_{\varrho,n}^{-1}(v)&=\Big \{t\in \partial\wt\TT:\Big  \langle v \Big | S_nX(t) - \displaystyle\sum_{k=1}^{n} \nabla \widetilde P_{A_{j_k}}(q_k)\Big \rangle \le   -n \delta_n \Big \}.
\end{split}
\end{equation*}
Suppose we have shown that for all $\lambda\in\{-1,1\}$ and $v\in\mathcal B$ we have 
\begin{equation} \label{eq2}
\E\Big (\displaystyle \displaystyle  \sup_{\varrho \in {\mathcal J} } \sum_{n\geq 1} \mu_\varrho (E_{\varrho,n}^{\lambda}(v))\Big )< \infty. 
\end{equation}
Then,  with probability  $1$, for all  $\varrho \in {\mathcal J} $, $\lambda\in\{-1,1\}$ and $v\in \mathcal B$, $\displaystyle\sum_{n\geq 1} \mu_\varrho (E_{\varrho,n,\epsilon}^{\lambda}(v))<\infty$. Consequently, by the Borel-Cantelli lemma applied to $\mu_\varrho$ whenever this measure does not vanish, for $\mu_\varrho$-almost every $t\in\partial\wt \TT$, for all $v\in\mathcal B$,  for $n$ large enough, we have 
$$
\Big |\Big  \langle v \Big | n^{-1}\Big( \displaystyle S_nX(t)  - \displaystyle\sum_{k=1}^{n} \nabla \widetilde P_{A_{j_k}}(q_k)\Big )\Big \rangle\Big |\le \delta_n, 
$$
which yields the desired result. 

Now we prove \eqref{eq2} when $\lambda=1$ (the case $\lambda=-1$ is similar). Let   $\varrho \in {\mathcal J}$ and $n\ge 1$. For  every $\gamma>0$ have 
\begin{eqnarray*}
\mu_\varrho (E_{\varrho,n}^{1}(v))\leq f^v_{n,\gamma}(\varrho):= \displaystyle\sum_{u\in \TT_n} \mu_\varrho ([u]) \prod_{k=1}^{n} \exp \big (\gamma \langle v |X_{u_{|k}}\rangle - \gamma  \langle v|\nabla \widetilde P_{A_{j_k}}(q_k)\rangle - \gamma   \delta_n\big ),
\end{eqnarray*} 
and due to \eqref{defmurho}, $ f_{n,\gamma}(\varrho)$ can be written
$$
 f^v_{n,\gamma}(\varrho)=\displaystyle\sum_{u\in \TT_n }  Y(\varrho,u) \prod_{k=1}^{n} \exp \big ( \langle q_k +\gamma v|X_{u_{|k}}\rangle - \widetilde P_{A_{j_k}}(q_k) - \langle \gamma   v|\nabla \widetilde P(q_k)\rangle - \gamma   \delta_n\big ).
 $$
We have 
\begin{eqnarray*}
\sup_{\varrho\in \mathcal J}  f^v_{n,\gamma}(\varrho)\le \sup_{\varrho_{|n}: \varrho\in\mathcal J}\displaystyle\sum_{u\in \TT_n }M(u)  \prod_{k=1}^{n} \exp \big ( \langle q_k+\gamma v |X_{u_{|k}}\rangle - \widetilde P_{A_{j_k}}(q_k) - \langle \gamma   v|\nabla \widetilde P_{A_{j_k}}(q_k)\rangle - \gamma   \delta_n\big ),
\end{eqnarray*}
with $M(u)=\sup_{\varrho\in \mathcal J}Y(\varrho,u)$. Consequently, since $\mathbb{E} (M(u))\le C_{\mathcal J}\exp (\epsilon_{|u|} |u|) $ by \eqref{control3}, we have (taking into account the independences and bounding $ \sup_{\varrho_{|n}: \varrho\in\mathcal J}$ by $ \sum_{\varrho_{|n}: \varrho\in\mathcal J}$)
\begin{eqnarray*}
&&\E (\sup_{\varrho\in \mathcal J} f^v_{n,\gamma} (\varrho))\\
&\le & C_{\mathcal J}\exp (n\epsilon_n)  \sum_{\varrho_{|n}: \varrho\in\mathcal J}\mathbb{E}\Big ( \sum_{u\in T_n} \prod_{k=1}^{n} \exp\big (\langle q_k+\gamma v|X_{u_{|k}}\rangle- \widetilde P_{A_{j_k}}(q_k) - \langle \gamma   v|\nabla \widetilde P_{A_{j_k}}(q_k)\rangle - \gamma   \delta_n\big )\Big )\\
&=& C_{\mathcal J}\exp(n\epsilon_n) \sum_{\varrho_{|n}: \varrho\in\mathcal J}\prod_{k=1}^{n} \exp\big (\widetilde P_{A_{j_k}}( q_k+\gamma v)- \widetilde P_{A_{j_k}}(q_k) - \langle \gamma   v|\nabla \widetilde P_{A_{j_k}}(q_k)\rangle - \gamma   \delta_n\big ).
\end{eqnarray*} 
Writing for each  $1\le k\le n$ the Taylor expansion with integral rest of order 2 of $\gamma\mapsto \widetilde P_{A_{j_k}}( q_k+\gamma v)- \widetilde P_{A_{j_k}}(q_k) - \langle \gamma   v|\nabla \widetilde P(q_k)\rangle$ at $0$, taking  $\gamma=\gamma_{j_n}$, and using \eqref{mj} and \eqref{mjj} we get 
$$
\sum_{k=1}^n \wt P_{A_{j_k}}( q_k+\gamma v)- \widetilde P_{A_{j_k}}(q_k) - \langle \gamma   v|\nabla \widetilde P_{A_{j_k}}(q_k)\rangle\le n\gamma_{j_n}^2m_{j_n}
$$
uniformly in $\varrho\in \mathcal J$. Consequently, using that $\epsilon_n=\gamma_{j_n}^2m_{j_n}$ and $\#(\{\varrho_{|n}: \varrho\in\mathcal J\})=d_{j_n}$, we get
$$
\E (\sup_{\varrho\in \mathcal J} f^v_{n,\gamma_{j_{n}}} (\varrho)\le C_{\mathcal J} d_{j_n}\exp  (-n\gamma_{j_n}(\delta_n-2\gamma_{j_n}m_{j_n}) ).
$$
Now, take $\delta_n=3\gamma_{j_n}m_{j_n}$. Since by  \eqref{control0} we have   $d_{j_n}\le \exp (L_{j_n-1}^{1/5})\le \exp (n^{1/5})$ and $\gamma_{j_{n}}^2m_{j_n}\ge L_{j_n}^{-1/4}\ge n^{-1/4}$, we get
$$
\E \big (\sum_{n\ge 1}\sup_{\varrho\in \mathcal J} f^v_{n,\gamma_{j_{n}}} (\varrho))\le  C_{\mathcal J} \sum_{n\ge 1} \exp (n^{1/5}-n^{3/4})<\infty.
$$
\end{proof}

\begin{prop}\label{conclusion}
With probability 1, conditionally on $\partial \wt \TT\neq\emptyset$, for every $K\in \mathcal K$ such that $K\subset I$,  there exists $\varrho=(q_k)_{k\ge 1}\in \mathcal J$ such that 
$$
\begin{cases}
 \displaystyle \bigcap_{N\ge 1}\overline{ \Big \{n^{-1}\sum_{k=1}^n\nabla\wt P (q_k):n\ge N\Big\}}=K\\
 \dim \mu_\varrho= \inf\{\wt P^*(\alpha):\alpha\in K\},\ \Dim\mu_\varrho= \sup\{\wt P^*(\alpha):\alpha\in K\}.
\displaystyle
\end{cases}.
$$
If, moreover, $K$ is not compact, $\varrho$ can be chosen so that $\Dim\mu_\varrho=\dim \partial\TT$. 
\end{prop}
\begin{proof}
Let $m_0\in\N_+$ such that $K\cap B(0,m_0)\neq\emptyset$. For all $m\ge m_0$ let $\beta_m,\gamma_m\in K\cap B(0,m)$ such that $\wt P^*(\beta_m)\le \inf_{\alpha\in K\cap B(0,m)}\wt P^*(\alpha)+1/m$ and $\wt P^*(\gamma_m)= \sup_{\alpha\in K\cap B(0,m)}\wt P^*(\alpha)$ (recall that $\wt P^*$ is upper-semi continuous so it reaches its supremum over compact sets).  Let $(B(\alpha_{m,\ell},1/m))_{1\le \ell\le b_m}$ be an $1/m$ centered covering of $K\cap B(0,m)$ such that $\beta_m$ and $\gamma_m$ belong to the collection $\{\alpha_{m,\ell}:1\le \ell\le b_m\}$ and $\{\alpha_{m-1,\ell}:1\le \ell\le b_{m-1}\}\subset\{\alpha_{m,\ell}:1\le \ell\le b_m\}$ for $m\ge m_0+1$. Since $K\in\mathcal K$, we can find  $(j_m)_{m\ge m_0}$, an increasing sequence of integers, such that for each $m$ we have $j_m\ge m$ and we can complete the collection  $(B(\alpha_{m,\ell},1/m))_{1\le \ell\le b_m}$ into a finite collection $(B(\alpha_{m,\ell},1/m))_{1\le \ell\le \ell_m}$ such that each $\alpha_{m,\ell}$ belongs to $K\cap B(0,j_m)$ and $\{\alpha_{m,\ell}: 1\le \ell\le \ell_m\}$ is $1/m$-chained, that is for any pair $\{\beta,\beta'\}$ of points in $\{\alpha_{m,\ell}: 1\le \ell\le \ell_m\}$ one can find $\beta_0=\beta, \beta_1,\ldots\beta_p=\beta'$ in $\{\alpha_{m,\ell}: 1\le \ell\le \ell_m\}$ such that $B(\beta_{i},1/m)\cap B(\beta_{i+1},1/m)\neq\emptyset$ for all $0\le i<p$.

Without loss of generality we can then assume that $B(\alpha_{m,\ell},1/m)\cap B(\alpha_{m,\ell+1},1/m)\neq \emptyset$ for all $1\le \ell\le \ell_m-1$, and $B(\alpha_{m+1,1},1/m)\cap B(\alpha_{m,\ell_m},1/m)\neq\emptyset$. We can also assume that $\ell_m\ge j_{m+1}-j_m+1$ by repeating some ball in the collection $(B(\alpha_{m,\ell},1/m))_{1\le \ell\le \ell_m}$ if necessary.

By construction of the sequence $(D_j)_{j\ge 1}$ in Proposition~\ref{approxi}, for each $m\ge m_0$, for all $j\ge j_{m}$ and $1\le \ell\le \ell_m$, we can fix $q_{j,m,\ell}\in D_j$ such that $\|\nabla\wt P_{A_{j}}(q_{j,m,\ell})-\alpha_{m,\ell}\|\le  1/j\le  1/m$ and $|\wt P_{A_j}^*(\nabla\wt P_{A_{j}}(q_{j,m,\ell}))-\wt P^*(\alpha_{m,\ell})|\le1/j\le  1/m$.

\medskip

We build $\varrho$ as follows. We take any sequence $(q_k)_{1\le k\le M_{j_{m_0}-1}}$. Then, for all $1\le \ell\le \ell_{m_0}$ we set $q_k=q_{j_{m_0}-1+\ell,m_0,\ell}$ for all $k\in [M_{j_{m_0}-1+\ell-1}+1,M_{j_{m_0}-1+\ell}]$.  This yields $q_k$ for $M_{j_{m_0}-1}+1\le k\le M_{j_{m_0}-1+\ell_{m_0}}$.  Set $J_{m_0}=j_{m_0}-1+\ell_{m_0}$.  Then we define recursively the sequence $q_k$ by taking, for $m\ge m_0+1$, for all $1\le \ell\le \ell_m$,  $q_k=q_{J_{m-1}+\ell,m,\ell}$ for all $k\in [M_{J_{m-1}+\ell}+1,M_{J_{m-1}+\ell}]$, where $J_m= J_{m-1}+ \ell_m $. Notice that the property 
$\ell_m\ge j_{m+1}-j_m+1$ implies that $J_{m-1}\ge j_m$, hence the choice of $q_{J_{m-1}+\ell,m,\ell}$ is possible.  

Now, if  $n$ is large enough, $j_{n-1}$ takes the form  $j_{n-1}= J_{m-1}+\ell-1$. Recall that we have $\|\alpha_{m,\ell}-\alpha_{m,\ell-1}\|\le 1/m$ if $\ell\ge 2$ and $\|\alpha_{m,1}-\alpha_{m-1,\ell_{m-1}}\|\le 2/(m-1)$ if $\ell=1$, and $\|\nabla\wt P_{A_{j}}(q_{j,m,\ell})-\alpha_{m,\ell}\|\le  1/m$ for each $j\ge j_m$ and $1\le \ell\le \ell_m$. Then, it is an exercise to check that due to \eqref{control2bis}, we have 
$$
\Big \|\alpha_{m,\ell}- n^{-1}\sum_{k=1}^n \nabla\wt P_{A_{j_k}} (q_k) \Big\|= O(1/m+ 1/(J_{m-1}+\ell-1))
$$ (basically, due to \eqref{control2bis}, in $n^{-1}\sum_{k=1}^n \nabla\wt P_{A_{j_k}} (q_k)$ only the terms corresponding to $j_k=j_n$ and $j_k=j_n-1$ can have a significative contribution, and they differ from each other and $\alpha_{m,\ell}$ by at most $1/(m-1)$). 
Since by construction $\{\alpha_{m,\ell}: \, m\ge 1,\, 1\le \ell \le b_m\} $ is a dense subset of $K$ and each vector $\alpha_{m,\ell}$ occurs in  all the sets $\{\alpha_{m',\ell}: 1\le \ell \le b_{m'}\} $ with $m'>m$, we have $\displaystyle \bigcap_{N\ge 1}\overline{ \Big \{n^{-1}\sum_{k=1}^n\nabla\wt P (q_k):n\ge N\Big\}}=K$. 

Also, still due to \eqref{control2bis}, taking $n=M_{J_{m-1}+\ell}$, we have 
$$
\Big |\wt P^*(\alpha_{m,\ell})- n^{-1}\sum_{k=1}^n \wt P^*(\nabla\wt P_{A_{j_k}} (q_k)) \Big|=O( 1/m + 1/(J_{m-1}+\ell)).
$$ 
Then, remembering that the vectors $\beta_m$ and $\gamma_m$ belong to the collection of the $\alpha_{m,\ell}$, we get $\displaystyle \liminf_{n\to\infty}n^{-1}\sum_{k=1}^n\wt P^*(\nabla\wt P_{A_{j_k}} (q_k)) =\inf_{\alpha\in K}\wt P^*(\alpha)$ and $\displaystyle\limsup_{n\to\infty}n^{-1}\sum_{k=1}^n\wt P^*(\nabla\wt P_{A_{j_k}} (q_k))= \sup_{\alpha\in K}\wt P^*(\alpha)$. Then the conclusion about the dimensions of $\mu_\varrho$ follows from Proposition~\ref{lb}. 

Notice that at this step, taking $K=\{\alpha\}$ we get $\mu_\varrho(E_X(\alpha))=\|\mu_\varrho\|$ and $\dim E_X(\alpha)\ge \wt P^*(\alpha)$, hence $\dim \partial \TT\ge \dim\partial \wt \TT\ge \sup_{\alpha\in I}\wt P^*(\alpha)=\wt P(0)$, conditionally on $\{\partial \wt \TT\neq\emptyset\}$.

\medskip

Suppose now that $K$ is not compact.  

{\bf Case 1}:  $\dim \partial\TT=\sup_{\alpha\in I}\wt P^*(\alpha)$ is not reached in $I$. 

We can suppose that $\sup_{\alpha\in K}\wt P^*(\alpha)<\dim \partial\TT$, otherwise there is nothing to prove. Let us explain how to modify $\varrho$ so that $\Dim\mu_\varrho=\dim \partial\TT$. 

Consider an unbounded sequence of points $(\wt \gamma_m)_{m\ge 1}$ in $I$, such that $\wt P^*(\wt \gamma_m)>\sup\{\wt P^*(\alpha):\alpha\in K\}$, $\wt P^*(\wt \gamma_m)$ converges increasingly to $\dim \partial\TT$, and $\wt P^*$ is continuous at $\wt \gamma_m$. This is possible, for otherwise due to the concavity and upper semi-continuity of $\wt P^*$, and the definition of $I$,  $I$ should be compact. Also, let $(\wt\beta_{m})_{m\ge 1}$ an unbounded sequence in $K$ such that $\|\wt \beta_m\|\ge m \|\wt \gamma_m\|$ for each $m\ge 1$. Due to our assumption about the $\wt\gamma_m$ and the concavity of $\wt P^*$, for each $m\ge 1$ there is a point $\wt \beta'_{m}$ in $K\cap  [\wt \beta_m,\wt\gamma_m]$ such that $K\cap (\wt \beta'_{m}, \wt\gamma_m]=\emptyset$.

Suppose that for all $\epsilon>0$, there exists $M\ge 1$ such that for all $m\ge M$, we have $d([\wt \beta'_{m},\wt\gamma_m],\bigcup_{k=1}^M[\wt \beta'_{k},\wt\gamma_k])\le \epsilon$. Apply this with $\epsilon=1$. There exists $R>0$ such that $\bigcup_{k=1}^M[\wt \beta'_{k},\wt\gamma_k]\subset B(0,R)$, hence   $B(0,R+1)\cap   [\wt \beta'_{m},\wt\gamma_{m}])\neq\emptyset$ for all $m\ge 1$. We have $\sup_{\alpha\in B(0,R+1)}\wt P^*(\alpha)<\sup_{\alpha\in I}\wt P^*(\alpha)$. For each $m\ge 1$ let $\delta_m\in B(0,R+1)\cap   [\wt \beta'_{m},\wt\gamma_{m}]$. By construction, we have $\delta_m= (1-\lambda_m)\wt \beta_m+\lambda_m\wt \gamma_m$ with $\lambda_m\to 1$ since $\|\delta_m-\wt \gamma_m\|\sim \|\wt \gamma_m\|=o(\|\wt\beta_m\|)$.  However, we have $\wt P^*(\delta_m)\le \sup_{\alpha\in B(0,R+1)}\wt P^*(\alpha) <\sup_{\alpha\in I}\wt P^*(\alpha)=\lim_{m\to\infty}  (1-\lambda_m)\wt P^*(\wt \beta_m)+\lambda_m\wt P^*(\wt \gamma_m)$, which contradicts the concavity of $\wt P$ by considering a large enough integer $m$. Consequently, by extracting a subsequence if necessary, we can assume that the semi-open segments $(\wt\beta'_{m}, \wt\gamma_m]$, $m\ge 1$,  are distant from each other by a constant $\epsilon_0>0$ independent of~$m$.  

Now we modify the construction of $\varrho$ as follows. At first, for all $m\ge 1$, we consider  a finite sequence of balls $(B(\wt\alpha_{m,\ell},1/m))_{1\le \ell\le \wt \ell_m}$ with the following properties: the balls are centered on $[\wt \beta'_m,\wt\gamma_m]$ and form a covering of $[\wt \beta'_m,\wt\gamma_m]$; $\wt\alpha_{m,1}=\wt\beta'_m=\wt\alpha_{m,\wt\ell_m}$; $\wt\gamma_m$ belongs to  $\{\wt\alpha_{m,\ell}: 1\le \ell\le \wt \ell_m\}$; the balls corresponding to two consecutive indices are not disjoint. 

The sequence $\{\wt\alpha_{m,\ell}:1\le \ell\le \wt\ell_m\}$ can be inserted in the sequence $\{\alpha_{m,\ell}:1\le \ell\le \ell_m\}$, by increasing $j_m$ if necessary. Then the construction of $\varrho$ is done exactly as above. As a result we still have $\displaystyle \liminf_{n\to\infty}n^{-1}\sum_{k=1}^n\wt P^*(\nabla\wt P_{A_{j_k}} (q_k)) =\inf_{\alpha\in K}\wt P^*(\alpha)$ but $\displaystyle\limsup_{n\to\infty}n^{-1}\sum_{k=1}^n\wt P^*(\nabla\wt P_{A_{j_k}} (q_k))= \lim_{m\to\infty} \wt P^*(\wt\gamma_m)=\dim\partial T$. On the other hand, the set of limit points of $n^{-1}\sum_{k=1}^n \nabla\wt P_{A_{j_k}} (q_k)$ is still equal to $K$, because by construction, for each $m> \epsilon_0^{-1}$, the set $\{\alpha\in\R^d : d(\alpha,K)> 1/m\}\cap \bigcup_{\ell=1}^{\wt\ell_m} B(\wt\alpha_{m,\ell},1/m)$ does not intersect $\bigcup_{m'>m}\bigcup_{\ell=1}^{\wt\ell_{m'}} B(\wt\alpha_{m',\ell},1/m')$. 

\medskip

{\bf Case 2}:  $\dim \partial\TT=\sup_{\alpha\in I}\wt P^*(\alpha)=\wt P(0)$ is reached at, say, $\alpha_0$.  We exclude the obvious case $\dom \wt P=\{0\}$ for which $\wt P^*$ is constant.  

Since $K$ is not bounded, so is $I$. Suppose first that $\mathrm{int}(\dom \wt P^*)\neq \emptyset$. Then $0$ is in the  boundary of $\dom \wt P$, for otherwise $I$ is compact by comment (5) of Section~\ref{mainr}. Then let $u\in\R^d\setminus \{0\}$ such that the convex set $\dom \wt P$ belongs to the half-space $\langle q|u\rangle \le 0$. Since $\wt P^*(\alpha_0)=\wt P(0)$, we have $\alpha_0\in \partial \wt P(0)$, and the fact that $\langle q|u\rangle \le 0$ for $q\in \dom \wt P$ implies that the half-line $\{\alpha_0+\lambda u: \lambda\ge 0\}$ belongs to the subdifferential $\partial \wt P(0)$, giving $\wt P^*(\alpha_0+\lambda u)=\wt P(0)$ for all $\lambda \ge 0$. 

Now let $(\widetilde \beta_m)_{m\ge 1}$ be an unbounded sequence in $K$. By extracting a subsequence if necessary, we can assume that $\widetilde \beta_m/\|\widetilde \beta_m\|$ converges to $v\in\mathbb S^{d-1}$. If $v=-u$, since $I$ is convex and closed and $(\widetilde \beta_m)_{m\ge 1}$ unbounded, the half line $\{\alpha_0-\lambda u: \lambda\ge 0\}$ also belongs to $I$. If $\wt P^*(\alpha_0-\lambda u)=\wt P(0)$ for all $\lambda<0$, we must have $\wt P(q)=+\infty$ outside $\{0\}$, which is a contradiction; otherwise, taking $q$ in $\mathrm{int}(\dom \wt P)$ such that $\langle q|u\rangle <0$ and letting $\lambda$ tend to $-\infty$ gives a new contradiction. Consequently, the function $\wt P^*$ takes values smaller than $\wt P(0)$ over $\{\alpha_0-\lambda u: \lambda>0\}$, and since it is concave and positive over the whole line $\{\alpha_0-\lambda u: \lambda\ge 0\}\subset I$ this is a new contradiction. 
Thus $v\neq -u$.  It is then easy,  by extracting a subsequence of $(\wt \beta_m)_{m\ge 1}$ if necessary, to choose $(\widetilde \beta_m)_{m\ge 1}$ tending to $\infty$ as well as $(\wt \gamma_m)_{m\ge 1}$ a sequence in $\{\alpha_0+\lambda u: \lambda\ge 0\}$ such that $d([\wt \beta_i, \wt \gamma_i], [\wt \beta_j, \wt \gamma_j])\ge 1$ for all $i\neq j$. Recalling that we only have to consider the case $\sup_{\alpha\in K}\wt P^*(\alpha)<\dim\partial \TT$, the proof ends as in the first case after considering the points $\wt \beta_m'$, $m\ge 1$, and $\epsilon_0=1$. 

\medskip

Now suppose that $\mathrm{int}(\dom \wt P^*)= \emptyset$, and consider a vector $u\in\R^{d}\setminus\{0\}$ orthogonal to the affine subspace spanned by $\dom \wt P$.  We have $\wt P^*(\alpha_0+\lambda u)=\wt P(0)$ for all $\lambda\in\R$.  Considering $(\wt\beta_m)_{m\ge 1}$ as above and assume without loss of generality that $\widetilde \beta_m/\|\widetilde \beta_m\|$ converges to $v\in\mathbb S^{d-1}$. We have $v\neq u$ or $v\neq -u$ so that we can conclude as above.  
 
\end{proof}

\section{Multifractal analysis, large deviations and free energy. Proofs of Theorem~\ref{thm-1.2} and~\ref{thm-1.3} when $\dom\wt P\neq\emptyset$}\label{other proofs}
\subsection{Multifractal analysis, and $0$-$\infty$ law for Hausdorff and packing measures of the level sets $E_X(\alpha)$. Proof of Theorem~\ref{thm-1.2}}\label{mutualsing}

(1) It is a corollary of Theorem~\ref{thm-1.1}.

\smallskip

\noindent
(2) Given $\alpha\in I$, for each $j\ge 1$, fix in $D_{j}$  two distinct points $q^{(0)}_{\alpha,j}$ and $q^{(1)}_{\alpha,j}$ among those $q$ such that $(\nabla\wt P_{A_j}(q), \wt P^*(\nabla\wt P_{A_j}(q))$ is as close as possible to $(\alpha,\wt P^*(\alpha))$. Now let $\mathcal J_\alpha$ stand for the set of those sequences $\varrho$ of the form $(\underbrace{q^{(\mathbf{1}_C(j))}_{\alpha,j},\cdots,q^{(\mathbf{1}_C(j))}_{\alpha,j}}_{L_j})_{j\ge 1}$, where $C$ runs in the set of equivalent classes of subsets of $\N_+$ under the relation $S\sim S'$ if $\mathbf{1}_S(j)= \mathbf{1}_{S'}(j)$ for $j$ large enough. By construction, $\mathcal J_\alpha$ is not countable and two distinct of its elements do not coincide ultimately.  Now, it follows from Propositions~\ref{pp3}, \ref{lb} and \ref{pro-2.4}(3) that, with probability 1, conditionally on $\{\partial \wt \TT\neq\emptyset\}$,  for all $\alpha\in I$, for all $\varrho\in \mathcal J_\alpha$ the measure $\mu_\varrho$ is carried by $E_X(\alpha)$ and has exact dimension $\wt P^*(\alpha)$, and if $\varrho$ and $\varrho'$ are two distinct elements of 
 $\varrho\in \mathcal J_\alpha$, the measures $\mu_\varrho$ and $\mu_{\varrho'}$ are mutually singular. 
 
\smallskip

\noindent
(3) Let us start with Hausdorff measures. Fix $A_0$ like in Section~\ref{LB} and a set $\Omega(A_0)$ of full probability in $\{\partial \wt\TT\neq\emptyset\}$, over which the conclusions  of Lemma~\ref{dDimmurho} and Proposition~\ref{pp3} hold with the same sequence $(\delta_n)_{n\ge 1}$. 

Fix $\omega\in\Omega(A_0)$. Let $\alpha\in I$ such that  $0<\wt P^*(\alpha)<\dim \partial \TT=\wt P(0)$,  as well as $g$,  a gauge function satisfying the property $\limsup_{r\to 0^+}\log (g(r))/\log (r)\le \wt P^*(\alpha)$ (the fact that $\mathcal H^g(E_X(\alpha))=0$ if this does not hold  is obvious). Let us show that there exists $\varrho\in  \mathcal J$  such that $\mu_\varrho(E_X(\alpha))>0$ and there is  a positive sequence $(\delta'_n)_{n\ge 1}$ such that $\lim_{
n\to\infty} n\delta'_n=\infty$ and   for $\mu_\varrho$-almost every $t$, for $n$ large enough, we have $g(\mathrm{diam}([t_{|n}])\ge  \mu_\varrho([t_{|n}])^{1-\delta'_n}$.

There is a non decreasing function $\theta$ such that  $g(r)\ge r^{\wt P^*(\alpha)+\theta(r)}$, with $\lim_{r\to0^+}\theta(r)=0$. Let $\theta_n=\theta (e^{-n})$.

Let $(\eta_k)_{k\ge 1}$ be a decreasing sequence converging to 0, and such that $\lim_{n\to\infty} n\eta_n=\infty$ and $n \delta_n=o\Big (\sum_{k=1}^n\eta_k\Big )$ as $n\to\infty$. Fix $\varrho\in \mathcal J$ such that $\nabla \wt P_{A_{j_k}}(q_k)$ converges to $\alpha$ and $\wt P^*_{A_{j_k}}(\nabla \wt P_{A_{j_k}}(q_k))$  converges to $\wt P^*(\alpha)$ from above as $k$ tends to $\infty$, but enough slowly so that for $k$ large enough we have $\wt P^*_{A_{j_k}}(\nabla \wt P_{A_{j_k}}(q_k))\ge (1+\eta_k)^2 (\wt P^*(\alpha)+\theta_k)$. Lemma~\ref{dDimmurho} applied with $\mu_\varrho$ combined with \eqref{defmurho} then yields, for $\mu_\varrho$-almost every $t$, for $n$ large enough,
\begin{eqnarray*}
\mu_\varrho([t_{|n}])&\le& \exp \Big (2n\delta_n -\sum_{k=1}^n(1+\eta_k)^2 (\wt P^*(\alpha)+\theta_k)\Big )\\
&\le & \exp \Big (2n\delta_n -(1+\eta_n)(\wt P^*(\alpha)+\theta_n) \sum_{k=1}^n(1+\eta_k) \Big ) .
\end{eqnarray*}
Since  $n\delta_n=o\Big (\sum_{k=1}^n\eta_k\Big )$,  for $n$ large enough we have $2n\delta_n - (1+\eta_n)(\wt P_{\alpha}^*(0) +\theta_n) \sum_{k=1}^n\eta_k\le - (1+\eta_n)(\wt P_{\alpha}^*(0) +\theta_n) n \delta_n$, so that 
$$
\mu_\varrho([t_{|n}])\le  \exp \big (-(1+\eta_n)(\wt P^*(\alpha)+\theta_n)  (n\delta_n +n )\big ).
$$
On the other hand,  $\mathrm{diam}([t_{|n}])=e^{-n}$, so 
\begin{multline*}
\mu_\varrho([t_{n}])\le (\mathrm{diam}([t_{|n}]))^{(1+\eta_n)(\wt P^*(\alpha)+\theta_n)}\\\le (\mathrm{diam}([t_{|n}]))^{(1+\eta_n)(\wt P^*(\alpha)+\theta (\mathrm{diam}([t_{|n}]))}\le   g(\mathrm{diam}([t_{|n}]))^{(1+\eta_n)}.
\end{multline*}
Finally, the sequence $\delta'_n= \eta_n/(1+\eta_n)$ is as desired. 

Now, let $n_0\in\N_+$ and $E\in E_X(\alpha)$ such that $\mu_\varrho (E)>0$ and for all $n\ge n_0$ and $t\in E$ we have $g(\mathrm{diam}([t_{|n}])\ge  \mu_\varrho([t_{|n}])^{1-\delta'_n}$ and $\mu_\varrho ([t_{|n}])\le e^{-n\wt P^*(\alpha)/2}$.  To compute $\mathcal H^g(E)$, due to the ultrametric nature of the distance we use, we can use coverings by cylinders.  For any $e^{-n'_0}$ covering $\{[u_i]\}_{i\in\mathcal I}$ of $E$ by such cylinders intersecting $E$, with $n'_0\ge n_0$, we have 
\begin{multline*}
\sum_{i\in\mathcal I} g(\mathrm{diam}([u_i])\ge \sum_{i\in\mathcal I}  \mu_\varrho([u_i])^{1-\delta'_{|u_i|}}\\\ge \sum_{i\in\mathcal I}  \mu_\varrho([u_i]) e^{|u_i|\delta'_{|u_i|}\wt P^*(\alpha)/2}\ge e^{n'_0\delta'_{n'_0}\wt P^*(\alpha)/2}  \sum_{i\in\mathcal I}  \mu_\varrho([u_i]) \ge e^{n'_0\delta'_{n'_0}\wt P^*(\alpha)/2}\mu_\varrho(E).
\end{multline*}
As $\lim_{n\to\infty} n\delta'_n=\infty$ and $ P^*(\alpha)>0$ , we get $\mathcal H^g(E_X(\alpha))\ge \mathcal H^g(E)=\infty$. 

This yields the desired result on $\{\partial\wt \TT\neq\emptyset\}$, uniformly in $\alpha$ and $g$, and  since the set $\{\partial\wt \TT\neq\emptyset\}$ (depending on $A_0$) increases to a set of probability 1 as $A_0$ tends to $\infty$, we have the result claimed in our statement. 

\medskip

Now let us consider packing measures. We still fix $\omega\in \Omega(A_0)$ as above. Let $\alpha\in I$ such that  $0<\wt P^*(\alpha)<\dim \partial \TT=\wt P(0)$,  as well as $g$,  a gauge function satisfying the property $\liminf_{r\to 0^+}\log (g(r))/\log (r)\le \wt P^*(\alpha)$ (the simple proof of the fact that $\mathcal P^g(E_X(\alpha))=0$ if this does not hold is left to the reader).
There exist a decreasing sequence $(r_j)_{j\ge 1}\in (0,1)^{\N_+}$ converging to 0 as well as  a non increasing positive sequence $(\theta_j)_{j\ge 1}$ converging to 0 such that $g(r_j)\ge r_j^{\wt P^*(\alpha)+\theta_j} $.  Setting $n_j=\lfloor \log (r_j^{-1})\rfloor$ and redefining $\theta_j$ as being $\theta_j+n_j^{-1}(\wt P^*(\alpha_j)+\theta_j)$, we still have $g(r_j)\ge r_j^{\wt P^*(\alpha)+\theta_j} $ with $r_j=e^{-n_j}$, since $g$ is non decreasing. Now take $(\eta_n)_{n\ge 1}$, $\mu_\varrho$, and $(\delta'_n)_{n\ge 1}$ as above. By construction, for $\mu_\varrho$-almost every $t$, for $j$ large enough, we have  $g(\mathrm{diam}([t_{|n_j}])\ge  \mu_\varrho([t_{|n_j}])^{1-\delta'_{n_j}}$. Now let  $E$ be a subset of $E_X(\alpha)$ as above. Since for $j$ large enough we have $n_j\ge n_0$, considering the packing of $E$ by the cylinders of generation $n_j$ which intersect $E$ as well as the same estimate as above yields $\overline P^g(E)\ge e^{n_j\delta'_{n_j}\wt P^*(\alpha)/2}\mu_\varrho(E)$ for all $j\ge 1$, hence $\overline P^g(E)=\infty$. Since any countable covering of $E_X(\alpha)$ must contain a set of positive $\mu_\varrho$-measure, by definition of $\mathcal P^g$ we get $\mathcal P^g(E_X(\alpha))=\infty$. 

%%%%%%%%%%
%%%%%%%%%%
%%%%%%%%%%

\subsection{Large deviations and free energy. Proof of Theorem~\ref{thm-1.3}}

(1) This follows from  Theorem~\ref{thm-1.2}(1), \eqref{LD'}, and the standard fact that $\dim E_X(\alpha)\le \lim_{\epsilon\to 0}\liminf_{n\to\infty} f(n,\alpha,\epsilon)$. 

\medskip

\noindent
(2)$(a)$ We use the approximation of $\wt P$ by the functions $\wt P_A$, $A\ge A_0$, introduced in Section~\ref{sec2.1}. We first establish the result for $\wt P_A$. 

We have proved that for such a function the set $I_A=\{\alpha\in\R^d: \wt P_A^*(\alpha)\ge 0\}$ is compact. Set 
$$
f_A(\alpha):=\begin{cases}\widetilde P_A^*(\alpha)&\text{if }\alpha\in I_A\\
-\infty&\text{otherwise}
\end{cases}\quad \text{and}\quad \widehat P_A(q):= \sup\{ \langle q|\alpha\rangle +f_A(\alpha):\alpha\in\R^d\}\quad (q\in\R^d). 
$$
By construction $\widehat  P_A(q)$ is 
finite over $\R^d$. Let us show that
\begin{equation}\label{hatPA}
\widehat P_A(q)=\sup\{\langle q|\alpha\rangle +f_A(\alpha):\alpha\in\R^d\}=\inf\{\theta^{-1}\wt P_A(\theta q): \theta \in (0,1]\}.  
\end{equation}
Let $q\in\R^d$.  Due to Proposition~\ref{detI} and Corollary~\ref{approxi}, there exists a sequence $(q_n)_{n\ge 1}$ in $J_A^{\N_+}$ such that $\widehat   P_A(q)=\lim_{n\to\infty} \langle q|\nabla \wt P_A(q_n) \rangle +\wt P_A^*(  \nabla \wt P_A(q_n) )$. 

If $q\in \overline J_A$, due to the duality between $\wt P_A$ and $\wt P_A^*$ (see \eqref{duality}), we have $\wt P_A(q)= \sup\{ \langle q|\alpha\rangle +\wt P^*_A(\alpha):\alpha\in\R^d\}$ so that the sequence $(q_n)_{n\ge 1}$ can be taken equal to $q$ and we have $\widehat P_A(q)=\widetilde P_A(q)$. Moreover, an elementary study of the function $\theta\in(0,1]\mapsto \frac{ \wt P_A(\theta q)}{\theta}$ shows that it reaches its infimum at $\theta =1$, because $\wt P_A^*( \nabla \wt P_A(q))\ge 0$. 

If $q\notin \overline J_A$, there exists a unique $\theta_0\in (0,1]$ such that $\theta_0 q\in  \partial J_A$: since $\wt P_A$ is strictly convex, $\theta_0$ is the unique $\theta \in (0,1]$ at which $\theta\in(0,1] \mapsto \frac{ \wt P_A(\theta q)}{\theta}$ reaches its minimum. Then, again by duality we have $\langle \theta_0 q|\nabla \wt P_A(q_n) \rangle +\wt P_A^*(  \nabla \wt P_A(q_n) ) \le \wt P_A(\theta_0 q)$, and since  $\wt P_A^*(  \nabla \wt P_A(q_n) )\ge 0$ we also have $\langle \theta_0 q|\nabla \wt P_A(q_n) \rangle\le \wt P_A(\theta_0 q)$. This implies that $\langle q|\nabla \wt P_A(q_n) \rangle +\wt P_A^*(  \nabla \wt P_A(q_n) )\le \frac{ \wt P_A(\theta_0 q)}{\theta}$ hence $\widehat   P_A(q)\le \frac{ \wt P_A(\theta_0 q)}{\theta_0}$. On the other hand, by definition $\widehat   P_A(q)\ge  \langle q|\nabla \wt P_A(\theta_0 q) \rangle +\wt P_A^*(  \nabla \wt P_A(\theta_0 q) )=\langle q|\nabla \wt P_A(\theta_0 q)\rangle=\frac{ \wt P_A(\theta_0 q)}{\theta_0}$. Consequently \eqref{hatPA} holds.

Now, on the one hand, we notice that since,  as $A\to\infty$, the non decreasing sequence $(\wt P_A^*)$ converges to $\wt P^*$ over $\mathrm{int} (\dom \wt P^*)$, the sequence $(f_A)$ is non decreasing and converges  to $f$ over $\mathrm{int} (\dom \wt P^*)=\mathrm{int} (\dom f)$, so that $\widehat P_A$ is also non decreasing in $A$ and it converges to $\widehat P$ (here we used \eqref{duality}).  

On the other hand,  since the non-decreasing  sequence  $\wt P_A$, $A\ge A_0$, converges to the closed convex function $\wt P$, for each $q\in \R^d\setminus \{0\}$, we have the following alternative, denoting by $\theta_A$ the unique $\theta\in (0,1]$ at which $\frac{ \wt P_A(\theta q)}{\theta}$ reaches its minimum: either $\wt P(\theta q)=\infty$ for all $\theta\in (0,1]$, in which case we clearly have both $\inf\{\theta^{-1}\wt P(\theta q): \theta \in (0,1]=\infty$ and $\lim_{A\to\infty} \frac{ \wt P_A(\theta_A q)}{\theta_A}=\infty$, or there exists $\theta\in (0,1]$ such that $\wt P(\theta q)<\infty$. In this case, since $\wt P$ is strictly convex, there is a unique $\theta_0$ at which $\frac{ \wt P(\theta q)}{\theta}$ reaches its minimum. 

We have $\limsup_{A\to\infty}  \frac{ \wt P_A(\theta_A q)}{\theta_A}\le \frac{ \wt P(\theta_0 q)}{\theta_0}$, otherwise there exists $\epsilon>0$ such that for infinitely many $A$ we have $\wt P_A(\theta_0 q)\ge \theta_0 \frac{ \wt P_A(\theta_A q)}{\theta_A}> \theta_0 (\frac{ \wt P(\theta_0 q)}{\theta_0}+\epsilon)= \wt P(\theta_0 q)+\theta_0\epsilon$, which is a contradiction. 

Now suppose that there exists a sequence $(A_n)_{n\ge 1}$ such that $\ell=\lim_{n\to\infty} \frac{ \wt P_{A_n}(\theta_{A_n} q)}{\theta_{A_n}} =\frac{ \wt P(\theta_0 q)}{\theta_0}-\epsilon$ with $\epsilon>0$. Clearly $(\theta_{A_n})_{n\ge 1}$ cannot be  ultimately equal to a constant $\theta_0'$, for otherwise this would contradict the definition of $\theta_0$. Suppose that a subsequence $(\theta_{A_{n_j}})_{j\ge 1}$ is increasing and denote its limit by $\theta$. We have $\ell= \lim_{j\to\infty} \frac{ \wt P_{A_{n_j}}(\theta_{A_{n_j}} q)}{\theta_{A_{n_j}}} \le\frac{ \wt P(\theta_0 q)}{\theta_0}-\epsilon\le   \frac{ \wt P(\theta q)}{\theta}-\epsilon$.  Since $\wt P_{A_{n_j}}(\theta q)\nearrow \wt P (\theta q)$, we can fix  $j_0\ge 1$ such that $j\ge j_0$ implies that $\wt   P_{A_{n_j}}(\theta q) \ge (\ell+\epsilon/2)\theta$.  For $j$ large enough, we have $\wt P_{A_{n_j}}(\theta_{A_{n_j}}q)\sim \theta_{A_{n_j}} \ell\sim \theta \ell <\wt P_{A_{n_{j_0}}}(\theta_{A_{n_j}}q)$ since $\wt P_{A_{n_{j_0}}}(\theta_{A_{n_j}}q)$ tends to  $\wt P_{A_{n_{j_0}}}(\theta q)$.  This contradicts the fact that $\wt P_{A_{n_j}}\ge  \wt P_{A_{n_{j_0}}}$. If a subsequence $(\theta_{A_{n_j}})_{j\ge 1}$ is decreasing, a similar arguments yields a new contradiction. In conclusion, we get that $\liminf_{A\to\infty}  \frac{ \wt P_A(\theta_A q)}{\theta_A}\ge \frac{ \wt P(\theta_0 q)}{\theta_0}$, hence $\lim_{A\to\infty}  \frac{ \wt P_A(\theta_A q)}{\theta_A}= \frac{ \wt P(\theta_0 q)}{\theta_0}$. 

Then, since  we already know \eqref{hatPA} and  $\widehat P(q)=\lim_{A\to\infty}\widehat P_A(q)$, we get the desired conclusion for $q\neq 0$. The case $q=0$ is trivial.

\medskip

\noindent(2)$(b)$ Let $q\in\R^d$. The inequality $\widehat P(q)\le \underline P(q)$ will be proved to hold almost surely with $(2)(c)$.  On the other hand, Proposition~\ref{upb1}(1) implies that $\overline P(q)\le \widehat P(q)$ almost surely. 

\medskip

\noindent
(2)$(c)$ 
Let us first assume that  $\E(N\log (N))<\infty$ and $I$ is compact, and directly show (2)$(d)$ in this case. The first assumption ensures that the martingale $\frac{\#\TT_n}{\E(N)^n}$ converges to a positive limit almost surely (\cite{Biggins1}). Then, for $n\ge 1$ define the Borel probability measure on $\R^d$: $\displaystyle 
\nu_n=\frac{1}{\#\TT_n}\sum_{u\in \TT_n}\delta_{S_n(u)/n}.
$
Part (1) of the theorem implies that almost surely, for all $\alpha\in I$ we have 
$$
\lim_{\epsilon\to 0^+}\liminf_{n\to\infty} n^{-1}\log \nu_n(B(\alpha,\epsilon))=\lim_{\epsilon\to 0^+}\limsup_{n\to\infty} n^{-1}\log \nu_n(B(\alpha,\epsilon))=f(\alpha)-\log(\E(N)). 
$$
Now, using that the balls $B(\alpha,\epsilon)$ form a basis of the Euclidean space $\R^d$ topology,  we can apply  \cite[Theorem 4.1.11]{De-Zei} to conclude that the family $\{\nu_n\}_{n\ge 1}$ satisfies the weak large deviations principle with good rate function $\log(\E(N))-f$. However, the fact that $I$ is compact implies that the family $\{\nu_n\}_{n\ge 1}$ is exponentially tight so that the associated strong large deviations holds. Then Varadhan's integral lemma  \cite[Theorem 4.3.1]{De-Zei} directly yields that $\lim_{n\to\infty} P_n=P$ almost surely over $\R^d$. To see that $\{\nu_n\}_{n\ge 1}$ is exponentially tight, let $R>0$ such that $I\subset B(0,R/2)$. Let us show that, with probability 1, for $n$ large enough, $\{u\in \TT_n: n^{-1}S_n(u)\not\in B(0,R)\}=\emptyset$. For each $\alpha\in \partial B(0,R)$, we have $\wt P^*(\alpha)<0$. Let $q_\alpha$ such that $\wt P(q_\alpha)-\langle q_\alpha|\alpha\rangle <0$. By continuity of $\alpha\mapsto P(q_\alpha)-\langle q_\alpha|\alpha\rangle$, there is $0<\epsilon_\alpha<R/4$ such that $P(q_\alpha)-\langle q_\alpha|\alpha'\rangle <0$ for all $\alpha'\in B(\alpha,\epsilon_\alpha)$. By $\partial B(0,R)$ compactness, let $\{\alpha_i\}_{i\in \mathcal I}$ be a finite collection of $\alpha$' in $\partial B(0,R)$ such that $\partial B(0,R)\subset \bigcup_{i\in\mathcal I} B(\alpha_i,\epsilon_i)$, where $\epsilon_i=\epsilon_{\alpha_i}$. Set $q_i=q_{\alpha_i}$, and recall that without loss of generality we can assume that $0\in I$. Since for every $i\in\mathcal I$, we have $0\in I\subset \{\alpha: \wt P(q_i)-\langle q_i|\alpha\rangle \ge 0\}$, we necessarily have $\langle q_i|\lambda \alpha\rangle\ge \langle q_i|\alpha\rangle >0$ for all $\alpha\in B(\alpha_i,\epsilon_i)\cap \partial B(0,R)$ and $\lambda>1$. Let $n\ge 1$. Define $U_{n,i}=\{u\in\TT_n: R^{-1}n^{-1}S_n(u) \in B(\alpha_i,\epsilon_i)\}$. By construction, since due to the previous remark we have $\langle q_i|n^{-1}S_n(u)\rangle \ge s_i=\sup\{\langle q_i|\alpha\rangle: \alpha\in  B(\alpha_i,\epsilon_i)\}$ for all $u\in U_{n,i}$,  we also have $
\# U_{n,i}\le \sum_{u\in U_{n,i}} \exp (\langle q_i|S_n(u)\rangle -ns_i).
$
It follows that 
$$
\E(\#\{u\in\TT_n: n^{-1}S_n(u)\not\in B(0,R)\})\le \E(\sum_{i\in\mathcal I}\# U_{n,i})\le \sum_{i\in\mathcal I} \exp (n(\wt P(q_i)-s_i)).
$$
Since by construction we have $\wt P(q_i)-s_i=\sup\{\wt P(q_i)-\langle q_i|\alpha\rangle: \alpha\in  B(\alpha_i,\epsilon_i)\}<0$ for each $i\in\mathcal I$, we deduce that $\E\Big (\sum_{n\ge 1}\#\{u\in\TT_n: n^{-1}S_n(u)\not\in B(0,R)\}\Big )<\infty$, hence the desired result. 

\medskip

We now consider the general case. We use again the approximation of $\wt P$ by the functions $\wt P_A$, $A\ge A_0$, introduced in Section~\ref{sec2.1}. For each $A\ge A_0$, let $\partial \TT_A$ be the subset of $\partial \TT$ defined as $\{u_1\cdots u_n\cdots \in \N_+^{N_+}: 1\le u_k \le N_{u_1\cdots u_{k-1}}\land A\text{ and } \|X_{u_1\cdots u_k}\|\le A,\, \forall \, k\ge 1\}$. It is  the boundary of the Galton-Watson subtree $\TT_A$  of $\TT$ obtained as $\bigcup_{n\ge 1} \TT_{A,n}$, where $\TT_{A,n}=\{u_1\cdots u_n \in \N_+^{n}: 1\le u_k \le N_{u_1\cdots u_{k-1}}\land A\text{ and } \|X_{u_1\cdots u_k}\|\le A,\, \forall \, 1\le k\le n\}$. Since $N\land A$ is bounded and $I_A$ is compact, conditionally on $\{\partial \TT_A\neq\emptyset\}$, we can apply what preceeds to  $S_{n}(t)$ on $\partial\TT_A$ and get,  with probability 1, conditionally on  $\{\partial \TT_A\neq\emptyset\}$, for all $q\in \R^d$ 
$$
\lim_{n\to\infty}n^{-1}\sum_{u\in \TT_{A,n}}\exp (\langle q|S_n(u)\rangle)=\widehat P_A(q),
$$
where $\widehat P_A$ has been introduced to prove $(2)(a)$. This yields, conditionally on $\{\partial \TT_A\neq\emptyset\}$, $\underline P(q)\ge \widehat P_A(q)$.   Since $\widehat P$ is the limit of the  sequence $(\widehat  P_A)_{A\ge A_0}$ as $A\to\infty$ and  the events $\{\partial \TT_A\neq\emptyset\}$ are non decreasing and converge to a set of probability 1, we conclude that, with probability 1, we have $\underline P\ge \widehat P$ on $\R^d$. 

Now, since for each $q\in\R^d$ we have $\overline P(q)=\widehat P(q)$ almost surely and both $\overline P$ and $\widehat P$ are convex, we conclude that almost surely, for all $q\in\R^d$, $\overline P(q)=\widehat P(q)$ over  $\R^d\setminus \mathrm{rel}\,\partial (\dom \widehat P)$. Moreover, since $\widehat P$ is lower semi-continuous, the equality extends to those points of $\mathrm{rel}\,\partial (\dom \widehat P)$ at which $\overline P$ is lower semi-continuous. 

\medskip

\noindent
(2)$(d)$ If $I$ is compact we have $\dom \widehat P=\R^d$ so the conclusion comes from (2)$(c)$. If $d=1$, $\mathrm{rel}\,\partial (\dom \widehat P)$ consists of at most two points, so the conclusion comes from (2)$(d)$ over $\R\setminus \mathrm{rel}\,\partial (\dom \widehat P)$ and (2)$(b)$ over $ \mathrm{rel}\,\partial (\dom \widehat P)$. 

\subsection{Minimal supporting subtree for the free energy. Proof of Theorem~\ref{thm-1.4}}

Property \eqref{subest} follows from arguments similar to those developed in \cite{Moerters-Ortgiese}. The proof of the other properties is based on two facts: (1) the proof of Theorem~\ref{thm-1.2} provides, with probability~1, for all $\alpha\in I$,  uncountably sequences $\varrho_\alpha$ for which the measure $\mu_{\varrho_\alpha}$ is of Hausdorff dimension $\wt P^*(\alpha)$ and there exists a sequence $(\epsilon_n)_{n\ge 1}$ such that $\mu_{\varrho_\alpha}(\bigcup_{N\ge 1}\bigcap_{n\ge N}E_{X,n}(\alpha))=\|\mu_{\varrho_\alpha}\|$, where $E_{X,n}(\alpha)=\{t\in \partial \TT: |S_nX(t)-n\alpha|\le n\epsilon_n\}$; (2) Theorem~\ref{thm-1.3} on the large deviations properties. Indeed, on a set of probability equal to 1, for all $\alpha\in I$, for each $\varrho_\alpha$ as above, we can fix an integer $N_\alpha$ such that $\mu_{\varrho_\alpha}(\bigcap_{n\ge N_\alpha}E_{X,n}(\alpha))>0$. This implies that there exists $u_\alpha\in \TT_{N_\alpha}$ such that $[u_\alpha]\bigcap \bigcap_{n\ge N_\alpha}E_{X,n}(\alpha))$  has positive $\mu_{\varrho_\alpha}$-measure, hence a Hausdorff dimension larger than or equal to $\wt P^*(\alpha)$. Denoting by $\TT_\alpha$ the associated tree whose generation $n$ is given by $\TT_{\alpha,n}=\{u_{\alpha,1}\cdots u_{\alpha,n}\}$ if $n\le |u_\alpha|$ and  $\TT_{\alpha,n}=\{w\in \TT_n:[w]\cap [u_\alpha]\bigcap \bigcap_{n\ge N_\alpha}E_{X,n}(\alpha)\neq\emptyset\}$ if $n>|u_\alpha|$, we get that $\liminf_{n\to\infty}\frac{1}{n}\log \# \TT_{\alpha,n}\ge \wt P^*(\alpha)$ since $\bigcup_{w\in \TT_{\alpha,n}}[w]$ is an $e^{-n}$ covering of $[u_\alpha]\bigcap \bigcap_{n\ge N_\alpha}E_{X,n}(\alpha)$ for all $n\ge 1$; moreover, by construction, for all $w\in \TT_{\alpha,n}$, one has $|S_nX(w)-n\alpha|\le n\epsilon_n$. On the other hand, the result on the large deviations spectrum implies that $\limsup_{n\to\infty}\frac{1}{n}\log \# \TT_{\alpha,n}\le \wt P^*(\alpha)$. Also, all the trees so built can be chosen pairwise distinct  since the measures $\mu_{\varrho_\alpha}$ can be chosen to be pairwise mutually singular. Then, the claims result from a direct estimate of the partition function restricted to the elements of~$\TT_{\alpha,n}$. 

%%%%%%%%%%%
%%%%%%%%%%%    

\section{The case  $\dom\wt P^*=\emptyset$} \label{domvide}
We do not have to consider Theorem \ref{thm-1.2}(3). We can approximate from below the function $\wt P^*$ by the convex proper functions 
$$
\wt {\wt P}_B(q)=\begin{cases} \wt P_B(0)=\log \E(N\land B)&\text{if } q=0\\
\infty&\text{otherwise}
\end{cases} \quad (B\in\N_+),
$$
for which  $\wt P^*$ equals $\log \E(N\land B)$ everywhere. Then, each function $\wt {\wt P}_B$ is the non decreasing limit of the functions
$$
{\wt P}_{A,B}(q)=\log \E\Big (\sum_{i=1}^{N\land B} \mathbf{1}_{\{\|X_i\|\le A\}} \exp (\langle q|X_i\rangle )\Big )  \quad (A\in\N_+),
$$
which for $A$ and $B$ large enough are strictly convex and analytic over $\R^d$. By using the same strategy as in the previous sections, we can find a family $\Omega_{A,B}$ of subsets of $ \Omega$, as well as random subsets $\partial \wt \TT_{A,B}$ of $\partial \TT$ defined on  the sets $\Omega_{A,B}$,  such the lower bounds for Hausdorff dimensions of the sets $E_X(K)\cap \partial \wt \TT_{A,B}$ for $\mathcal K\ni K\subset \dom \wt {\wt P}_B^*=\R^d$ are given by $\inf_{\alpha\in K} \wt {\wt P}_B^*(\alpha)=\wt {\wt P}_B(0)=\log \E(N\land B)$, and $\P(\Omega_{A,B})$ tends to 1 as $A,B$ tend to $\infty$. Then, since $\wt {\wt P}_B(0)=\log \E(N\land B)$ tends to $\wt P(0)=+\infty$, this yields the results regarding the dimensions of the sets. Also, going to the proof of Theorem~\ref{thm-1.3} in the case $\dom \wt P^*\neq\emptyset$, we see that part (1) of the theorem does not depend on this assumption, and for part (2), mimicking the arguments we see that on subsets $\Omega_B$ of $\Omega$ of probability arbitrarily close to 1, we have $\underline P(q)\ge \wt {\wt P}_B(q)$ for all $q\in\R^d$, hence $\underline P(q)\ge \wt P(q)=\infty$ almost surely over $\R^d$.

\section{Proof of Theorem~\ref{MM} about Mandelbrot measures}\label{MMM} With the notations of the introduction, denote $Z(1,u)=Z(u)$ and $Z(\emptyset)=Z=\|\mu\|$. The result follows from a series of remarks.

\subsection{Upper bounds for Hausdorff and packing dimensions}

For $n\ge 1$ and $q\in \R$ let  
$$
P_{\mu,n}(q)= n^{-1}\log \sum_{u\in\TT_n} \mu ([u])^q=n^{-1}\log \sum_{u\in\TT_n} \exp (q S_nX(u
)) Z(u)^q.$$

The free energy $\limsup_{n\to\infty} P_{\mu,n}$ is also called $L^q$-spectrum of $\mu$. With respect to $P_n(q)$, we replaced $S_nX(u)$ by $S_nX(u)+ \log (Z(u))$. Then, with respect to the  estimates achieved in the proof of Proposition~\ref{upb1},  we first have an additional factor $\E(Z^{\theta q})$ in \eqref{Pntheta}. If $q$ is negative, we know that $\wt P(\theta q)<\infty $ implies that $\E(Z^{\theta q})<\infty$. Hence the estimate $\limsup_{n\to\infty} P_{\mu,n}(q)\le \inf\{\theta^{-1}\wt P(\theta q):0<\theta \le 1\}$ holds almost surely. If $0\le q\le 1$, the same argument works. Suppose that $q>1$. It is not hard to check that  $\inf\{\theta^{-1}\wt P(\theta q):0<\theta \le 1\}$ is reached at $\theta_q=1$ if $q\in J$ or $\theta_q=\sup (J)/q$ otherwise. Under our assumptions, for all $0<\theta<\theta_q$ we have $\theta q\in J$ hence $\E(Z^{\theta q})<\infty$, so that $\limsup_{n\to\infty} P_{\mu,n}(q)\le \theta^{-1}\wt P(\theta q)$ almost surely, and consequently $\limsup_{n\to\infty} P_{\mu,n}(q)\le \theta_0^{-1}\wt P(\theta_0 q)$. It follows that all the estimates of dimensions or large deviations spectrum from above obtained in Sections~\ref{upb} and \ref{uppack} for the sets $E_X(K)$ hold if we replace $S_nX(t)$ by $\log (\mu([t_{|n}]))$. 

\subsection{The estimates of Proposition~\ref{pp3}, lower bounds for dimensions, and $0$-$\infty$ laws when $\log (\mu([t_{|n}]))$ replaces $S_nX(t)$} 

We have seen after the statement of Theorem~\ref{MM} that under our assumptions $Z$ possesses finite moments of negative orders. Since $\E(Z)=1$, it follows that all the moments of positive orders of $|\log (Z)|$ are finite. 

Now, for $j\ge 1$ let $e'_j=\|(|\log (Z)|^3)\|_{\frac{p_j}{p_j-1}}$ and replace the condition \eqref{control2} on the sequence $(L_j)_{j\ge 1}$ by 
\begin{equation*}%\label{control2bis}
\frac{d_j (e_{j}+e'_j) s_{j}}{ 1-\exp( r_{j})}+\frac{d_{j+1} e_{j+1} s_{j+1}}{ 1-\exp( r_{j+1})}\le C_0\exp (\log({L_{j-1}}) \gamma_{j}^2 m_{j})
\end{equation*}
for $j\ge 2$, with  $C_0=\displaystyle \frac{ (e_1+e'_1)s_{1}}{ 1-\exp( r_{1})}+\frac{2 (e_2+e'_2) s_{2}}{1-\exp( r_{2})}$. Then \eqref{control2} is still valid, so that all the estimates achieved in until now about properties related with the measures $\mu_\varrho$ still hold. Moreover, mimicking the estimates leading to \eqref{control3} and \eqref{control33} we can get $C_{\mathcal J}$ such that for all $|u|\ge 1$ 
\begin{eqnarray*}
&&\E(\sup_{\varrho\in\mathcal J}Y(\varrho, u) |\log(Z(u))|^3)\\
&\le& \E(|\log(Z(u))|^3)+\sum_{n\ge 1} \|\sup_{\varrho\in\mathcal J}|Y_n(\varrho,u)-Y _{n-1}(\varrho,u)\|_{p_{j_{|u|+n}}}   \|(|\log (Z)|^3)\|_{\frac{p_{j_{|u|+n}}}{p_{j_{|u|+n}}-1}}\\
&\le& C_{\mathcal J} |u|^{\epsilon_{|u|}}.
\end{eqnarray*}
Then, mimicking the estimates achieved to prove Lemma~\ref{dDimmurho} and Proposition~\ref{pp3}, we can get 
\begin{multline*}
\E\Big (\sum_{n\ge 1}\sup_{\varrho\in\mathcal J}\mu_\varrho\Big( \Big \{t\in\partial\wt\TT: \frac{|\log (Z(t_1\cdots t_n))|}{n}\ge n^{-1/2}\Big \}\Big ) \Big )\\
\le \sum_{n\ge 1} d_{j_n} n^{-3/2} \E(\sup_{\varrho\in\mathcal J}Y(\varrho, u) |\log(Z(u))|^3)\le C_{\mathcal J}  \sum_{n\ge 1} n^{1/5-3/2+\epsilon_n}<\infty.
\end{multline*}
Consequently, with probability 1, conditionally on $\{\partial\wt \TT\neq\emptyset\}$, for all $\varrho\in\mathcal J$, for $\mu_\varrho$-almost every $t$, for $n$ large enough we have $ \frac{|\log (Z(t_1\cdots t_n))|}{n}\le \delta_n =n^{-1/2}$. This is enough to extend Propositions~\ref{pp3} and \ref{conclusion}, as well as Theorem~\ref{thm-1.2}(3).

\subsection{Large deviations and $L^q$-spectrum}
The extension of Theorem~\ref{thm-1.3} with  $\log (\mu([t_{|n}]))$ instead of $S_nX(t)$ is now direct.

\section{Appendix: Hausdorff and packing measures and dimensions}\label{HandP}
Given a subset $K$ of $\N_+^{\N_+}$ endowed with a metric $d$ making it $\sigma$-compact,  $g:\R_+\to\R_+$ a continuous non-decreasing function near 0 and such that $g(0)=0$,  and $E$ a subset of $K$, the Hausdorff measure of $E$ with respect to the gauge function $g$ is defined as 
$$
\mathcal H^g(E)=\lim_{\delta\to 0^+}\inf\Big\{\sum_{i\in\N}g(\mathrm{diam}(U_i))\Big\},
$$
the infimum being taken over all the countable coverings $(U_i)_{i\in\N}$ of $E$ by subsets of $K$ of diameters less than or equal to $\delta$. 

If $s\in\R_+^*$ and $g(u)=u^s$, then $\mathcal H^g(E)$ is also denoted $\mathcal H^s(E)$ and called the $s$-dimensional Hausdorff measure of $E$. Then, the Hausdorff dimension of $E$ is defined as 
$$
\dim E=\sup\{s> 0: \mathcal H^s(E)=\infty\}=\inf\{s> 0: \mathcal H^s(E)=0\},
$$
with the convention $\sup \emptyset= 0$ and $\inf \emptyset=\infty$.  

Packing measures and dimensions are defined as follows. Given $g$ and $E\subset K$ as above, one first defines 
$$
\overline {\mathcal{P}}^g (E)= \lim_{\delta\to 0^+} \sup\Big\{\sum_{i\in\N}g(\mathrm{diam}(B_i))\Big\},
$$
the supremum being taken over all the packings $\{B_i\}_{i\in\N}$ of $E$ by balls centered on $E$ and with diameter smaller than or equal to $\delta$. Then, the packing measure of $E$ with respect to the gauge $g$ is defined as
$$
{\mathcal{P}}^g (E)=\lim_{\delta\to 0^+}\inf\Big\{\sum_{i\in\N}\overline {\mathcal{P}}^g (E_i)\Big\},
$$
the infimum being taken over all the countable coverings $(E_i)_{i\in\N}$ of $E$ by subsets of $K$ of diameters less than or equal to $\delta$. If $s\in\R_+^*$ and $g(u)=u^s$, then $\mathcal P^g(E)$ is also denoted $\mathcal P^s(E)$ and called the $s$-dimensional measure of $E$. Then, the packing  dimension of $E$ is defined as 
$$
\Dim E=\sup\{s> 0: \mathcal P^s(E)=\infty\}=\inf\{s> 0: \mathcal P^s(E)=0\},
$$
with the convention $\sup \emptyset= 0$ and $\inf \emptyset=\infty$.  For more details the reader is referred to \cite{Falconer,Mattila}.

If  $\mu$ is a positive and finite Borel measure supported on $K$, then its lower Hausdorff and packing dimensions is defined as 
\begin{equation*}
\begin{split}
\underline{\dim} (\mu)&=\inf\big \{\dim\, F: \ F\  \mathrm{ Borel},\ \mu(F)>0\big \}\\
\underline{\Dim} (\mu)&=\inf\big \{\Dim\, F: \ F\  \mathrm{ Borel},\ \mu(F)>0\big \},
\end{split}
\end{equation*}
and its upper Hausdorff and packing dimensions are defined as  
\begin{equation*}
\begin{split}
\overline{\dim} (\mu)&=\inf\big \{\dim\, F: \ F\  \mathrm{ Borel},\ \mu(F)=\|\mu\|\big \} \\
\overline{\Dim} (\mu)&=\inf\big \{\Dim\, F: \ F\  \mathrm{ Borel},\ \mu(F)=\|\mu\|\big \}.
\end{split}
\end{equation*}
We have (see~\cite{Cutler,Fan}) 
$$
\underline{\dim} (\mu)={\mathrm{ess}\,\inf}_\mu \ \liminf_{r\to 0^+}\frac{\log \mu(B(t,r))}{\log (r)}, \ \underline{\Dim} (\mu)={\mathrm{ess}\,\inf}_\mu \ \limsup_{r\to 0^+}\frac{\log \mu(B(t,r))}{\log (r)}  
$$
and 
$$
\overline{\dim} (\mu)={\mathrm{ess}\,\sup}_\mu \ \liminf_{r\to 0^+}\frac{\log \mu(B(t,r))}{\log (r)}, \ \overline{\Dim} (\mu)={\mathrm{ess}\,\sup}_\mu \ \limsup_{r\to 0^+}\frac{\log \mu(B(t,r))}{\log (r)},  
$$
where $B(t,r)$ stands for the closed ball of radius $r$ centered at $t$. If $\underline{\dim} (\mu)=\overline{\dim} (\mu)$ (resp. $\underline{\Dim} (\mu)=\overline{\Dim} (\mu)$), this common value is denoted $\dim \mu$ (resp. $\Dim (\mu)$), and if $\dim \mu=\Dim\mu$, one says that $\mu$ is exact dimensional.

%%%%%%%%%%%%%%%%%%%%%%%%%%%%%%%%%%%%%%%%%%%%%%%%%%%%%%%%%%%%%%%%%%%%%%%%%%%%%%%%%%%%%%%%%%%%%%%%%%%%
%%%%%%%%%%%%%%%%%%%%%%%%%%%%%%%%%%%%%%%%%%%%%%%%%%%%%%%%%%%%%%%%%%%%%%%%%%%%%%%%%%%%%%%%%%%%%%%%%%%%
%%%%%%%%%%%%%%%%%%%%%%%%%%%%%%%%%%%%%%%%%%%%%%%%%%%%%%%%%%%%%%%%%%%%%%%%%%%%%%%%%%%%%%%%%%%%%%%%%%%%%%%%%%%%%%%%%%%%%%%%%%%%%%%%%%%%%%%%%%%%%%%%%%%%%%%%%%%%%%%%%%%%%%%%%%%%%%%%%%%%%%%%%%%%%%%%%%%%%%%%%%

\end{document}